\documentclass[10pt,a4paper]{article}
\usepackage[latin1]{inputenc}
\usepackage{authblk}

\usepackage{amsmath}
\usepackage{amsfonts}
\usepackage{amssymb}
\usepackage{hyperref}
\usepackage{graphicx}
\usepackage{amsthm}
\usepackage{xcolor}
\usepackage{xspace}
\usepackage{appendix}
\usepackage{virginialake}
\usepackage{mathtools}
\usepackage{stmaryrd}
\allowdisplaybreaks[1]
\usepackage[shortlabels]{enumitem}  
\usepackage{adjustbox}
\usepackage{nicefrac}
\usepackage{mathrsfs}
\usepackage{calligra}

\usepackage{cleveref}

\newcommand{\HT}{\mathsf{HT}}
\newcommand{\htr}[2]{\HT(#2)(#1)}

\renewcommand{\emptyset}{\varnothing}
\newcommand{\storageone}{}

\renewcommand{\vec}[1]{\mathbf{#1}}
\renewcommand{\epsilon}{\varepsilon}

\newcommand{\blue}[1]{{\color{blue} #1}}
\newcommand{\red}[1]{{\color{red} #1}}
\newcommand{\orange}[1]{{\color{orange} #1}}

\newcommand{\violet}[1]{{\color{violet} #1}}
\newcommand{\magenta}[1]{{\color{magenta} #1}}

\newcommand{\Nat}{\mathbb N}

\newcommand{\defsym}{:=}
\newcommand{\defl}{::=}
\newcommand{\boxP}[1]{[#1]}
\newcommand{\diaP}[1]{\langle#1\rangle}
\newcommand{\comp}{;}
\newcommand{\union}{\cup}

\newcommand{\ST}{\mathsf{ST}}

\newcommand{\st}[2]{\ST(#2)({#1})}

\newcommand{\test}[1]{#1 ?}

\newcommand{\fmi}[1]{\mathit{fm}(#1)}

\newcommand{\plus}[1]{#1^+}

\newcommand{\nf}[1]{\overline{#1}}

\newcommand{\pred}{\mathsf{Pr}}

\newcommand{\rel}{\mathsf{Rel}}

\newcommand{\interp}[2]{#2^{#1}}
\newcommand{\domain}[1]{|#1|}

\newcommand{\TC}{\mathit{TC}}
\newcommand{\RTC}{\mathit{RTC}}
\newcommand{\tc}[3]{\TC ({#1})(#2,#3)}
\newcommand{\rtc}[3]{\RTC ({#1}) (#2,#3)}

\newcommand{\coTC}{\overline\TC}
\newcommand{\coRTC}{\overline\RTC}

\newcommand{\cotc}[3]{\coTC ({#1})(#2,#3)}

\newcommand{\PDL}{\mathrm{PDL}}
\newcommand{\PDLtr}{\PDL^+}

\newcommand{\TCL}{\mathrm{TCL}}

\newcommand{\reset}{\mathsf{id}}
\newcommand{\casezero}[1]{#1_0}
\newcommand{\caseone}[1]{#1_1}

\newcommand{\seqsys}{\mathsf{L}}
\newcommand{\klassische}{\mathsf K}
\newcommand{\lk}{\seqsys \klassische}
\newcommand{\lkeq}{\lk_=}

\newcommand{\pdl}{\mathsf{PD}}
\newcommand{\pdltr}{\pdl^+}

\newcommand{\lpdl}{\seqsys\pdl}
\newcommand{\lpdltr}{\seqsys\pdltr}
\newcommand{\tcl}{\mathsf{TC}}
\newcommand{\rtcl}{\tcl_=}
\newcommand{\ltcl}{ \tcl_G}
\newcommand{\lrtcl}{ \mathsf R\ltcl}

\newcommand{\seqar}{\Rightarrow}

\newcommand{\init}{\mathsf{init}}
\newcommand{\id}{\mathsf{id}}
\newcommand{\wk}{\mathsf{wk}}

\newcommand{\krule}[1]{\mathsf k_{#1}}

\newcommand{\instrule}{\mathsf{inst}}
\newcommand{\str}[2]{\{ #1\}^{#2}}

\newcommand{\bv}{\vec x}
\newcommand{\bx}{\vec x}
\newcommand{\by}{\vec y}

\newcommand{\sq}{\vec{S}}

\newcommand{\infrule}{\mathsf r}
\newcommand{\branch}{\mathcal B}
\newcommand{\htrace}{\mathcal H}
\newcommand{\ced}{S}
\newcommand{\trace}{\tau}

\newcommand{\derivation}[1]{\mathcal{#1}}
\newcommand{\derd}{\derivation D}

\newcommand{\hyper}{\mathsf H}
\newcommand{\htcl}{ \hyper\tcl}
\newcommand{\hrtcl}{ \hyper\rtcl}

\newcommand{\denot}[2]{#2^{#1}}
\newcommand{\M}{\mathcal M}
\newcommand{\A}{\mathcal A}
\newcommand{\An}[1]{\A_{#1}}

\newcommand{\cyclic}{\mathit{cyc}}
\newcommand{\circular}{\mathit{cyc}}
\newcommand{\dercirc}{\vdash_\circular}
\newcommand{\dercyc}{\vdash_\cyclic}
\newcommand{\nwf}{\mathit{nwf}}
\newcommand{\dernwf}{\vdash_\nwf}

\newcommand{\fv}[1]{\mathsf{ann}(#1)}
\newcommand{\qq}{\mathbf{Q}}

\newcommand{\freev}[1]{\mathsf{fv}(#1)}

\newcommand{\card}[1]{|\,#1 \,|}
\newcommand{\hyp}{\mathcal{H}}

\newcommand{\fmd}{\mathcal{M}^\times}

\newcommand{\fbr}{\mathcal{B}^\times}
\newcommand{\fhy}{\mathcal{H}}
\newcommand{\ftr}{\tau^\times}
\newcommand{\fint}{\rho^\times}

\newcommand{\cons}[1]{{#1}}
\newcommand{\deltah}{\delta_\hyp}
\newcommand{\ntuple}{\vec{d}}

\newcommand{\ct}[2]{\mathsf{HT}(#1)(#2)}

\newcommand{\strg}{\ct{\Gamma}{\cons{c}}}
\newcommand{\strd}{\ct{\Delta}{\cons{c}}}

\newcommand{\doubleline}{Pr.~\ref{prop:ht_st_tcl}}

\newcommand{\der}{\mathcal{D}}
\newcommand{\tder}{\ct{\mathcal{D}}{\cons{c}} }

\newcommand{\tcrule}{\TC}
\newcommand{\cotcrule}{\coTC}

\newcommand{\eqrule}{=}
\newcommand{\coeqrule}{\neq}

\newcommand{\diat}[2]{\mathsf{CT}(#1)(#2)}

\newcommand{\bvvar}[1]{\bv_{#1}}

\theoremstyle{plain}
\newtheorem{theorem}{Theorem}[section]
\newtheorem{proposition}[theorem]{Proposition}
\newtheorem{lemma}[theorem]{Lemma}
\newtheorem{example}[theorem]{Example}
\newtheorem{corollary}[theorem]{Corollary}
\theoremstyle{definition}

\newtheorem{definition}[theorem]{Definition}
\newtheorem{remark}[theorem]{Remark}

	\newcommand{\aredf}[2]{\textcolor{black}{\alpha(#1,#2)}}

	\newcommand{\abluef}[2]{\textcolor{black}{a(#1,#2)}}
	\newcommand{\bbluef}[2]{\textcolor{black}{\beta(#1,#2)}}

	\newcommand{\coaredf}[2]{\nf{\alpha}(#1,#2)}

	\newcommand{\bblue}{\textcolor{black}{\beta}}

	\newcommand{\coared}{\nf{\alpha}}

\newcommand{\anupam}[1]{}
\newcommand{\marianna}[1]{}
\newcommand{\todo}[1]{}
\begin{document}
\title{Cyclic Proofs, Hypersequents, and \\ Transitive Closure Logic %\thanks{Supported by organization x.}
}
\author{Anupam Das}
\author{Marianna Girlando}
\affil{University of Birmingham}
% \author{Anupam Das%\inst{1}
% 	\and
% Marianna Girlando%\inst{1}
% }

% \date{}
%
%\authorrunning{A. Das, M. Girlando}
%\institute{University of Birmingham 
%\and
%\\
%\email{\{a.das,m.girlando\}@bham.ac.uk}
%}
%

\maketitle              % typeset the header of the contribution

\begin{abstract}
We propose a cut-free cyclic system for Transitive Closure Logic (TCL) based on a form of \emph{hypersequents}, suitable for automated reasoning via proof search.
We show that previously proposed sequent systems are cut-free incomplete for basic validities from Kleene Algebra (KA) and Propositional Dynamic Logic ($\PDL$), over standard translations. 
%but that 
On the other hand, our system faithfully simulates known cyclic systems for KA and $\PDL$, thereby inheriting their completeness results. 
A peculiarity of our system is its richer correctness criterion, exhibiting `alternating traces' and necessitating a 
more intricate 
% complex 
soundness argument than for traditional cyclic proofs. 
%\keywords{Cyclic proofs \and Transitive Closure Logic \and Hypersequents \and Propositional Dynamic Logic.}
\end{abstract}

%%%%%%%%%%%%%%%%%%%%%%%%%%%%%%%%%%%%%%%%%%%%%%%%%%%%%%%%%%%%%%%%%%%%%%%%%%%
%%%%%%%%%%%%%%%%%%%%%%%%%%%%%%%%%%%%%%%%%%%%%%%%%%%%%%%%%%%%%%%%%%%%%%%%%%%

\section{Introduction}
\label{sec:intro}
%%%%%%%%%%%%%%%%%%%%%%%%%%%%%%%%%%%%%%%%%%%%%%%%%%%%%%%%%%%%%%%%%%%%%%%%%%%
%%%%%%%%%%%%%%%%%%%%%%%%%%%%%%%%%%%%%%%%%%%%%%%%%%%%%%%%%%%%%%%%%%%%%%%%%%%

\emph{Transitive Closure Logic} ($\TCL$) is the extension of first-order logic by an operator computing the transitive closure of definable binary relations.
It has been studied by numerous authors, e.g.\ \cite{Immerman87:languages-capture-ccs,Gurevich1988:logic-challenges-cs,Gradel91:on-tcl}, and in particular has been proposed as a foundation for the mechanisation and automation of mathematics \cite{Avron2003:tcl-mechanization}.

Recently, Cohen and Rowe have proposed \emph{non-wellfounded} and \emph{cyclic} systems for $\TCL$ \cite{CohenRowe18:cyclic-tcl,cohen2020non}.
These systems differ from usual ones by allowing proofs to be infinite (finitely branching) trees, rather than finite ones, under some appropriate global correctness condition (the `progressing criterion').
One particular feature of the cyclic approach to proof theory is the facilitation of automation, since complexity of inductive invariants is effectively traded off for a richer proof structure.
In fact this trade off has recently been made formal, cf.~\cite{BerardiTatsuta17:lkid-neq-clkid,Das20:cyclic-arithmetic},
and has led to successful applications to automated reasoning, e.g.\ \cite{BrotherstonGorogiannisPetersen12:cyclist,BrotherstonDistefanoPetersen11:aut-cyc-ent-sep-log,BrotherstonRowe17:aut-cyc-term-sep-log,BrotherstonTellez18:aut-ver-temp-props,BrotherstonTellez20:aut-ver-temp-props}.

In this work we investigate the capacity of cyclic systems to automate reasoning in $\TCL$.
Our starting point is the demonstration of a key shortfall of Cohen and Rowe's system: its cut-free fragment, here called $\ltcl$, is unable to cyclically prove even standard theorems of relational algebra, e.g.\ $(a\cup b)^* = a^*(ba^*)^*$ and  $(aa \cup aba)^+ \leq a^+((ba^+)^+ \cup a)$) (Thm.~\ref{thm:incompleteness}).
An immediate consequence of this is that cyclic proofs of $\ltcl$ do not enjoy cut-admissibility (Cor.~\ref{lrtcl-cut-not-admissible}).
On the other hand, these (in)equations are theorems of Kleene Algebra (KA) \cite{Kozen91:completeness-ka,Krob91:completeness-ka}, a decidable theory which admits automation-via-proof-search thanks to the recent cyclic system of Das and Pous \cite{das2017cut}.

What is more, $\TCL$ is well-known to interpret Propositional Dynamic Logic ($\PDL$), a modal logic whose modalities are just terms of KA, by a natural extension of the `standard translation' from (multi)modal logic to first-order logic (see, e.g., \cite{BlackburnvanBenthem:modal-logic-semantic-perspective:in-handbook,blackburn2002modal}).
Incompleteness of cyclic-$\ltcl$ for $\PDL$ over this translation is inherited from its incompleteness for KA.
This is in stark contrast to the situation for modal logics without fixed points: the standard translation from $K$ (and, indeed, all logics in the `modal cube') to first-order logic actually \emph{lifts} to cut-free proofs for a wide range of modal logic systems, cf.~\cite{miller2015focused,MarinMillerVolpe16:emulating-modal-proof-systems}.

A closer inspection of the systems for KA and $\PDL$ reveals the stumbling block to any simulation: these systems implicitly conduct a form of `deep inference', by essentially reasoning underneath $\exists $ and $\vlan$.
Inspired by this observation, we propose a form of \emph{hypersequents} for predicate logic, with extra structure admitting the deep reasoning required.
We present the cut-free system $\htcl$ and a novel notion of cyclic proof for these hypersequents.
In particular, the incorporation of some deep inference at the level of the rules necessitates an `alternating' trace condition corresponding to \emph{alternation} in automata theory. 
% \emph{alternation} in the
% While traditional cyclic proof checking is usually decided by reduction to universality of \emph{non-deterministic} $\omega$-automata, for our system \emph{alternating} automata are required due to the underlying deep inference reasoning available. 
% \footnote{This novelty is already witnessed in the complexity of proof checking: while usual cyclic proof checking is decidable by reduction to universality of \emph{nondeterministic} $\omega$-automata, our criterion requires \emph{alternating} automata.}
% At the same time, the cyclic system induced for $\TCL$ thereof requires a rather sophisticated progressing criterion: while usual progressiveness for cyclic proofs is decidable by reduction to universality of nondeterministic $\omega$-automata, our criterion requires \emph{alternating} automata.

% \subsection*{Contributions}
% We propose the system $\htcl$, based on hypersequents, for $\TCL$.
Our first main result is the Soundness Theorem (Thm.~\ref{thm:soundness_tcl}): non-wellfounded proofs of $\htcl$ are sound for \emph{standard semantics}.
The proof is rather more involved than usual soundness arguments in cyclic proof theory, due to the richer structure of hypersequents and the corresponding progress criterion.
Our second main result is the Simulation Theorem (Thm.~\ref{thm:completeness-htcl-pdltr}): $\htcl$ is complete for $\PDL$ over the standard translation, by simulating a cut-free cyclic system for the latter.
This result can be seen as a formal interpretation of cyclic modal proof theory within cyclic predicate proof theory, in the spirit of \cite{miller2015focused,MarinMillerVolpe16:emulating-modal-proof-systems}.

To simplify the exposition, we shall mostly focus on equality-free $\TCL$ and `identity-free' $\PDL$ in this extended abstract, though all our results hold also for the `reflexive' extensions of both logics.
We discuss these extensions in Sec.~\ref{sec:pdl}, and present further insights and conclusions in Sec.~\ref{sec:conclusions}.
%\marianna{Maybe we can put a paragraph with: The paper is structured as follows.. }
%\anupam{i don't think there's space}
%We also include a detailed appendix including full proofs for our results and further remarks and examples.

%%%%%%%%%%%%%%%%%%%%%%%%%%%%%%%%%%%%%%%%%%%%%%%%%%%%%%%%%%%%%%%%%%%%%%%%%%%
%%%%%%%%%%%%%%%%%%%%%%%%%%%%%%%%%%%%%%%%%%%%%%%%%%%%%%%%%%%%%%%%%%%%%%%%%%%
\section{Preliminaries}
\label{sec:prelim}
%%%%%%%%%%%%%%%%%%%%%%%%%%%%%%%%%%%%%%%%%%%%%%%%%%%%%%%%%%%%%%%%%%%%%%%%%%%
%%%%%%%%%%%%%%%%%%%%%%%%%%%%%%%%%%%%%%%%%%%%%%%%%%%%%%%%%%%%%%%%%%%%%%%%%%%

% Since we shall also consider PDL later, it will be convenient for our modal and predicate logics to be based on the same underlying language, which we construe as a first-order vocabulary.
We shall work with a fixed first-order vocabulary consisting of a countable set $\pred$ of unary \emph{predicate} symbols, written $p,q,$ etc., and of 
a countable set $\rel$ of binary \emph{relation} symbols, written $a,b,$ etc.
% In the propositional modal setting predicate symbols are identified with \emph{propositional variables}, while relation symbols index modalities, and are called \emph{(atomic) actions}. The well-formed formulas of the logics we consider will be defined in different ways using the language described above.
% However, 
We build formulas from this language differently in the modal and predicate settings, but
all our formulas may be formally evaluated within \emph{structures}:
\begin{definition}
	[Structures]
	\label{def:semantics_basic}
A \emph{structure} $\mathcal M$ consists of a set $D$, called the \emph{domain} of $\mathcal M$, which we sometimes denote by $\domain \M$; a subset $\interp{\mathcal M} p \subseteq D$ for each $p \in \pred$; and 
a subset $\interp{\mathcal M} a \subseteq D\times D$ for each $a \in \rel$.
\end{definition}
	
%A \emph{structure} $\mathcal M$ is specified by the following data:
%\begin{itemize}
%	\item A set $D$, called the \emph{domain} of $\mathcal M$ (sometimes written $\domain \M$).
%	\item A subset $\interp{\mathcal M} p \subseteq D$ for each $p \in \pred$.
%	\item A subset $\interp{\mathcal M} a \subseteq D\times D$ for each $a \in \rel$.
%\end{itemize}
%\end{definition}
%The well-formed formulas of the logics we consider will be defined in different ways using the language described above.
%However, all formulas can be formally evaluated within structures. % as we shall see.

As above, we shall generally distinguish the words `predicate' (unary) and `relation' (binary). 
We could include further relational symbols too, of higher arity, but choose not to in order to calibrate the semantics of both our modal and predicate settings.

%%%%%%%%%%%%%%%%%%%%%%%%%%%%%%%%%%%%%%%%%%%%%%%%%%%%%%%%%%%%%%%%%%%%%%%%%%%
%%%%%%%%%%%%%%%%%%%%%%%%%%%%%%%%%%%%%%%%%%%%%%%%%%%%%%%%%%%%%%%%%%%%%%%%%%%

%%%%%%%%%%%%%%%%%%%%%%%%%%%%%%%%%%%%%%%%%%%%%%%%%%%%%%%%%%%%%%%%%%%%%%%%%%%
%%%%%%%%%%%%%%%%%%%%%%%%%%%%%%%%%%%%%%%%%%%%%%%%%%%%%%%%%%%%%%%%%%%%%%%%%%%
\subsection{Transitive Closure Logic}
\label{ssec:tcl}
%%%%%%%%%%%%%%%%%%%%%%%%%%%%%%%%%%%%%%%%%%%%%%%%%%%%%%%%%%%%%%%%%%%%%%%%%%%
%%%%%%%%%%%%%%%%%%%%%%%%%%%%%%%%%%%%%%%%%%%%%%%%%%%%%%%%%%%%%%%%%%%%%%%%%%%
% \anupam{why include functions here? they are only needed for hypersequent system?}
% \marianna{yes, I think so, but in this way the language we use is presented all together so I don't mind it}
In addition to the language introduced at the beginning of this section, in the predicate setting we further make use of a countable set of \emph{function} symbols, 
written $f^i,g^j,$ etc., where the superscripts $i,j \in \Nat$ indicate the \emph{arity} of the function symbol and may be omitted when it is not ambiguous. 
Nullary function symbols (aka \emph{constant} symbols), are written $\cons{c}, \cons{d} $ etc. 
We shall also make use of
\emph{variables}, written $x,y,$ etc., typically bound by quantifiers.
\emph{Terms}, written $s,t, $ etc., are generated as usual from variables and function symbols by function application.
% by:
% \[
% s,t \ \defl \ x \mid f^n(t_1, \dots, t_n)
% \]
% Note that this grammar subsumes constant symbols too, in the case when $n=0$.
A term is \emph{closed} if it has no variables.

We consider the usual syntax for first-order logic formulas over our language, with an additional operator for transitive closure (and its dual).
Formally, {$\TCL$ formulas}, written $A,B,$ etc., are generated as follows:
%\footnote{When variables $x,y$ are clear from context, we 
%  write $\tc A s t $ or $\tc{A(x,y)}st$ instead of $\tc{\lambda x,y.A}st$, as an abuse of notation.
%  Similarly for $\coTC$.
%  } 
\[
\begin{array}{rl}
     A,B \ \defl\  &  p(t) \mid \bar p(t) 
    %  \mid s=t \mid s\neq t 
     \mid  a(s,t) \mid \bar a (s,t) \mid (A\vlan B) \mid (A \vlor B) \mid \forall x A \mid \exists x A \mid \\
     \noalign{\smallskip}
     & \tc{\lambda x,y.A} s t \mid \cotc{\lambda x,y. A} s t 
\end{array}
\]
% We shall 
% typically 
% work with \emph{closed} formulas in this work, where no variables occur free. 
% Nonetheless, 
When variables $x,y$ are clear from context, we may write $\tc{A(x,y)}st$ or $\tc A s t $ instead of $\tc{\lambda x,y.A}st$, as an abuse of notation, and similarly for $\coTC$. 
We write $A[t/x]$ for the formula obtained from $A$ by replacing every free occurrence of the variable $x$ by the term $t$.

\anupam{this remark is too early, it should go in next section}
 \begin{remark}
 	[Formula metavariables]
 	\label{same-formula-metavariables}
 	%Note that w
 	We are using the same metavariables $A,B,$ etc.\ to vary over both $\PDLtr$ and $\TCL$ formulas.
 	This should never cause confusion due to the context in which they appear.
 	Moreover, this coincidence is suggestive, since many notions we consider, such as duality and satisfaction, are defined in a way that is compatible with both notions of formula.
 \end{remark}

\begin{definition}
	[Duality]
	\label{dfn:rtcl-duality}
% 	The \emph{complement} $\bar A$ of a formula $A$ is defined by:
For a formula $A$ we define its \emph{complement}, $\bar A$, by:
	\[
	\begin{array}{r@{\ \defsym \ }l}
		\overline{p(t)} & \bar p(t)\\
			\overline {\bar p(t)} & p(t) \\
		\overline{a(s,t)} & \bar a (s,t) \\
			\overline{\bar a (s,t)} & a(s,t)
	\end{array}
	\hspace{0.2cm}
	\begin{array}{r@{\ \defsym \ }l}
% 	\overline {\bar p(t)} & p(t) \\
	     \overline{\forall x A}& \exists x \bar A \\
	\overline{\exists x A}& \forall x \bar A 
% 	\end{array}
% \hspace{0.2cm}
% \begin{array}{r@{\ \defsym \ }l}
\\
	\overline{A \vlan B} & \bar A \vlor \bar B \\
	\overline{A\vlor B} & \bar A \vlan \bar B \\
\end{array}
\hspace{0.2cm}
\begin{array}{r@{\ \defsym \ }l}
	\overline{\tc A s t} & \cotc{\bar A} s t \\
	\overline{\cotc A s t} & \tc{\bar A} s t
\end{array}
	\]
\end{definition}
%As we mentioned earlier in \Cref{same-formula-metavariables}, note that 
% As noted in \Cref{same-formula-metavariables},
% the notion of duality above is consistent with the one introduced for $\PDL$ formulas in \Cref{dfn:pdl-duality}, and so the notation $\bar A$ is unambiguous. 
% Again we may write, e.g., $A\vljm B$ for $\bar A \vlor B$.
% 
% 
% Just like for $\PDLtr$, we may evaluate formulas with respect to a structure, but now we need additional data for interpreting function symbols:
% 
We shall employ standard logical abbreviations, e.g.\ $A\vljm B$ for $\bar A \vlor B$.
We may evaluate formulas with respect to a structure, but we need additional data for interpreting function symbols:

\begin{definition}
	[Interpreting function symbols]
	Let $\mathcal M$ be a structure with domain $D$.
	An \emph{interpretation} is a map $\rho$ that assigns to each function symbol $f^n$ a function $D^n \to D$.
	We may extend any interpretation $\rho$ to an action on (closed) terms by setting recursively $\rho (f(t_1, \dots, t_n)) \defsym \rho(f)(\rho(t_1), \dots, \rho(t_n))$.
% 	
% 	A \emph{model} is a pair $(\M,\rho)$, where $\M$ is a structure and $\rho$ is an interpretation with respect to $\M$.
\end{definition}

% In order to facilitate the definition of satisfaction, namely for the quantifier and transitive closure cases, we shall adopt a standard convention of assuming among our constants arbitrary parameters from the model $\mathcal M$.
% Formally, this means that we construe each $v\in D$ as a constant symbol for which we shall always set $\rho (v) = v$.

We only consider \emph{standard} semantics in this work: $\TC$ (and $\coTC$) is always interpreted as the \emph{real} transitive closure (and its dual) in a structure, rather than being axiomatised by some induction (and coinduction) principle.

In order to facilitate the formal definition of satisfaction, namely for the quantifier and reflexive transitive closure cases, we shall adopt a standard convention of assuming among our constant symbols arbitrary parameters from the model $\mathcal M$.
Formally this means that we construe each $v\in D$ as a constant symbol for which we shall always set $\rho (v) = v$.
\begin{definition}
	[Semantics]
	\label{def:semantics_rtcl}
Given a structure $\M$ with domain $D$ and an interpretation $\rho$, the judgement $\M,\rho \models A$ is defined as follows:
% usual for first-order logic with the following additional clauses for $\TC$ and $\coTC$:
% (atomic and propositional cases are standard):
\begin{itemize}
 	\item $\M, \rho \models p(t)$ if $\rho (t) \in \denot \M p$.
 	\item $\M, \rho \models \bar p(t)$ if $\rho(t) \notin \denot \M p$.
 % 	\item $\M, \rho \models s=t$ if $\rho(s) = \rho (t)$.
 % 	\item $\M, \rho \models s\neq t$ if $\rho(s) \neq \rho(t)$.
 	\item $\M, \rho \models a(s,t)$ if $(\rho(s), \rho(t)) \in \denot \M a$.
 	\item $\M, \rho \models \bar a(s,t)$ if $(\rho(s),\rho(t)) \notin \denot \M a$.
 	\item $\M,\rho \models A\vlan B$ if $\M,\rho \models A$ and $\M,\rho \models B$.
 	\item $\M,\rho \models A\vlor B$ if $\M,\rho \models A$ or $\M,\rho \models B$.
 	\item $\M,\rho \models \forall x A$ if, for every $v \in D$, we have $\M,\rho \models A[v/x]$.
 	\item $\M,\rho \models \exists x A$ if, for some $v \in D$, we have $\M,\rho \models A[d/x]$. 
	\item $\M,\rho \models \tc A s t $ if  there are $v_0, \dots, v_{n+1} \in D$  with $\rho(s) = v_0$, $\rho(t)=v_{n+1} $, such that for every $ i\leq n$ we have $\M,\rho \models A(v_i,v_{i+1})$. 
	\item $\M, \rho \models \cotc A s t$ if for all $v_0, \dots, v_{n+1} \in D$  with $\rho(s)=v_0$ and $ \rho(t)=v_{n+1}$, there is  some $ i\leq n$ such that $\M,\rho \models A(v_i,v_{i+1})$. 
\end{itemize}
If $\M,\rho \models A$ for all $\M$ and $\rho$, we simply write $\models A$.
\end{definition}
% Again, note that we suggestively employ the same notation, $\models A$, to denote validity of $\PDLtr$ and $\TCL$ formulas, since the two notions are consistent with each other.
\begin{remark}
	[$\TC$ and $\coTC$ as least and greatest fixed points]
	\label{rtc-is-mu-cortc-is-nu}
	%Notice that, as expected, we do indeed have that 
	As expected, we have $\M, \rho \not\models \tc A s t $ just if $\M,\rho \models \cotc {\bar A} s t$, and so the two operators are semantically dual. 
	Thus, $\TC$ and $\coTC$ duly correspond to least and greatest fixed points, respectively, satisfying in any model:
	\begin{eqnarray}
	\label{eq:tc_fixpoint}
	\tc Ast & \iff & A(s,t) \vlor \exists x (A(s,x) \vlan \tc Axt ) \label{eq:tc-fixed-point-formula} \\
	\label{eq:cotc_fixpoint}
	\cotc Ast & \iff & A(s,t) \vlan \forall x (A(s,x) \vlor \cotc Axt ) \label{eq:cotc-fixed-point-formula}
	\end{eqnarray}
We have included both operators as primitive 
%it as a native operator 
so that we can reduce negation to atomic formulas, allowing a one-sided formulation of proofs. 

Let us point out that our $\coTC$ operator is not the same as Cohen and Rowe's transitive `co-closure' operator $\TC^{\mathit{op}}$ in \cite{CohenRowe:co-closure}.
As they already note there, $\TC^{\mathit{op}}$ cannot be defined in terms of $\TC$ (using negations), whereas $\coTC$ is the formal De Morgan dual of $\TC$ and, in the presence of negation, are indeed interdefinable, cf.~\Cref{dfn:rtcl-duality}. 
\end{remark}

%%%%%%%%%%%%%%%%%%%%%%%%%%%%%%%%%%%%%%%%%%%%%%%%%%%%%%%%%%%%%%%%%%%%%%%%%%%
%%%%%%%%%%%%%%%%%%%%%%%%%%%%%%%%%%%%%%%%%%%%%%%%%%%%%%%%%%%%%%%%%%%%%%%%%%%

%%%%%%%%%%%%%%%%%%%%%%%%%%%%%%%%%%%%%%%%%%%%%%%%%%%%%%%%%%%%%%%%%%%%%%%%%%%
%%%%%%%%%%%%%%%%%%%%%%%%%%%%%%%%%%%%%%%%%%%%%%%%%%%%%%%%%%%%%%%%%%%%%%%%%%%
\subsection{Cohen-Rowe cyclic system for $\TCL$}
\label{ssec:sequents_tcl}
Cohen and Rowe proposed in \cite{CohenRowe18:cyclic-tcl,cohen2020non} a non-wellfounded sequent system for $\TCL$, which we call $\ltcl$, 
 extending a 
 %usual 
 standard
 sequent calculus $\lkeq$ for first-order logic with equality and substitution by rules for $\TC$ 
% (and $\coRTC$) 
inspired by its characterisation as a least fixed point, 
% (and greatest fixed point, respectively), 
cf.~\eqref{eq:tc-fixed-point-formula}.
% and \eqref{eq:cortc-fixed-point-formula} .

\begin{definition}[System] 
A \emph{sequent}, written $\Gamma,\Delta$ etc., is a set of formulas.
The rules of $\ltcl$ are shown in \Cref{fig:seq-calc-fol}. \anupam{add colour-coding to figure too, to be consistent with remainder of paper}
\marianna{moved footnote on original proof system in terms of RTC in sec Differences }
\marianna{put the rules for TC together with first order rules}
% 
%\footnote{Cohen and Rowe's system is originally called $\lrtcl$, rather using a `reflexive' version $\RTC$ of the $\TC$ operator. However this (and its rules) can be encoded (and simulated) by defining $\rtc{\lambda x,y.A} st \defsym \tc{\lambda x,y(x=y \vlor A)}st$.}. 
$\ltcl$-\emph{preproofs} are 
	% defined similarly to $\lpdl$-preproofs: they are 
	possibly infinite trees of sequents generated by the rules of $\ltcl$.
% 	Again, 
	A preproof is \emph{regular} if it has only finitely many distinct sub-preproofs.
	\marianna{moved footnote on subst at the end of next paragraph}
	%\footnote{Note that the presence of the substitution rule is critical for the notion of regularity in predicate cyclic proof theory.}
\end{definition}

%$\lkeq$ is a standard Tait-style one-sided sequent calculus for first-order predicate logic with equality, given in Fig.~\ref{fig:seq-calc-fol}. 
In \Cref{fig:seq-calc-fol} $\sigma$ is a map (``substitution'') from constants to terms and other function symbols to function symbols of the same arity, extended to terms, formulas and sequents in the natural way.
The substitution rule is redundant for usual provability, but 
facilitates the definition of `regularity' in predicate cyclic proof theory.
%facilitates the definition of cyclic proofs. 

The notions of non-wellfounded and cyclic proofs for $\ltcl$ are formulated similarly to those for first-order logic with (ordinary) inductive definitions \cite{brotherston2011sequent}:

\begin{definition}
	[Traces and proofs]
	Given a $\ltcl$ preproof $\derd$ and a branch $\branch = (\infrule_i)_{i \in \omega}$ of inference steps, a \emph{trace} is a sequence of formulas of the form $(\cotc A {s_i}{t_i})_{i\geq k}$ such that for all $i\geq k$ either:
	\begin{itemize}
		\item $\infrule_i$ is not a substitution step and $(s_{i+1}, t_{i+1}) = (s_i,t_i)$; or,
		\item $\infrule_i$ is a $\coTC$ step with principal formula $\cotc A {s_i}{t_i}$ and $(s_{i+1}, t_{i+1}) = (\cons{c},t_i)$, where $\cons{c}$ is the eigenvariable of $\infrule_i$; or, 
		\item $\infrule_i$ is a substitution step with respect to $\sigma$ and $(\sigma (s_{i+1}), \sigma(t_{i+1}) ) = (s_i,t_i)$.
	\end{itemize}
	We say that the trace is \emph{progressing} if the second case above happens infinitely often.
	A $\ltcl$-preproof $\derd$ is a \emph{proof} if each of its infinite branches has a progressing trace.
	If $\derd$ is regular we call it a \emph{cyclic proof}.
	As in the main text, we write $\ltcl \dercyc A$ if there is a cyclic proof in $\ltcl$ of $A$.
\end{definition}

\begin{figure}[t]
	\[
	\vlinf{\id}{}{\Gamma, p(t), \bar p(t)}{}
	\quad 
	\vlinf{\vlor_i}{\text{\footnotesize $i \in \{0,1\}$ }}{\Gamma, A_0\vlor A_1 }{\Gamma, A_i}
	\quad
	\vliinf{\vlan}{}{\Gamma, A\vlan B}{\Gamma, A}{\Gamma, B}
	\quad
	\vlinf{\exists}{}{\Gamma, \exists x A(x)}{\Gamma, A(t)}
	\quad
	\vlinf{\forall}{\text{\footnotesize $\cons{c}$ fresh}}{\Gamma, \forall x A}{\Gamma, A[\cons{c}/x]}
	\]
	\[
	\vlinf{\wk}{}{\Gamma,\Gamma'}{\Gamma}
	\quad
	\vlinf{=}{}{\Gamma, t=t}{}
	\quad
	\vlinf{\neq}{}{\Gamma, s\neq t, A(s)}{\Gamma, A(t)}
	\quad
	\vlinf{\neq}{}{\Gamma, s\neq t, A(t)}{\Gamma, A(s)}
	\quad
	\vlinf{\sigma}{}{\sigma(\Gamma)}{\Gamma}
	\]
	\[
	\vlinf{\TC_0}{}{\Gamma, \tc A s t }{\Gamma, A(s,t)} 
	\qquad
	\vliinf{\TC_1}{}{\Gamma, \tc A s t}{\Gamma, A(s,r)}{\Gamma, \tc A r t}
	\]
	\[
	\vliinf{\coTC}{\text{\footnotesize $\cons{c}$ fresh}}{\Gamma ,\cotc A s t}{\Gamma , A(s,t)}{\Gamma , A(s,\cons{c}), \cotc A {\cons{c}} t}
	\]
	\caption{Sequent calculus $\ltcl$. The first two lines of the Figure contain the rules of the Tait-style sequent system $\lkeq$ for first-order predicate logic with equality. 
		% \anupam{do we need both $\neq$ rules?}
		%     \marianna{I don't think so, if we state somewhere that $ s \neq t $ is equivalent to $ t \neq s $}
		% \anupam{I think that is derivable, yes, but recall that we do not have cut to utilise that fact. } \anupam{i think it's okay to have both rules, it seems to be standard.}
	}
	\label{fig:seq-calc-fol}
\end{figure}

\begin{proposition}
	[Soundness, \cite{CohenRowe18:cyclic-tcl,cohen2020non}]
	\label{lrtcl-cyc-sound}
	If $\ltcl \dercyc A$ then $ \models A$.
\end{proposition}
In fact, this result is subsumed by our main soundness result for $\htcl$ (Thm.~\ref{thm:soundness_tcl}) and its simulation of $\ltcl$ (Thm.~\ref{simulation-of-cohen-rowe}). 
% proofs (\Cref{sec:simulation-cohen-rowe}, \Cref{simulation-of-cohen-rowe}).
% 
% 
% 
%In the presence of cut, we also have a form of converse of the above result in that it is `Henkin complete', i.e.\ complete for all models of a particular axiomatisation of $\TCL$ based on (co)induction principles.
%
In the presence of cut, a form of converse of Prop.~\ref{lrtcl-cyc-sound} holds: cyclic $\ltcl$ proofs 
are `Henkin complete', i.e.\ complete for all models of a particular axiomatisation of $\TCL$ based on (co)induction principles \cite{CohenRowe18:cyclic-tcl,cohen2020non}.
However, the counterexample we present in the next section implies that cut is not eliminable (Cor.~\ref{lrtcl-cut-not-admissible}). 

\subsection{Differences to \cite{CohenRowe18:cyclic-tcl,cohen2020non}}

Our formulation of $\ltcl$ differs slightly from the original presentation in \cite{CohenRowe18:cyclic-tcl,cohen2020non}, but in no essential way.
Nonetheless, let us survey these differences now.

\subsubsection{One-sided vs.\ two-sided.}
Cohen and Rowe employ a two-sided calculus as opposed to our one-sided one, but the difference is purely cosmetic.
Sequents in their calculus are written $A_1, \dots, A_m \seqar B_1, \dots, B_n$, which may be duly interpreted in our calculus as $\bar A_1, \dots, \bar A_m, B_1, \dots, B_n$.
Indeed we may write sequents in this two-sided notation at times in order to facilitate the reading of a sequent and to distinguish left and right formulas.
For this reason, Cohen and Rowe do not include a $\coTC$ operator in their calculus, but are able to recover it thanks to a formal negation symbol, cf.~Dfn.~\ref{dfn:rtcl-duality}.

\subsubsection{$\TC$ vs.\ $\RTC$.}	
Cohen and Rowe's system is originally called $\lrtcl$, rather using a `reflexive' version $\RTC$ of the $\TC$ operator. %However this (and its rules) can be encoded (and simulated) by defining $\rtc{\lambda x,y.A} st \defsym \tc{\lambda x,y(x=y \vlor A)}st$.
%Cohen and Rowe use a $\RTC$ operator instead of $\TC$ but, as 
As they mention, this makes no difference in the presence of equality.
Semantically we have $\rtc Ast \iff s=t \vlor \tc Ast$, but this encoding does not lift to proofs, i.e.\ the $\RTC$ rules of \cite{CohenRowe18:cyclic-tcl} are not locally derived in $\ltcl$ modulo this encoding.
However, the encoding $\rtc A st \defsym \tc{(x=y \vlor A)}st$ suffices for this purpose.
% For instance here is the simulation of Cohen and Rowe's $\RTC$ rules:
% \[
% \scriptsize
% \begin{array}{r@{\quad \leadsto\quad }l}
% \vlinf{\casezero{\RTC}}{}{\Gamma, \rtc Ast}{\Gamma s=t}
% &
% \vlderivation{
% 	\vlin{\casezero \TC}{}{\Gamma, \tc{x=y \vlor A}st}{
% 		\vlin{\casezero \vlor}{}{\Gamma, s=t \vlor A(s,t)}{
% 			\vlhy{\Gamma, s=t}
% 		}
% 	}
% }
% \\
% \noalign{\medskip}
% \vliinf{\caseone\RTC}{}{\Gamma, \rtc Ast}{\Gamma, A(s,r)}{\Gamma, \rtc Art}
% &
% \vlderivation{
% 	\vliin{\caseone\TC}{}{\Gamma, \tc{x=y \vlor A}st}{
% 		\vlin{\caseone\vlor}{}{\Gamma, s=r \vlor A(s,r)}{\vlhy{\Gamma,A(s,r)}}
% 	}{
% 		\vlhy{\Gamma, \tc{x=y\vlor A}rt}
% 	}
% }
% \\
% \noalign{\bigskip}
% \vliinf{\coRTC}{}{\Gamma, \cortc Ast}{\Gamma, s\neq t}{\Gamma, A(s,c),\cortc Act}
% &
% \end{array}
% \]
% \[
% \scriptsize
% \vlderivation{
% 	\vliin{\coTC}{\bullet}{\Gamma, \cotc{x\neq y \vlan A}st}{
% 		\vliin{\vlan}{}{\Gamma, s\neq t \vlan A(s,t)}{
% 			\vlhy{\Gamma, s\neq t}
% 		}{
% 			\vliq{}{}{\Gamma, A(s,t)}{\vlhy{\text{Proj.}}}
% 		}
% 	}{
% 		\vliin{\vlan}{}{\Gamma, s\neq c \vlan A(s,c),\cotc{x\neq y \vlan A}ct}{
% 			\vlin{\neq}{}{\Gamma, s\neq c, \cotc{x\neq y \vlan A}ct}{
% 				\vlin{\coTC}{\bullet}{\Gamma, \cotc{x\neq y \vlan A}st}{\vlhy{\vdots}}
% 			}
% 		}{
% 			\vlhy{\Gamma, A(s,c), \cotc{x\neq y \vlan A}ct}
% 		}
% 	}
% }
% \]
% where the derivation marked Proj.\ is obtained using a standard projection lemma: if there is cyclic proof of $ \Gamma, \cortc Ast$ then there is a \emph{smaller} cyclic proof of $\Gamma, A(s,t)$.\anupam{could say more here}
\anupam{commented translation of rules above.}
\subsubsection{Alternative rules and fixed point characterisations.}
Cohen and Rowe 
%actually 
use a slightly different fixed point formula to induce rules for $\RTC$ and $\coRTC$ (i.e.\ $\RTC$ on the left) based on the fixed point characterisation,
\begin{equation}
\rtc Ast \ \iff \ s=t \vlor \exists x (\rtc Asx \vlan A(x,t))
\label{eq:rtc-fixed-point-cohen-rowe}
\end{equation}
% The corresponding three rules generated by this characterisation replace rules $ \caseone \RTC $ and $ \coRTC $ of \Cref{fig:seq-calc-fol} by:
% %the second and third rules of \eqref{eq:rtc-rules} by:
% 
% \begin{equation}
% \scriptsize
% \vliinf{\caseone \RTC}{}{\Gamma, \Delta,\rtc Ast}{\Gamma,\rtc Asr}{\Delta, A(r,t)}
% \qquad
% \vliinf{\coRTC}{\text{\scriptsize $\cons{c}$ fresh}}{\Gamma, \Delta, \cortc Ast}{\Gamma, s\neq t}{\Delta, \cortc Asc,A(\cons{c},t)}
% \label{eq:rtc-rules-cohen-rowe}
% \end{equation}
% 
% % 	Finally one may consider a least fixed point encoding that 
% % 	%somewhat `subsumes' 
% % 	subsumes
% % 	both of the ones above,
% % 	\begin{equation}
% % 	  \label{eq:rtc-fixed-point-combined}
% %   \rtc A x y \ \iff \ x=y \vlor \exists z,z' (\rtc Axz \vlan A(z,z') \vlan \rtc A{z'}y)
% % \end{equation}
% % yielding the rules:
% % \begin{equation}
% %   \label{eq:rtc-rules-combined}
% %   \begin{array}{c}
% % \vliiinf{\caseone\RTC}{}{\Gamma,\Delta,\Sigma \rtc A s t}{\Gamma, \rtc A s r}{\Delta, A(r,r')}{\Sigma , \rtc A{r'} t}
% % \\
% % \vliinf{\coRTC}{\text{\footnotesize $\cons{c},\cons{c}'$ fresh}}{\Gamma, \Delta, \cortc A s t }{\Gamma, s\neq t}{\Delta, \cortc A s {\cons{c}}, A(\cons{c},\cons{c}'), \cortc A {\cons{c}'} t}
% % \end{array}
% % \end{equation}
% % %
\noindent
decomposing paths `from the right' rather than the left. 
\anupam{commented more details above}
These alternative rules induce analogous notions of trace and progress for preproofs such that progressing preproofs enjoy a similar soundness theorem, cf.~\Cref{lrtcl-cyc-sound}.
The reason we employ a slight variation of Cohen and Rowe's system is
% but we take a slight alternative in order 
to remain consistent with how the rules of $\lpdltr$ and $\htcl$ are devised later.
% 	Let us point out that we do not aim to explicitly work with sequent systems for $\lrtcl$, so the choice of rules is rather of expositional benefit.
To the extent that we prove things \emph{about} $\ltcl$, namely its (cut-free) regular \emph{in}completeness in \Cref{thm:incompleteness}, the particular choice of rules turns out to be unimportant. 
The counterexample we present there is robust: 
it applies to systems with any (and indeed all) of the above rules.
% 	\footnote{This is one reason why we have taken the time to present the above alternative rules, which shall reference again in \Cref{sec:adapting-counterexample-tcl} when demonstrating the adaptability of our counterexample.}
% 	Further discussion about alternative rules is given in \Cref{alternative-rules-for-rtc}, and we explain how the unprovability argument of \Cref{sec:counterexample} can be adapted to systems of alternative rules in \Cref{adapting-counterexample-alternative-rules}.

\section{Interlude: motivation from PDL and Kleene Algebra}
\label{sec:interlude}
Given the $\TCL$ sequent system proposed by Cohen and Rowe, why do we propose a hypersequential system? 
Our main argument is that proof search in $\ltcl$ is rather weak, to the extent that cut-free cyclic proofs are unable to simulate a basic (cut-free) system for modal logic $\PDL$ (regardless of proof search strategy).
At least one motivation here is to `lift' the \emph{standard translation} from cut-free cyclic proofs for $\PDL$ to cut-free cyclic proofs in an adequate system for $\TCL$.

\subsection{Identity-free PDL}
\label{ssec:pdltr}
%%%%%%%%%%%%%%%%%%%%%%%%%%%%%%%%%%%%%%%%%%%%%%%%%%%%%%%%%%%%%%%%%%%%%%%%%%%
%%%%%%%%%%%%%%%%%%%%%%%%%%%%%%%%%%%%%%%%%%%%%%%%%%%%%%%%%%%%%%%%%%%%%%%%%%%
\emph{Identity-free propositional dynamic logic} ($\PDLtr$) is a version of the modal logic $\PDL$ without tests or identity, thereby admitting an `equality-free' standard translation into predicate logic.
Formally, \emph{$\PDLtr$ formulas}, written $ A, B, $ etc., and \emph{programs}, written $ \alpha, \beta, $ etc., are generated by the following grammars:
\begin{eqnarray*}
	A,B &\defl &  p \mid \nf{p} \mid A \vlan B \mid A \vlor B \mid  \boxP{\alpha} A \mid  \diaP{\alpha}A\\
	\alpha, \beta &\defl & a \mid \alpha \comp \beta \mid \alpha \union \beta 
	\mid \alpha^+ 
\end{eqnarray*}
We sometimes simply write $\alpha\beta$ instead of $\alpha;\beta$,
% We sometimes write
and
$(\alpha)A$ for a formula that is either $\diaP \alpha A$ or $\boxP \alpha A$.

\begin{definition}
	[Duality]
	\label{dfn:pdl-duality}
	For a formula $A$ we define its \emph{complement}, $\bar A$, by:% follows:
	\[
	\bar {\bar p} \ \defsym \ p
	\qquad
	\begin{array}{r@{\ \defsym \ }l}
% 		\bar {\bar p} & p \\
		\overline{A \vlan B} & \bar A \vlor \bar B \\
		\overline{A \vlor B} & \bar A \vlan \bar B
	\end{array}
	\qquad
	\begin{array}{r@{\ \defsym \ }l}
	\overline{\boxP \alpha A} & \diaP\alpha \bar A \\
	\overline{\diaP \alpha A} & \boxP\alpha \bar A
			\end{array}
	\]
\end{definition}
% Under duality, we sometimes use further propositional connectives for their usual abbreviations, e.g.\ we write $A\vljm B$ for $\bar A \vlor B$.

We % may 
\emph{evaluate} $\PDLtr$ formulas using the traditional relational semantics of modal logic, by associating each program with a binary relation in a structure. % (see App.~\ref{ssec:pdltr-appendix} for further details).
Again, we only consider {standard} semantics:

% \footnote{We only consider \emph{standard} semantics in this work, in the sense that the $ + $ operator, is interpreted as the true transitive closure, rather than being axiomatised by some induction principle. }
%In fact, since our syntax for formulas and programs are mutually defined, evaluation of $ \pdl $ formulas is also defined by mutually with the interpretation of programs as binary relations.

\begin{definition}
	[Semantics]
	\label{def:semantics_pdl}
	For structures $\M$ with domain $D$, elements $v \in D$, programs $\alpha$ and formulas $A$,
	we define $\interp \M \alpha \subseteq D\times D$ as follows:
\begin{itemize}
	\item ($\interp {\mathcal M} a$ is already given in the specification of $\mathcal M$, cf.~Dfn.~\ref{def:semantics_basic}).
	\item $\interp{\mathcal M}{(\alpha\comp \beta)} := \{ (u,v) : \text{there is } w\in D \text{ s.t.\ }  (u,w) \in \interp {\mathcal M} \alpha \text{ and } (w,v) \in \interp{\mathcal M} \beta   \}$.
	\item $\interp{\mathcal M}{(\alpha \union \beta)} := \{ (u,v) : (u,v) \in \interp{\mathcal M} \alpha \text{ or } (u,v) \in \interp{\mathcal M} \beta \} $.
	\item $\interp{\mathcal M}{(\alpha^+)} := \{ (u,v) : \text{there are } w_0, \dots , w_{n+1} \in D
		\text{ s.t.\ } u = w_0, v = w_{n+1}
		\text{ and  } \text{for } $ $ \text{every } i\leq n  (w_i,w_{i+1}) \in \interp{\mathcal M} \alpha \}$.\anupam{hack here with line break and braces}
\end{itemize}
 and:
 \begin{itemize}
 	\item $\mathcal M,v \models p$ if $v \in \interp{\mathcal M} p$.
 	\item $\mathcal M,v \models \nf p$ if $v \notin \interp{\mathcal M} p$.
 	\item $\mathcal M , v \models A \vlan B$ if $\mathcal M , v \models A $ and $\mathcal M , v \models B $.
 	\item $\mathcal M , v \models A \vlor B$ if $\mathcal M , v \models A $ or $\mathcal M , v \models  B$.
 	\item $\mathcal M , v \models \boxP{\alpha} A$ if $\forall\, (v,w) \in \interp{\mathcal M} \alpha$ we have $\mathcal M, w \models A$.
 	\item $\mathcal M, v \models \diaP{\alpha } A$ if $\exists\, (v,w) \in \interp{\mathcal M} \alpha$ with $\mathcal M,w \models A$.
 \end{itemize}
If $\M,v \models A$ for all $\M$ and $v\in D$, then we write $\models A$.
% (or $ \PDL\models A $). 
\end{definition}
%\anupam{for conference version can omit semantics if need be by reduction to standard translation}

Notw that we are overloading the satisfaction symbol $\models$ here, for both $\PDLtr$ and $\TCL$. This should never cause confusion, in particular since the two notions of satisfaction are `compatible', given that we employ the same underlying language and structures.
In fact such overloading is convenient for relating the two logics, as we shall now see.

\subsection{The standard translation}
\label{ssec:standard_tr}
%%%%%%%%%%%%%%%%%%%%%%%%%%%%%%%%%%%%%%%%%%%%%%%%%%%%%%%%%%%%%%%%%%%%%%%%%%%
%%%%%%%%%%%%%%%%%%%%%%%%%%%%%%%%%%%%%%%%%%%%%%%%%%%%%%%%%%%%%%%%%%%%%%%%%%%

The so-called ``standard translation'' of modal logic into predicate logic is induced by reading the semantics of modal logic as first-order formulas.
We now give a natural extension of this that interprets $\PDLtr$ into $\TCL$.
At the logical level our translation coincides with the usual one for basic modal logic; our translation of programs, as expected, requires the $\TC$ operator to interpret the $+$ of $\PDLtr$.
\begin{definition}
    %[Standard translation]
    \label{dfn:stand-trans-pdl-rtcl}
    % For $\PDLtr$ formulas $A, B$ and programs $\alpha, \beta$, 
    \marianna{added B and $\beta$}\anupam{I don't understand why? We are defining $\st x A$ and $\st{x,y}\alpha$, but not $\st x B$ or $\st{x,y}\beta$. I have commented for now}
    For a $\PDLtr$ formula $A$ and program $\alpha$, 
    we define the \emph{standard translations} $\st x A$ and $\st{x,y}\alpha$ as $\TCL$-formulas with free variables $x$ and $x,y$, resp., mutually inductively as follows,
    
    \medskip
    
    \begin{adjustbox}{max width = \textwidth}
    \begin{tabular}{r @{\ $\defsym$ \ } l @{\hspace{-0.3cm}} r @{\ $\defsym$ \ } l}
    $ \st x p  $& $p(x)$ & $ \st{x,y} a$ & $a(x,y)$ \\[0.1cm]
    $ \st x {\bar p}$ &$ \bar p (x)$ & $\st{x,y} {\alpha \union \beta }$  & $\st{x,y} \alpha \vlor \st{x,y} \beta $ \\[0.1cm]
    $ \st x {A \vlor B} $& $\st x A \vlor \st x B$ &  $\st{x,y} {\alpha\comp  \beta }$ & $\exists z (\st{x,z}\alpha \vlan \st{z,y}\beta) $\\[0.1cm]
    $ \st x {A \vlan B}$ & $\st x A \vlan \st x B$ & $\st{x,y}{\alpha^+}$ & $\tc{\ST(\alpha)} x y  $\\[0.1cm]
    $ \st x {\diaP{\alpha}A}$ & \multicolumn{3}{l}{$ \exists y (\st{x,y} \alpha \vlan \st y A)$  } \\[0.1cm]
    $\st x {\boxP \alpha A}$ & $\forall y (\overline{\st{x,y}\alpha} \vlor \st y A)$& \multicolumn{2}{l}{} \\
    \end{tabular}
    \end{adjustbox}
    
    \medskip
    
    \noindent
    where we have written simply $\TC(\ST(\alpha))$ instead of $\TC(\lambda x,y.\st{x,y}\alpha)$.
    % \[
    % \scriptsize
    % \begin{array}{l@{\ \defsym \ }l}
    %     \st x p & p(x) \\[0.1cm]
    %     \st x {\bar p} & \bar p (x) \\[0.1cm]
    %     \st x {A \vlor B} & \st x A \vlor \st x B \\[0.1cm]
    %     \st x {A \vlan B} & \st x A \vlan \st x B \\[0.1cm]
    %     \st x {\diaP{\alpha}A} & \exists y (\st{x,y} \alpha \vlan \st y A) \\[0.1cm]
    %     \st x {\boxP \alpha A} & \forall y (\overline{\st{x,y}\alpha} \vlor \st y A)
    % \end{array}
    % \quad
    % \begin{array}{l@{\ \defsym \ }l}
    %     \st{x,y} a & a(x,y)\\[0.1cm]
    %     \st{x,y} {\alpha \union \beta } & \st{x,y} \alpha \vlor \st{x,y} \beta \\[0.1cm]
    %     \st{x,y} {\alpha\comp  \beta } & \exists z (\st{x,z}\alpha \vlan \st{z,y}\beta) \\[0.1cm]
    %     \st{x,y}{\alpha^+} & \tc{\ST(\alpha)} x y 
    %   % \st{x,y}{A?} & x=y \vlan \st x A
    % \end{array}
    % \]
\end{definition}

It is routine to show that $\overline{\st x A} = \st x {\bar A}$, by structural induction on $A$, justifying our overloading of the notation $\bar A$, in both $\TCL$ and $\PDLtr$.
\anupam{could give details in appendix}
% \begin{remark}
% %	[Duality commutes with standard translation]
% 	\label{duality-commutes-with-st}
% %Note that we have that 
% It holds that 
% $\overline{\st x A} = \st x {\bar A}$, by %a routine 
% structural induction on $A$.
% \end{remark}
% 
% \begin{remark}
% %[Abbreviations under extensionality]
% \label{extensionality}
%   In the case of translating the $+$ operator, we have simply written $\TC(\ST(\alpha))$ instead of, say, $\TC(\lambda x,y.\st{x,y}\alpha)$. Note that such notation is suggestive due to the extensionality principle of type theory, namely $(\lambda x A) (x) = A$.
%   We shall make similar such abbreviations throughout, in the interest of lightening notation.
% \end{remark}
Yet another advantage of using the same underlying language for both the modal and predicate settings is that we can state the following (expected) result without the need for encodings, following by a routine structural induction (see, e.g., \cite{blackburn2002modal}):
\begin{theorem}
% [Modal correspondence for $ \PDL $]
	\label{thm:standard_t_pdl}
	For $\PDLtr$ formulas $ A $, we have $\mathcal M,v \models A$ iff $\mathcal M \models \st v A$.\anupam{could give proof in appendix}
\end{theorem}
% The proof of this follows by a standard structural induction on $A$.
% See \cite{blackburn2002modal} for a similar argument.

\subsection{Cohen-Rowe system is not complete for $\PDLtr$}
\label{ssec:counterexample}

$\PDLtr$ admits a standard cut-free cyclic proof system $\lpdltr$ (see Sec.~\ref{ssec:sequents_pdltr}) which is both sound and complete (cf.~Thm.~\ref{pdl-soundness-completeness}).
However, a shortfall of $\ltcl$ is that it is unable to cut-free simulate $\lpdltr$.
In fact, we can say something stronger:
%\marianna{refer to the appendix for the proof}
\begin{theorem}
    [Incompleteness]
    \label{thm:incompleteness}
    There exist a $\PDLtr$ formula $A$ such that $\models A$ but $\ltcl \not \dercyc \st x A$ (in the absence of cut).
\end{theorem}

%This means that, not only is the calculus $\ltcl$ unable to locally cut-free simulate the rules of $\lpdltr$, there are some validities for which there are no cut-free cyclic proofs at all in $\ltcl$.
% 
This means not only that $\ltcl$ is unable to locally cut-free simulate the rules of $\lpdltr$, but also that there are some validities for which there are no cut-free cyclic proofs at all in $\ltcl$.
One example of such a formula is:
\begin{equation}
    \label{eq:counterexample-formula-pdltr}
    \diaP{(aa\cup aba)^+} p \vljm \diaP{a^+((ba^+)^+ \cup a)}p
\end{equation}
%A detailed proof of this is given in App.~\ref{sec:counterexample}, but let us 
%devote some discussion to this now.
%shortly discuss it here. 
%First, the formula above is not artificial:
This formula is derived from the well-known $\PDL$ validity $\diaP{(a\cup b)^*}p \vljm \diaP{a^*(ba^*)^*}p$
by identity-elimination.
% \[
% \diaP{(a\cup b)^*}p \vljm \diaP{a^*(ba^*)^*}p
% \]
This in turn is essentially a theorem of relational algebra, namely $(a\cup b)^* \leq  a^*(ba^*)^*$, which is often used to eliminate $\cup$ in (sums of) regular expressions.
The same equation was (one of those) used by Das and Pous in \cite{das2017cut} to show that the sequent system $\mathsf{LKA}$ for Kleene Algebra is cut-free cyclic incomplete.

\marianna{Suggestion: ``consequent'' instead of ``$RHS$'' and ``antecedent'' instead of ``LHS''} 
%Our argument for incompleteness in App.~\ref{sec:counterexample} 
In the remainder of this Section, we shall give a proof of Thm.~\ref{thm:incompleteness}. The argument is 
 much more involved than the one from \cite{das2017cut}, 
due to the fact we are working in predicate logic, but the underlying basic idea is similar.
At a very high level, 
the RHS 
of \eqref{eq:counterexample-formula-pdltr} (viewed as a relational inequality) is translated to an existential formula  $\exists z(\st{x,z}{a^+} \vlan \st {z,y}{(ba^+)^+\cup a}$ that, along some branch (namely the one that always chooses $aa$ when decomposing the LHS of \eqref{eq:counterexample-formula-pdltr}) can never be instantiated while remaining valid.
This branch witnesses the non-regularity of any proof.
%%%%%%%%%%%%%%%%%%%%%%%%%%%%%%%%%%%%%%%%%%%%%%%%%%%%%%%%%%%%%%%%%%%%%%%%%%%%%%%%%%%%%%%%

%Here we give a more detailed argument for Thm.~\ref{thm:incompleteness}.

\subsubsection{Some closure properties for cyclic proofs.}
Demonstrating that certain formulas do \emph{not} have (cut-free) cyclic proofs is a delicate task, made more so by the lack of a suitable model-theoretic account (indeed, cf.~\Cref{lrtcl-cut-not-admissible}).
In order to do so formally, we first develop some closure properties of cut-free cyclic provability.

\begin{proposition}
	[Inversions]
	\label{inversions-lrtcl}
	We have the following:
	\begin{enumerate}
		\item\label{or-inversion-lrtcl} If $\ltcl \dercyc \Gamma, A\vlor B$ then $\ltcl \dercyc \Gamma, A,B$.
		\item\label{and-inversion-lrtcl} If $\ltcl \dercyc \Gamma, A\vlan B$ then $\ltcl \dercyc \Gamma, A$ and $\ltcl \dercyc \Gamma, B$.
		\item\label{forall-inversion-lrtcl} If $\ltcl \dercyc \Gamma, \forall x A(x)$ then $\ltcl \dercyc \Gamma, A(\cons{c})$, as long as $\cons{c}$ is fresh.
	\end{enumerate}
\end{proposition}
\begin{proof}[sketch]
	All three statements are proved similarly.
	
	For \Cref{or-inversion-lrtcl}, replace every direct ancestor of $A\vlor B$ with 
	%the cedent \marianna{set?} 
	$A,B$.
	The only critical steps are when $A\vlor B$ is principal, in which case we delete the step, or is weakened, in which case we apply two weakenings, one on $A$ and one on $B$.
	If the starting proof had only finitely many distinct subproofs (up to substitution), say $n$, then the one obtained by this procedure has at most $2n$ distinct subproofs (up to substitution), since we simulate a weakening on $A\vlor B$ by two weakenings.
	
	For \Cref{and-inversion-lrtcl}, replace every direct ancestor of $A\vlan B$ with $A$ or $B$, respectively.
	The only critical steps are when $A\vlan B$ is principal, in which case we delete the step and take the left or right subproof, respectively, or is weakened, in which case we simply apply a weakening on $A$ or $B$, respectively.
	The proof we obtain has at most the same number of distinct subproofs (up to substitution) as the original one.
	
	For \Cref{forall-inversion-lrtcl}, replace every direct ancestor of $\forall x A(x)$ with $A(\cons{c})$.
	The only critical steps are when $\forall x A(x)$ is principal, in which case we delete the step and rename the eigenvariable in the remaining subproof everywhere with $\cons{c}$, or is weakened, in which case we simply apply a weakening on $A(\cons{c})$. 
		The proof we obtain has at most the same number of distinct subproofs (up to substitution) as the original one.
\end{proof}

\begin{proposition}
	[Predicate admissibility]
	\label{predicate-admissibility-lrtcl}
	Suppose $\ltcl \dercyc \Gamma,  p(t)$ or $\ltcl \dercyc \Gamma, \bar p(t)$, where $\bar p$ or $p$ (respectively) does not occur in $\Gamma$.
	Then it holds that $\ltcl \dercyc \Gamma$.
%	\begin{itemize}
%		\item $\Gamma$ has no occurrence of $p$ or $\bar p$; and,
%		\item $c$ and $d$ are distinct constant symbols; and,
%		\item $\derd$ has either no occurrence of 
%	\end{itemize}
%	Then there is a cyclic proof $\derd'$ of $\Gamma$ with at most the same number of distinct subproofs (up to substitution) as $\derd$.
\end{proposition}
%Let us take a moment to point out that the proposition above holds even if one (or both) of $\bar p(c)$ and $p(d)$ do not occur in the conclusion of $\derd$, thanks to weakening.

\begin{proof}[sketch]
	Delete every ancestor of $p(t)$ or $\bar p(t)$, respectively.
	The only critical case is when one of the formulas is weakened, in which case we omit the step.
	Note that there cannot be any identity on $p$, due to the assumption on $\Gamma$, and by the subformula property.
\end{proof}

\subsubsection{Reducing to a relational tautology.}
% \todo{add the macros for contants in this section: $\cons{c}$}
Here, and for the remainder of this section, we shall simply construe $\PDLtr$ programs $\alpha$ and formulas $A$ as $\TCL$ formulas with two free variables and one free variable, respectively, by identifying them with their standard translations $\st {x,y} \alpha$ and $\st x A$, respectively.
This modest abuse of notation will help suppress much of the notation in what follows.

\begin{lemma}
	\label{reduction-to-relational-tautology}
% 	If $\ltcl \dercyc \st {\cons{c}} {\diaP{(aa \cup aba)^+}p \vljm \diaP{a^+((ba^+)^+\cup a)}p}$ then also $\ltcl \dercyc \st{\cons{c},\cons{d}}{(aa \cup aba)^+} \vljm \st{\cons{c},\cons{d}}{a^+((ba^+)^+\cup a)}$.
	If $\ltcl \dercyc  ({\diaP{(aa \cup aba)^+}p \vljm \diaP{a^+((ba^+)^+\cup a)}p})(c)$ then also $\ltcl \dercyc {(aa \cup aba)^+}(c,d) \vljm ({a^+((ba^+)^+\cup a)})(c,d)$.
\end{lemma}

\begin{proof}[sketch]
Suppose $\ltcl \dercyc ({\diaP{(aa \cup aba)^+}p \vljm \diaP{a^+((ba^+)^+\cup a)}p})(c)$ so, by unwinding the definition of $\ST$ and since duality commutes with the standard translation, cf. Sec.~\ref{ssec:standard_tr}, we have that $\ltcl \dercyc  ({\boxP{(aa \cup aba)^+}\bar p})(c) \vlor({\diaP{a^+((ba^+)^+\cup a)}p})(c)$. By $\vlor$-inversion (Prop.~\ref{inversions-lrtcl}.\ref{or-inversion-lrtcl}) we have:
\[
\ltcl \dercyc  ({\boxP{(aa \cup aba)^+}\bar p})(c) , ({\diaP{a^+((ba^+)^+\cup a)}p})(c)\]
Again unwinding the definition of $\ST $, and by the definition of duality, we thus have:
\[
\ltcl \dercyc \forall x (\overline{(aa \cup aba)^+}(c,x)) \vlor \bar p(x)) , \exists y ((a^+((ba^+)^+ \cup a) (c,y))
  \vlan p(y))
\]
Now, by $\forall$-inversion and $\vlor$-inversion, Prop.~\ref{inversions-lrtcl}.\ref{forall-inversion-lrtcl},\ref{or-inversion-lrtcl}, we have:
\[
\ltcl \dercyc \overline{(aa \cup aba)^+}(c,d) , \bar p(d) , \exists y (
% \exists z (\tc acz \vlan (\tc{(ba^+)}zy \vlor a(z,y)))  
(a^+((ba^+)^+ \cup a) (c,y))
\vlan p(y))
\]
Without loss of generality we may instantiate the $\exists y$ by $d$
and so by $\vlan$-inversion, Prop.~\ref{inversions-lrtcl}.\ref{and-inversion-lrtcl}, we have:
% \[
% \lrtcl \dercyc  \cortc{\bar a \vlan \bar b }cd , \bar p(d) , \exists z (\rtc acz \vlan \rtc{\ST(ba^*)}zd ) 
% \]
\[
\ltcl \dercyc \overline{(aa \cup aba)^+}(c,d) , \bar p(d) , 
% \exists z (\tc acz \vlan (\tc{(ba^+)}zy \vlor a(z,y)))  
a^+((ba^+)^+ \cup a) (c,d)
\]
%Now notice that
Since there is no occurrence of $p$ above, by Prop.~\ref{predicate-admissibility-lrtcl} we conclude
% \[
% \lrtcl \dercyc  \cortc{\bar a \vlan \bar b }cd ,  \exists z (\rtc acz \vlan \rtc{\ST(ba^*)}zd ) 
% \] 
\[
\ltcl \dercyc \overline{(aa \cup aba)^+}(c,d) , 
% \exists z (\tc acz \vlan (\tc{(ba^+)}zy \vlor a(z,y)))  
a^+((ba^+)^+ \cup a) (c,d)
\]
as required.
\end{proof}

\subsubsection{Irregularity via an adversarial model.}
In the previous subsubsection we reduced the incompleteness of cut-free cyclic sequent proofs for $\TCL$ over the image of the standard translation on $\PDLtr$ to the non-regular cut-free provability of a particular relational validity.  
Unwinding this a little, the sequent that we shall show has no (cut-free) cyclic proof in $\ltcl$ can be written in `two-sided notation' as follows:
\begin{equation}
	\label{eq:irregular-sequent-relational}
	\tc{aa \vlor aba}cd \seqar \exists z (a^+(c,z) \vlan ({(ba^+)^+ \cup a})(z,d))
\end{equation}
This two-sided presentation is simply a notational variant that allows us to more easily reason about the proof search space (e.g.\ referring to `LHS' and `RHS'). 
Formally:
\begin{remark}
[Two-sided notation]
We may write $\Gamma \seqar \Delta $ as shorthand for the sequent $\bar \Gamma ,\Delta$, where $\bar \Gamma = \{ \bar A: A \in \Gamma\}$.
References to the `left-hand side (LHS)' and `right-hand side (RHS)' have the obvious meaning, always with respect to the delimiter $\seqar$.
\end{remark}

To facilitate our argument, we shall only distinguish sequents `modulo substitution' rather than allowing explicit substitution steps when reasoning about (ir)regularity of a proof.

We shall design a family of `adversarial' models, and instantiate proof search to just these models.
In this way, we shall show that any non-wellfounded $\ltcl$ proof of the sequent \eqref{eq:irregular-sequent-relational} must have arbitrarily long branches without a repetition (up to substitution).
Since $\ltcl$ is finitely branching, by K\"onig's Lemma this means that any non-wellfounded $\ltcl$ proof of \eqref{eq:irregular-sequent-relational} has an infinite branch with no repetitions (up to substitution), as required.

\begin{definition}
	[An adversarial model]
	\label{adversarial-model}
	For $n\in \Nat$, define the structure $\An n$ as follows:
	\begin{itemize}
		\item The domain of $\An n$ is $\{u_0,u_0', \dots, u_{n-1}',u_n, v \}$.
		\item $\interp {\An n} a = \{(u_i, u_{i}'),(u_{i}',u_{i+1}) \}_{i<n}$.
		\item $\interp {\An n} b = \{(u_n,v) \}$.
	\end{itemize}
\end{definition}
\todo{give a picture of the structure}

Note that, since the sequent \eqref{eq:irregular-sequent-relational} that we are considering is purely relational, it does not matter what sets $\An n$ assigns to the predicate symbols.

\begin{lemma}
	\label{branch-no-repetitions-until-arbitrary}
	Let $n\in \Nat$. 
	Any proof $\derd$ of \eqref{eq:irregular-sequent-relational} has a branch with no repetitions (up to substitutions) among its first $n$ sequents.
\end{lemma}
\begin{proof}
Set $c_0=c$.
	Consider some (possibly finite, but maximal) branch $\branch = (\infrule_i)_{i\leq \nu}$ (with $\nu\leq \omega$) of $\derd$ satisfying:
	\begin{itemize}
		\item\label{branch-convention-rtc-right} whenever $\TC$ on the LHS is principal, the right premiss is followed; and,
		\item\label{branch-convention-aorb-left-a} whenever $(aa)(s,t)\vlor (aba)(s,t)$ is principal (for any $s$ and $t$) the left premiss (corresponding to $(aa)(s,t)$) is followed.
	\end{itemize}
	Let $k \leq n$ be maximal such that, for each $i\leq k$, $\infrule_i$ has principal formula on the LHS.
	Now:
	\begin{enumerate}
		\item\label{lhs-branch-form} For $i\leq k$, each $\infrule_i$ has conclusion with LHS of the form:
		\begin{equation}
			\label{eq:counterexample-general-form-lhs}
			\Gamma_j(c_{j-1} , c_j), \tc {aa\vlor aba} {c_l} d
		\end{equation}
		where $\vec c=c_0, \dots, c_{l-1}$
		for some $l\leq i$ 
		and each $\Gamma_j(c_{j-1},c_j)$ has the form $a(c_{j-1}, c_{j-1}'),a(c_{j-1}',c_j)$ or
		$a(c_{j-1}, c_{j-1}')\vlan a(c_{j-1}',c_j)$ or
		$(aa)(c_{j-1},c_{j})$
		or $(aa)(c_{j-1}, c_j ) \vlor (aba)(c_{j-1}, c_{j}) $
% 		and with each $A_j(c_{j-1},c_j)$ of the form $a(c_{j-1}, c_j)$ or $a(c_{j-1},c_j)\vlor b(c_{j-1},c_j)$.
		To see this, proceed by induction on $i\leq k$:
		\begin{itemize}
			\item The base case is immediate, by setting $l=0$.
			\item For the inductive step, note that the principal formula of $\infrule_i$ must be on the LHS, since $i\leq k$.
			Thus by the inductive hypothesis the principal formula of $\infrule_i$ must have the form: 
			\begin{itemize}
			\item $a(c_{j-1}, c_j')\vlan a(c_j',c_j)$ (on the LHS), in which case the premiss of $\infrule_i$ (which is a left-$\vlan$ step) replaces it by $a(c_{j-1}, c_j'),a(c_j',c_j)$;
			\item $(aa)(c_{j-1,c_j})$ (on the LHS), in which case the premiss of $\infrule_i$ (which is a left-$\exists$ step) replaces it by $a(c_{j-1}, c_j')\vlan a(c_j',c_j)$
				\item $(aa)(c_{j-1},c_j)\vlor (aba)(c_{j-1},c_j)$ (on the LHS), in which case, by definition of $\branch$, the $\branch$-premiss of $\infrule_i$ (which is a left-$\vlor$ step) replaces this formula by some $a(c_{j-1},c_j)$; or,
				\item $\tc{aa\vlor aba}{c_j} d$ (on the LHS), in which case, by definition of $\branch$, the $\branch$-premiss of $\infrule_i$ (which is a left-$\TC$ step) replaces it by the cedent $(aa)(c_j,c_{j+1})\vlor (aba)(c_j,c_{j+1}),\tc{aa\vlor aba}{c_{j+1}}d$. 
			\end{itemize}
		\end{itemize}
	\item\label{lhs-branch-does-not-repeat} Moreover, for $i<j\leq k$, the conclusion of $\infrule_i$ and $\infrule_j$ are not equal (up to substitution). 
	To see this, note that any rule principal on an LHS of form in \eqref{eq:counterexample-general-form-lhs} either decreases the size of some $\Gamma_j(c_{j-1},c_j)$ (when it is a left $\vlan,\exists$ or $\vlor$ step) or increases the number of eigenvariables in the sequent (when it is a left $\TC$ step), in particular the $l$ such that $\tc{aa \vlor aba}{c_l,d}$ appears.
	\item\label{lhs-branch-has-true-lhs} 
		Since proofs must be sound for all models (by soundness), we shall work in $\An n $
		with respect to an interpretation $\rho_n$ satisfying $ c_i \mapsto u_i$ for $i\leq n$ and $c_i'\mapsto u_i'$ for $i<n$ and $ d \mapsto v$.
	It follows by inspection of \eqref{eq:counterexample-general-form-lhs} that, for $i\leq k$, each formula on the LHS of the conclusion of $\infrule_i$ is true in $(\An n, \rho_n)$.
%	\begin{itemize}
%		\item The conclusion LHS is already satisfied in $(\An n, \rho_n)$.
%		\item Every formula in the LHS either has the form $a(c_i,c_{i+1})$, $a(c_i,c_{i+1})\vlor b(c_i,c_{i+1})$ or $\rtc {a\vlor b}{c_i,d}$ for some $i\leq k\leq n$, all of which are satisfied in $(\An n, \rho_n)$ as long as $i\leq n$
%	\end{itemize}
	\item\label{branch-principal-on-lhs-for-long} Along $\branch$, the RHS cannot be principal until $l=n$ in \eqref{eq:counterexample-general-form-lhs}, so in particular $k\geq n$. 
	To see this:
	\begin{itemize}
		\item Recall that the interpretation $\rho_n$ assigns to $c_0,c_0', \dots, c_{n-1}',c_n$ the worlds $u_0, u_0', \dots, u_{n-1}',u_n$ respectively.
		\item If the existential formula on the RHS is instantiated by some $c_i$ with $i\neq n$ or $c_i'$ with $i<n$ then the resulting sequent is false in $(\An n,\rho_n)$ (recall that, by \Cref{lhs-branch-has-true-lhs}, every formula on the LHS is true, so we require the RHS to be true too).
		To see this, note that the RHS in particular would imply $(ba^+)^+(c_i,d)$ or $a(c_i,d)$ or $(ba^+)^+(c_i',d)$ or $a(c_i',d)$.
		However when $i\leq n$, or $i<n$ respectively, this is not true with respect to $(\An n,\rho_n)$.
	\end{itemize}
	\end{enumerate}
By \Cref{branch-principal-on-lhs-for-long}, we have that $k\geq n$ and so $k=n$.
Thus, by \Cref{lhs-branch-form} and \Cref{lhs-branch-does-not-repeat}, there are no repeated sequents (up to substitution) in $(\infrule_i)_{i\leq n}$, as required.
\end{proof}

\subsubsection{Putting it all together.}
We are now ready to give the proof of the main result of this section.

\begin{proof}
	[Proof of Thm.~\ref{thm:incompleteness}, sketch]
	Since the choice of $n$ in \Cref{branch-no-repetitions-until-arbitrary} was arbitrary, any $\ltcl$ proof $\derd$ of \eqref{eq:irregular-sequent-relational} must have branches with arbitrarily long initial segments without any repetition (up to substitution).
	Since the system is finitely branching, by K\"onig's Lemma we have that there is an infinite branch through $\derd$ without any repetition (up to substitution), and thus $\derd$ is not regular.
	Thus $\ltcl \not\dercyc \eqref{eq:irregular-sequent-relational}$.
	Finally, by contraposition of \Cref{reduction-to-relational-tautology}, we have, as required: 
	$$\ltcl \not\dercyc  ({\diaP{(aa \cup aba)^+}p \vljm \diaP{a^+((ba^+)^+\cup a)}p})(c) \qedhere $$
\end{proof}

%%%%%%%%%%%%%%%%%%%%%%%%%%%%%%%%%%%%%%%%%%%%%%%%%%%%%%%%%%%%%%%%%%%%%%%%%%%%%%%%%%%%%%%%%%%%%%%%%%%%%%%%%%%%%

An immediate consequence of Thm.~\ref{thm:incompleteness} is:
\begin{corollary}
\label{lrtcl-cut-not-admissible}
The class of cyclic proofs of $\ltcl$ does not enjoy cut-admissibility.
\end{corollary}

%%%%%%%%%%%%%%%%%%%%%%%%%%%%%%%%%%%%%%%%%%%%%%%%%%%%%%%%%%%%%%%%%%%%%%%%%%%
%%%%%%%%%%%%%%%%%%%%%

% \section{Sequent-style circular systems}
% \label{sec:sequent_calculi}
%%%%%%%%%%%%%%%%%%%%%%%%%%%%%%%%%%%%%%%%%%%%%%%%%%%%%%%%%%%%%%%%%%%%%%%%%%%
%%%%%%%%%%%%%%%%%%%%%%%%%%%%%%%%%%%%%%%%%%%%%%%%%%%%%%%%%%%%%%%%%%%%%%%%%%%

%%%%%%%%%%%%%%%%%%%%%%%%%%%%%%%%%%%%%%%%%%%%%%%%%%%%%%%%%%%%%%%%%%%%%%%%%%%
%%%%%%%%%%%%%%%%%%%%%%%%%%%%%%%%%%%%%%%%%%%%%%%%%%%%%%%%%%%%%%%%%%%%%%%%%%%

%%%%%%%%%%%%%%%%%%%%%%%%%%%%%%%%%%%%%%%%%%%%%%%%%%%%%%%%%%%%%%%%%%%%%%%%%%%

%%%%%%%%%%%%%%%%%%%%%%%%%%%%%%%%%%%%%%%%%%%%%%%%%%%%%%%%%%%%%%%%%%%%%%%%%%%
%%%%%%%%%%%%%%%%%%%%%%%%%%%%%%%%%%%%%%%%%%%%%%%%%%%%%%%%%%%%%%%%%%%%%%%%%%%

%%%%%%%%%%%%%%%%%%%%%%%%%%%%%%%%%%%%%%%%%%%%%%%%%%%%%%%%%%%%%%%%%%%%%%%%%%%

%%%%%%%%%%%%%%%%%%%%%%%%%%%%%%%%%%%%%%%%%%%%%%%%%%%%%%%%%%%%%%%%%%%%%%%%%%%
%%%%%%%%%%%%%%%%%%%%%%%%%%%%%%%%%%%%%%%%%%%%%%%%%%%%%%%%%%%%%%%%%%%%%%%%%%%
%%%%%%%%%%%%%%%%%%%%%%%%%%%%%%%%%%%%%%%%%%%%%%%%%%%%%%

%%%%%%%%%%%%%%%%%%%%%%%%%%%%%%%%%%%%%%%%%%%%%%%%%%%%%%%%%%%%%%%%%%%%%%%%%%%
%%%%%%%%%%%%%%%%%%%%%%%%%%%%%%%%%%%%%%%%%%%%%%%%%%%%%%%%%%%%%%%%%%%%%%%%%%%
\section{Hypersequent calculus for $\TCL$}
\label{sec:hypersequents_tcl}
%%%%%%%%%%%%%%%%%%%%%%%%%%%%%%%%%%%%%%%%%%%%%%%%%%%%%%%%%%%%%%%%%%%%%%%%%%%

%\anupam{I moved the PDL system to later so that we can get to the hypersequent system sooner. The prose below should still be meaningful, but bear in mind the reader has not seen yet the rules for PDL. Maybe keep the system for PDL, but not the notion of cylic proof?}
%\marianna{Or maybe move this after the presentation of PDL. I fear there won't be space for this anyways..}

% The previous section formally demonstrated the inadequacy of (cut-free) cyclic sequent proofs for $\TCL$, at least with regards to simulating cut-free cyclic reasoning in Kleene Algebra or $\PDL$. 
% Morally speaking, it is not difficult to see why the standard translation does not lift to sequent proofs in this way.
Let us take a moment to examine why any `local' simulation of $\lpdltr$ by $\ltcl$ fails, in order to motivate the main system that we shall present. 
The program rules, in particular the $\diaP{\, }$-rules, require a form of \emph{deep inference} to be correctly simulated, over the standard translation.
For instance, let us consider the action of the standard translation on two rules we shall see later in $\lpdltr$ (cf.~Sec.~\ref{ssec:sequents_pdltr}):
\[
\scriptsize
%\begin{adjustbox}{max width = \textwidth}
%$
\begin{array}{c@{\qquad \leadsto \qquad}c}
     \vlinf{\casezero{\diaP\union}}{}{\Gamma, \diaP{a_0 \union a_1}p}{\Gamma, \diaP{a_0}p}
&
\vlinf{}{}{\st c \Gamma , \exists x (\blue{(a_0(c,x) \vlor a_1 (c,x)}) \vlan p(x)) }{\st c \Gamma, \exists x (\blue{a_0(c,x)} \vlan p(x)) }
\\
\noalign{\medskip}
\vlinf{\diaP;}{}{\Gamma, \diaP{a;b}p}{\Gamma, \diaP a \diaP b p}
&
\vlinf{}{}{\st c \Gamma , \red{\exists x ( \exists y ( a(c,y) \vlan b(y,x) ) \vlan p(x) )} }{ \st c \Gamma, {\exists y( a(c,y) \vlan \exists x ( b(y,x) \vlan p(x) ) )} }
\end{array}
%$
%\end{adjustbox}
\]
The first case above suggests that any system to which the standard translation lifts must be able to reason \emph{underneath $\exists$ and $\vlan$}, so that the inference indicated in blue is `accessible' to the prover.
The second case above suggests that the existential-conjunctive meta-structure necessitated by the first case should admit basic equivalences, in particular certain \emph{prenexing}.
This section is devoted to the incorporation of these ideas (and necessities) into a bona fide proof system.

%\anupam{issue that `structure' is overloaded. changing to `annotated cedent' ? }
%
%\anupam{Subtlety: our hypersequents should really be sets, in order to properly simulate the PDL calculus, but we can ignore this in the soundness proof without loss of generality (to be justified). I'm making everything a set for now.}
%
%\marianna{Specify conditions on variables in a hypersequent and in a derivation: 
%	\begin{itemize}
%		\item In the same structure, free and bound variables need to be distinct;
%		\item Along an hypertrace of a derivation, $x$ cannot be instantiated twice.
%	\end{itemize}
%}
%\anupam{First point is obsolete now since, strictly speaking, there are no free variables. Second point should be a consequence of only proving theorems where each variable is bound by at most one quantifier.}

\subsection{Annotated hypersequents}
\label{sec:htcl_hypersequents}

An \emph{annotated cedent}, or simply \emph{cedent}, written $\ced,\ced'$ etc., is an expression $\str \Gamma \bv$, where $\Gamma$ is a set of formulas and the \emph{annotation} $\bv$ is a set of variables.
% Annotations are usually written without braces using the comma for set union, e.g.\ writing simply $x_1, \dots, x_n$ for the set $\{x_1, \dots, x_n\}$.
We sometimes construe annotations as lists rather than sets when it is convenient, e.g.\ when taking them as inputs to a function.
% , in which case we assume some fixed canonical ordering of variables.

% We use the metavariables $\ced, \ced'$ etc.\ to vary over annotated cedents. 
%\marianna{what about $ \mathcal{C} $, $ \mathcal{C}' $ ? The $ S $ is too similar to $ \mathcal{S} $} \anupam{isn't that a feature not a bug? Could also change many mathcals to mathbfs to be more suggestive.}
Each cedent may be intuitively read as a $ \TCL$ formula, under the following interpretation:
\(
\fmi{\str \Gamma {x_1, \dots, x_n}} \  \defsym \ \exists x_1 \dots \exists x_n  \bigwedge \Gamma
\).
When $\bv= \emptyset$ then there are no existential quantifiers above, and when $\Gamma = \emptyset$ we simply identify $\bigwedge \Gamma$ with $\top$.
We also sometimes write simply $A$ for the annotated cedent $\str A \emptyset$.

A \emph{hypersequent}, written $\sq,\sq'$ etc., is a set of annotated cedents.
% , written:
% \(
% \str{\Gamma_1}{\bv_1}, \dots, \str{\Gamma_n}{\bv_n}
% \)
% where each  $ \str{\Gamma_i}{\bv_i} $ is called an existential structure, $ \Gamma_i $ is a set of formulas, and $ \bv_i $ is a (possibly empty) set of variables. 
% 
% Intuitively, the $ \bv_i $ are keeping track of the free  variables occurring in $ \Gamma_i $. We write $ \bv,x $ meaning $ \bv \cup x $. 
% The meaning of a structure is defined by the following formulas:
% \begin{eqnarray*}
% 	\fmi{\str{\;}{\varnothing}} & \defsym & \top\\
% 	\fmi{ \str{\Gamma}{\varnothing}} & \defsym &  \bigwedge \Gamma\\
% 	\fmi{ \str{\Gamma}{\bv}} & \defsym &  \exists x_1 \dots \exists x_n(\bigwedge \Gamma), \text{ for } \bv = \{x_1, \dots, x_n\}
% \end{eqnarray*}
% We use $\sq, \sq'$ etc.\ to vary over hypersequents.
% (since they are sets of annotated cedents).
Each hypersequent may be intuitively read as the disjunction of its cedents.
Namely we set:
\(
\fmi{\str{\Gamma_1}{\bv_1}, \dots, \str{\Gamma_n}{\bv_n} } \, \defsym \, \fmi{\str{\Gamma_1}{\bv_1}} \vlor \dots \vlor \fmi{\str{\Gamma_n}{\bv_n}}.
\)

\begin{figure}[t!]
	\begin{center}
		\footnotesize
		\begin{tabular}{|@{\hspace{2cm}} c c @{\hspace{2cm}}|}
			\hline 
			& \\
			\multicolumn{2}{|c|}{
			$ \vlinf{\init}{}{\str{\,}{\varnothing}}{} 
			\quad
			\vlinf{\wk}{}{\violet\sq,\sq'}{\violet\sq}
			\quad
			\vlinf{ \sigma}{}{\violet{\sigma(\sq)}}{\violet \sq}
			\quad \vliinf{\cup}{\textcolor{black}{\scriptsize{\begin{matrix}
							\freev{\Delta} \cap \bx = \emptyset \\ 
							\freev{\Gamma} \cap \by = \emptyset 
				\end{matrix}}}
			}{\violet \sq  ,\str{\magenta \Gamma, \orange \Delta}{\bx,\by} }{\violet \sq , \str{\magenta \Gamma}{\bx}}{\violet \sq, \str{ \orange \Delta}{\by}} $
			}
			\\[0.5cm]
			\multicolumn{2}{|c|}{
			$\quad  \vlinf{\reset}{\scriptsize{\textcolor{black}{
						\begin{matrix}
				% 		\freev{A} \cap \bv \neq \emptyset 
				        \text{$A$ closed}		
						\\
						\end{matrix}
			}}}{\violet \sq , \str{ \magenta \Gamma, A}{\bv},\str{\nf A }{\varnothing}}{\violet \sq , \str{\magenta \Gamma}{\bv}} \quad $
			$ \vlinf{\vlan}{}{\violet \sq, \str{\magenta \Gamma, \blue{A \vlan B}}{\bv}}{\violet \sq  ,\str{\magenta \Gamma, \blue A , \blue B}{\bv}} \quad $
			$ \vlinf{\vlor}{\text{\scriptsize $i \in \{0,1\}$ }}{\violet \sq , \str{\magenta \Gamma, \blue{A_0 \vlor A_1}}{\bv}}{\violet \sq , \str{\magenta\Gamma, \blue{A_i} }{\bv}} \quad $
			}
			\\[0.5cm]
			\multicolumn{2}{|c|}{
			$ \quad \vlinf{\mathsf{inst}}{}{ \violet \sq , \str{\magenta{\Gamma(y)}}{\bv ,y}}{\violet \sq , \str{\magenta{\Gamma(t)}}{\bv}} \quad $
			$ \vlinf{\exists}{\footnotesize{\textcolor{black}{
						\begin{matrix}
						\text{\scriptsize{$y$ fresh}} \\
						\end{matrix}
			}}}{ \violet \sq , \str{\magenta \Gamma, \blue{\exists x (A(x))}}{\bv} }{ \violet \sq , \str{\magenta \Gamma, \blue{A(y)}}{\bv, y} } \quad $
			$ \vlinf{\forall}{
			\footnotesize{
			\begin{matrix}
			\text{\scriptsize{$f$ fresh}}
			\end{matrix}
			}
			}{ \violet \sq, \str{\magenta \Gamma, \blue{\forall x (A(x))} }{\bv } }{\violet \sq , \str{\magenta \Gamma, \blue{A(f(\bv))}}{\bv}} \quad  $
			}
			\\[0.5cm]
%			$ \vlinf{\mathsf{wk}}{}{\sq, \str{\Gamma}{\bv}}{\sq} $ & 
%			$ \vlinf{\mathsf{ctr}}{}{\sq, \str{\Gamma}{\bv}}{\sq, \str{\Gamma}{\bv},\str{\Gamma}{\bv}} $ \\	[0.7cm]
			%\hline 
			%& \\
			\multicolumn{2}{|c|}{
				$ \vlinf{\TC}{
				\footnotesize{\begin{matrix}
					\text{\scriptsize{$z$ fresh}}
					\end{matrix}}
				}{\violet \sq, \str{\magenta \Gamma, \blue{\tc{A}{s}{t}}}{\bv}}{\violet \sq, \str{\magenta \Gamma, \blue{A(s,t)}}{\bv}, \str{\magenta \Gamma,\blue{A(s,z)}, \blue{\tc{A}{z}{t}}}{\bv,z} } $ }\\[0.5cm]
			\multicolumn{2}{|c|}{
				$ \vlinf{\coTC}{
					\footnotesize{
						\begin{matrix}
						\text{\scriptsize{$f$ fresh}}
						\end{matrix}
					}
				}{\violet \sq, \str{\magenta\Gamma, \blue{\cotc{A}{s}{t}} }{\bv}}{ \violet \sq, \str{\magenta\Gamma,  \blue{A(s,t)}, \blue{A(s,f(\bv))}}{\bv}, \str{\magenta\Gamma, \blue{A(s,t)},\blue{\cotc{A}{f(\bv)}{t}}}{\bv} } $}\\
			& \\
			\hline 
		\end{tabular}
	\end{center}
	\caption{Hypersequent calculus $\htcl$. 
	 $\sigma$ is a `substitution' map from constants to terms and a renaming of other function symbols and variables.
% 	 , extended to terms, formulas, cedents and hypersequents in the natural way.
% 	, extended to terms, formulas, cedents and hypersequents in the natural way.
	}
	\label{fig:rules:htcl}
\end{figure}

\subsection{Non-wellfounded hypersequent proofs}
\label{sec:htcl_proofs}

We now present our hypersequential system for $\TCL$ and its corresponding notion of `non-wellfounded proof'. 

\begin{definition}
    [System]
    	\label{def:pre-proof_tcl}
    The rules of $ \htcl $ are given in Fig.~\ref{fig:rules:htcl}.
	A $ \htcl $ \emph{preproof} is a (possibly infinite) derivation tree generated by the rules of $ \htcl $. 
	A preproof is \emph{regular} if it has only finitely many distinct subproofs.
\end{definition}

\todo{explain the rules, as much as possible}

% While we have included an explicit substitution rule we shall, as in earlier sections, often work `modulo substitution' when writing down cyclic preproofs.
% \anupam{what is the type of substitution? make clear.}
%Note that we have simply omitted the substitution rule in our hypersequential setting, in favour of distinguishing hypersequents only up to substitution.
%Note that, if we admit a substitution rule,
%\[
%\vlinf{}{}{\sigma(\sq)}{\sq}
%\]
%where $\sigma$ is a renaming of constant and variable symbols, then we can assume without loss of generality that preproofs are regular if and only if they have only finitely many distinct subproofs.

\begin{remark}
	[Herbrand constants]
	Our rules for $\TC$ and $\coTC$ are induced by the characterisation of $\TC$ as a least fixed point in \eqref{eq:tc-fixed-point-formula}.
	Note that the rules $\coTC$ and $\forall$ introduce, bottom-up, the fresh function symbol $f$, which plays the role of the \emph{Herbrand} function of the corresponding $\forall$ quantifier:
		just as $\forall \bv \exists x A( x)$ is equisatisfiable with $\forall \bv A( f(\bv))$, when $f$ is fresh, by Skolemisation, by duality $\exists \bv \forall x A(x)$ is equivalid with $\exists \bv A(f(\bv))$, when $f$ is fresh, by Herbrandisation.
		Note that the usual $\forall$ rule of the sequent calculus is just a special case of this, when $\vec x = \emptyset$, and so $f$ is a constant symbol.
\end{remark}
\anupam{explain more the intricacies of progress in hypersequents?}

Our notion of ancestry, as compared to traditional sequent systems, must account for the richer structure of hypersequents.
Informally, referring to \Cref{fig:rules:htcl}, 
	a formula $C$ in the premiss is an \emph{immediate ancestor} of a formula $C'$ in the conclusion if they have the same colour;
	if $C,C' \in \magenta \Gamma$ then we further require $C=C'$, and if $C,C'$ occur in $\violet \sq$ then $C=C'$ occur in the same cedent.
A cedent $S$ in the premiss is an \emph{immediate ancestor} of a cedent $S'$ in the conclusion if some formula in $S$ is an {immediate ancestor} of some formula in $S'$.
The following definitions formally explain these notions independently of of the colouring in \Cref{fig:rules:htcl}.
\anupam{moved colour based definitions here.}

\begin{definition}
    [Ancestry for cedents]
    \label{def:ancestry_htcl}
    Fix an inference step $\infrule$, as typeset in Fig.~\ref{fig:rules:htcl}. 
    We say that a cedent $\ced$ in a premiss of $\infrule$ is an \emph{immediate ancestor} of a cedent $\ced'$ in the conclusion of $\infrule$ if either:
    \begin{itemize}
    	\item $\ced = \ced' \in \sq $, i.e.\ $\ced$ and $\ced'$ are identical `side' cedents of $\infrule$; or,  
    	%\marianna{what's $ \mathcal{S} $?} \anupam{`as typeset in Figure 3'}
    	\marianna{first bullet point: identical modulo substitution?}
    	\item $\infrule \neq \reset$, and $\ced$ is the (unique) cedent indicated in the conclusion of $\infrule$, and $\ced'$ is a cedent indicated in a premiss of $\infrule$; or,
    	\item $\infrule = \reset$ and $\ced$ is the (unique) cedent indicated in the premiss of $\reset$ and $\ced'$ is the cedent $\str{\Gamma,A}\bv$ indicated in the conclusion of $\reset$.
    \end{itemize}
%    A formula $C$ in the premiss is an \emph{immediate ancestor} of a formula $C'$ in the conclusion if they have the same colour;
%    if $C,C' \in \magenta \Gamma$ then we further require $C=C'$, and if $C,C'$ occur in $\violet \sq$ then $C=C'$ occur in the same cedent.
%        
 %   A cedent $S$ in the premiss is an immediate ancestor of a cedent $S'$ in the conclusion if some formula in $S$ is an \emph{immediate ancestor} of some formula in $S'$.
    % \footnote{This compact notion of ancestry for cedents suffices for our purposes, but let us point out that it excludes the empty cedent ever being an immediate ancestor.}
\end{definition}

Note in particular that in $\reset$, as typeset in \Cref{fig:rules:htcl}, $\str \Gamma \bv$ is {not} an immediate ancestor of $\str {\bar A} \emptyset$.

%In order to avoid a technicality of distinguishing formulas and formula occurrences, we shall specialise our notion of formula ancestry to a fixed hypertrace.
%Otherwise, when speaking of `formulas in a rule', we would need to indicate further the cedent in which the formula occurs. 

\begin{definition}[Ancestry for formulas]
	\marianna{modified this}
	Fix an inference step $\infrule$, as typeset in Fig.~\ref{fig:rules:htcl}. 
	We say that a formula $\ced$ in a premiss of $\infrule$ is an \emph{immediate ancestor} of a formula $\ced'$ in the conclusion of $\infrule$ if either:
	\begin{itemize}
		\item $ F = F' \subseteq \sq$, i.e., $ F $ and $ F' $ are formulas occurring in some cedent $ S \in \sq $, and are identical modulo substitution;
		\item $ F = F' \in \Gamma $, i.e., $ F $ and $ F' $ are formulas occurring in $ \Gamma $, and are identical modulo substitution; or,
		\item $ F $ is one of the formulas explicitly indicated in the premiss of $ \infrule $ and $ F' $ is the formula explicitly indicated in the conclusion of $ \infrule $.
	\end{itemize}
\end{definition}

Immediate ancestry on both formulas and cedents is a binary relation, inducing a directed graph whose paths form the basis of our correctness condition:

\begin{definition}[(Hyper)traces]
\label{def:hypertraces_htcl}
    A \emph{hypertrace} is a maximal path in the graph of immediate ancestry on cedents.
    A \emph{trace} is a maximal path in the graph of immediate ancestry on formulas.
\end{definition}

\marianna{added this}
Thus, in the $\reset$ rule, as typeset in \Cref{fig:rules:htcl}, no (infinite) trace can include the indicated $A$ or $\bar A$. 
From the above definitions it follows that whenever a cedent $S$ in the premiss of a rule $ \infrule $ is an immediate ancestor of a cedent $S'$ in the conclusion, then some formula in $S$ is an immediate ancestor of some formula in $S'$, and vice-versa.  Thus, for a hypertrace $ (\ced_i)_{i < \omega} $, there is at least one trace $  (F_i)_{i< \omega} $ which is `within' or `along' the hypertrace, i.e., such that $ F_i \in \ced_i $ for all $ i $.

\begin{definition}
[Progress and proofs]
	\label{def:progress_tcl}
	Fix a preproof $\derd $. 
% 	be a preproof with an infinite branch $\branch = (\infrule_i)_{i \in \omega}$, and let $\htrace = (S_i)_{i \in \omega}$ be a hypertrace along $\branch$.
	A (infinite) trace $ (F_i)_{i \in \omega} $
% 	along $\htrace$ 
	is \emph{progressing} if there is $k$ such that, for all $i>k$, $F_i$ has the form $ \cotc A{s_i}{t_i}$ and is infinitely often principal.\footnote{
		In fact, by a
		a simple well-foundedness argument, it is equivalent to say that $(F_i)_{i <\omega}$ is progressing if it is infinitely often principal for a $\coTC$-formula.}
% 	, i.e.\ indicated in the conclusion of the $\coTC$ rule, as typeset in \Cref{fig:rules:htcl}.
	A (infinite) hypertrace $\htrace$ is \emph{progressing} if every infinite trace along it is progressing. 
	A (infinite) branch is progressing if it has a progressing hypertrace.
	$\derd$ is a 
% 	\emph{non-wellfounded proof} (or simply \emph{proof}) 
\emph{proof}
	if every infinite branch is progressing. 
% 	has a progressing hypertrace.
	If, furthermore, $\derd$ is regular, we call it a \emph{cyclic proof}.
	
	We write $\htcl \dernwf \sq$ (or $\htcl \dercyc \sq$) if there is a proof (or cyclic proof, respectively) of $\htcl$ of the hypersequent $\sq$.
\end{definition}

\subsection{Some examples}
Let us consider some examples of cyclic proofs in $\htcl$ and compare the system to $\ltcl$.

\begin{example}
	[Fixed point identity]
	\label{fixed-point-identity}
    Here follows a cyclic proof in $\ltcl$ of sequent $\str{\cotc a c d }\emptyset, \str{\tc {\bar a} c d }\emptyset$:
	\[
	\vlderivation{
	\vliin{\coTC}{\bullet}{\cotc a c d \tc{\bar a}cd }{
		\vlin{\TC}{}{a(c,d) ,\tc{\bar a }cd}{
		\vlin{\id}{}{a(c,d),\bar a(c,d)}{\vlhy{}}
		}
	}{
		\vliin{\TC}{}{a(c,e),\cotc a e d, \tc{\bar a}cd}{
			\vlin{\id}{}{a(c,e), \bar a(c,e)}{\vlhy{}}
		}{
			\vlin{\coTC}{\bullet}{\cotc a e d , \tc{\bar a}e d}{\vlhy{\vdots}}
		}
	}
	}
	\]
	There is not much choice in the construction of this cyclic proof, bottom-up: we must apply $\coTC$ first and branch before applying $\TC$ differently on each branch.
	This cyclic proof is naturally simulated by the following $\htcl$ one, where the progressing hypertrace is marked in blue:\anupam{colour progressing hypertrace and explain}
	
	%\[
	\begin{adjustbox}{max width = \textwidth}
	$
	\vlderivation{
	\vlin{\coTC}{\bullet}{\blue{\str{\cotc a c d}\emptyset}, \str{\tc{\bar a}c d}\emptyset}{
	\vliin{2\cup}{}{\str{a(c,d), a(c,e)}\emptyset, \blue{\str{a(c,d), \cotc a e d}\emptyset}, \str{\tc{\bar a}cd}\emptyset}{
		\vlin{\TC}{}{\str{a(c,d)}\emptyset, \str{\tc{\bar a}c d }\emptyset}{
		\vlin{\reset}{}{\str{a(c,d)}\emptyset, \str{\bar a (c,d)}\emptyset}{
		\vlin{\init}{}{\str{\, }\emptyset}{\vlhy{}}
		}
		}
	}{
		\vlin{\TC}{}{\str{a(c,e)}\emptyset ,    \blue{\str{\cotc a e t}\emptyset} , \str{\tc{\bar a}c d }\emptyset}{
		\vlin{\instrule}{}{\str{a(c,e)}\emptyset , \blue{\str{\cotc a e t}\emptyset} , \str{\bar a (c,x),\tc{\bar a}x d }x}{
		\vliin{\cup}{}{\str{a(c,e)}\emptyset , \blue{\str{\cotc a e t}\emptyset} , \str{\bar a (c,e),\tc{\bar a}e d }\emptyset}{
			\vlin{\reset}{}{\str{a(c,e)}\emptyset, \str{\bar a(c,e)}\emptyset}{
			\vlin{\init}{}{\str{\, }\emptyset}{\vlhy{}}
			}
		}{
			\vlin{\coTC}{\bullet}{\blue{\str{\cotc a e d}\emptyset}, \str{\tc{\bar a}ed}\emptyset}{\vlhy{\vdots}}
		}
		}
		}
	}
	}
	}
	$
	%\]
	\end{adjustbox}
	
	\ 
	
	The sequent $\str{\cotc a c d }\emptyset,\str{\tc {\bar a} c d }\emptyset$ is finitely derivable using rule $\reset$ on $\tc a c d$ and the $\init$ rule.
	However we can also cyclically reduce it to a simpler instance of $\reset $.
	Due to the granularity of the inference rules of $\htcl$, we actually have some liberty in how we implement such a derivation.
	E.g., the $\htcl$-proof below applies $\TC$ rules below $\coTC$ ones, and delays branching until the `end' of proof search, which is impossible in $\ltcl$.
	The only infinite branch, looping on $\bullet$, is progressing by the blue hypertrace.
	
	\begin{adjustbox}{max width = \textwidth}
	$
	\scriptsize
	\vlderivation{
	\vlin{\TC}{\bullet}{\blue{\str{\cotc a c d}\emptyset}, \str{\tc{\bar a}c d }\emptyset}{
	\vlin{\coTC}{}{\blue{\str{\cotc a c d}\emptyset}, \str{\bar a(c,d)}\emptyset, \str{\bar a(c,x), \tc{\bar a}xd}x}{
	\vlin{\instrule}{}{\str{a(c,d),a(c,e)}\emptyset, \blue{\str{a(c,d),\cotc aed}\emptyset} , \str{\bar a(c,d)}\emptyset, \str{\bar a(c,x), \tc{\bar a}xd}x}{
	\vliin{2\cup}{}{\str{a(c,d),a(c,e)}\emptyset, \blue{\str{a(c,d),\cotc aed}\emptyset} , \str{\bar a(c,d)}\emptyset, \str{\bar a(c,e), \tc{\bar a}ed}\emptyset }{
		\vlin{\reset}{}{\str{a(c,d)}\emptyset,\str{\bar a(c,d)}\emptyset}{
		\vlin{\init}{}{\str{\, }\emptyset}{\vlhy{}}
		}
	}{
		\vliin{\cup}{}{\str{a(c,e)}\emptyset, \blue{\str{\cotc a e d}\emptyset}, \str{\bar a(c,e),\tc{\bar a}ed}\emptyset}{
			\vlin{\reset}{}{\str{a(c,e)}\emptyset, \str{\bar a(c,e)}\emptyset}{
			\vlin{\init}{}{\str{\, }\emptyset}{\vlhy{}}
			}
		}{
			\vlin{\TC}{\bullet}{\blue{\str{\cotc aed}\emptyset},\str{\tc{\bar a}ed}\emptyset}{\vlhy{\vdots}}
		}
	}
	}
	}
	}
	}
	$
	\end{adjustbox}

    \ 

	\noindent This is an example of the more general `rule permutations' available in $\htcl$, hinting at a more flexible proof theory (we discuss this further in Sec.~\ref{sec:conclusions}).
% 	Such consideration, however, is beyond the scope of this work.
\end{example}

\begin{example}
	[Transitivity]
	\label{ex:transitivity}
	$\TC$ can be proved transitive by way of a cyclic proof in $\ltcl$ of the sequent $\cotc acd ,\cotc ade,\tc{\bar a}ce$.
	As in the previous example we may mimic that proof line by line, but we give a slightly different one that cannot directly be interpreted as a $\ltcl$ proof:
	\renewcommand{\storageone}{Ex.~\ref{fixed-point-identity}}

% \begin{adjustbox}{max width = \textwidth}
% $
% 	\vlderivation{
% 	\vlin{\TC}{\circ}{ \red{\cotc acd}  , {\cotc ade}  , {\tc{\bar a}ce}  }{
% 	\vlin{\coTC}{}{ \red{\cotc acd}  , {\cotc ade}  ,  {\bar a(c,e)}  ,\str{\bar a(c,x),\tc{\bar a}xe}x}{
% 		\vliin{2\cup}{}{ { a(c,d),a(c,c')}  , \red{a(c,d),\cotc a{c'}d}  , {\cotc ade}  ,  {\bar a(c,e)}  ,\str{\bar a(c,x),\tc{\bar a}xe}x}{
% 			\vlin{\instrule}{}{ { a(c,d)}  ,  {\cotc ade}  ,\str{\bar a(c,x),\tc{\bar a}xe}x}{
% 				\vliin{\cup}{}{ { a(c,d)}  ,  {\cotc ade}  , {\bar a(c,d),\tc{\bar a}de}  }{
% 					\vlin{\reset}{}{ {a(c,d)}  , {\bar a(c,d)}  }{
% 						\vlin{\init}{}{ \str{\, }\emptyset  }{\vlhy{}}
% 					}
% 				% \vlhy{\star}
% 				}{
% 					\vliq{}{}{ {\tc{\bar a}de}  , {\cotc ade}  }{\vlhy{\text{\storageone}}}
% 				}
% 			}
% 		}{
% 			\vlin{\instrule}{}{ {a(c,c')}  , \red{\cotc a{c'}d}  ,  {\cotc ade}  ,  {\bar a(c,e)}  ,\str{\bar a(c,x),\tc{\bar a}xe}x}{
% 			\vliin{\cup}{}{ {a(c,c')}  , \red{\cotc a{c'}d}  ,  {\cotc ade}  ,  {\bar a(c,e)}  , {\bar a(c,c'),\tc{\bar a}{c'}e}  }{
% 				\vlin{\reset}{}{ {a(c,c')}  ,  {\bar a(c,c')}  }{
% 				\vlin{\init}{}{ \str{\, }\emptyset  }{\vlhy{}}
% 				}
% 				% \vlhy{\star}
% 			}{
% 				\vlin{\TC}{\circ}{ \red{\cotc a{c'}d}  , {\cotc ade}  , {\tc{\bar a}{c'}e}   }{\vlhy{\vdots}}
% 			}
% 			}
% 		}
% 	}
% 	}
% 	}
% $
% \end{adjustbox}
\begin{adjustbox}{max width = \textwidth}
$
	\vlderivation{
	\vlin{\TC}{\circ}{ \red{\cotc acd}  , {\cotc ade}  , {\tc{\bar a}ce}  }{
	\vlin{\coTC}{}{ \red{\cotc acd}  , {\cotc ade}  ,  {\bar a(c,e)}  ,\str{\bar a(c,x),\tc{\bar a}xe}x}{
		\vliin{2\cup}{}{ { a(c,d),a(c,c')}  , \red{a(c,d),\cotc a{c'}d}  , {\cotc ade}  ,  {\bar a(c,e)}  ,\str{\bar a(c,x),\tc{\bar a}xe}x}{
			\vlin{\instrule}{}{ { a(c,d)}  ,  {\cotc ade}  ,\str{\bar a(c,x),\tc{\bar a}xe}x}{
				\vliin{\cup}{}{ { a(c,d)}  ,  {\cotc ade}  , {\bar a(c,d),\tc{\bar a}de}  }{
					\vlin{\reset}{}{ {a(c,d)}  , {\bar a(c,d)}  }{
						\vlin{\init}{}{ \str{\, }\emptyset  }{\vlhy{}}
					}
				% \vlhy{\star}
				}{
					\vliq{}{}{ {\tc{\bar a}de}  , {\cotc ade}  }{\vlhy{\text{\storageone}}}
				}
			}
		}{
			\vlin{\instrule}{}{ {a(c,c')}  , \red{\cotc a{c'}d}  ,  {\cotc ade}  ,  {\bar a(c,e)}  ,\str{\bar a(c,x),\tc{\bar a}xe}x}{
			\vliin{\cup}{}{ {a(c,c')}  , \red{\cotc a{c'}d}  ,  {\cotc ade}  ,  {\bar a(c,e)}  , {\bar a(c,c'),\tc{\bar a}{c'}e}  }{
				\vlin{\reset}{}{ {a(c,c')}  ,  {\bar a(c,c')}  }{
				\vlin{\init}{}{ \str{\, }\emptyset  }{\vlhy{}}
				}
				% \vlhy{\star}
			}{
				\vlin{\TC}{\circ}{ \red{\cotc a{c'}d}  , {\cotc ade}  , {\tc{\bar a}{c'}e}   }{\vlhy{\vdots}}
			}
			}
		}
	}
	}
	}
$
\end{adjustbox}

\

\noindent 	The only infinite branch (except for that from Ex.~\ref{fixed-point-identity}), looping on $\circ$, is progressing by the red 
%hypertrace.
trace. 
\end{example}

    \begin{example}
    \label{ex_htcl_appendix_1}
    We show a cyclic proof of the following hypersequent:
    \begin{equation*}
    \str{\cotc{\nf{\alpha}}{c}{d}}{\varnothing}, \str{\tc{a}{c}{d}}{\varnothing}, 
    \str{\tc{\beta}{c}{d}}{\varnothing},
    \str{\tc{a}{c}{y}, \tc{\beta}{y}{d}}{y}
    \end{equation*}
    where $\aredf{c}{d}   = \st{c, d}{a  a \union a  b  a}$ and $\gamma(c,d)  =  \st{c, d}{b  \plus{a}}$. The progressing hypertraces of the two infinite branches of the proof are highlighted in red. 
    
    \begin{adjustbox}{max width = \textwidth}
    $ 
    \vlderivation{
    \vlin{\cotcrule}{\circ}{ 
    \str{\cotc{\nf{\alpha}}{c}{d}}{\varnothing}, \red{\str{\tc{a}{c}{d}}{\varnothing}}, 
    \str{\tc{\gamma}{c}{d}}{\varnothing},
    \str{\tc{a}{c}{y}, \tc{\gamma}{y}{d}}{y} }{
    \vliin{\cup}{}{
    \str{\nf{\alpha}(c,d),\nf{\alpha}(c,e)}{\varnothing},
    \red{\str{\nf{\alpha}(c,d),\cotc{\nf{\alpha}}{e}{d}}{\varnothing}}, \str{\tc{a}{c}{d}}{\varnothing}, 
    \str{\tc{\gamma}{c}{d}}{\varnothing},
    \str{\tc{a}{c}{y}, \tc{\gamma}{y}{d}}{y}
    }{
    \vlhy{\qq_1}
    }{
   \vliin{\cup}{}{
   \str{\nf{\alpha}(c,e)}{\varnothing},
   \red{\str{\nf{\alpha}(c,d),\cotc{\nf{\alpha}}{e}{d}}{\varnothing}},  \str{\tc{a}{c}{d}}{\varnothing}, 
   \str{\tc{\gamma}{c}{d}}{\varnothing},
   \str{\tc{a}{c}{y}, \tc{\gamma}{y}{d}}{y}}{
    \vlhy{\qq_2}}{
    \vlid{\vlan}{}{
    \str{\nf{\alpha}(c,e)}{\varnothing},
    \red{\str{\cotc{\nf{\alpha}}{e}{d}}{\varnothing}},  \str{\tc{a}{c}{d}}{\varnothing}, 
    \str{\tc{\gamma}{c}{d}}{\varnothing},
    \str{\tc{a}{c}{y}, \tc{\gamma}{y}{d}}{y}
    }{
    \vliin{\cup}{}{
    \str{\nf{aa}(c,e),\nf{aba}(c,e)}{\varnothing},
\red{\str{\cotc{\nf{\alpha}}{e}{d}}{\varnothing}},   \str{\tc{a}{c}{d}}{\varnothing}, 
    \str{\tc{\gamma}{c}{d}}{\varnothing},
    \str{\tc{a}{c}{y}, \tc{\gamma}{y}{d}}{y}
    }{
    \vlhy{\qq_3}
    }{
    \vlin{\wk}{}{
    \str{\nf{aa}(c,e)}{\varnothing},
    \red{\str{\cotc{\nf{\alpha}}{e}{d}}{\varnothing}},   \str{\tc{a}{c}{d}}{\varnothing}, 
    \str{\tc{\gamma}{c}{d}}{\varnothing},
    \str{\tc{a}{c}{y}, \tc{\gamma}{y}{d}}{y}
    }{
    \vlid{\forall, \vlor}{}{
    \str{\nf{aa}(c,e)}{\varnothing},
   \red{\str{\cotc{\nf{\alpha}}{e}{d}}{\varnothing}},  \str{\tc{a}{c}{d}}{\varnothing}, 
    %\str{\tc{\gamma}{c}{d}}{\varnothing},
    \str{\tc{a}{c}{y}, \tc{\gamma}{y}{d}}{y}
    }{
    \vlin{\tcrule,\wk}{}{
    \str{\nf{a}(c,e_1)}{\varnothing},
    \str{\nf{a}(e_1,e)}{\varnothing},
    \red{\str{\cotc{\nf{\alpha}}{e}{d}}{\varnothing}},  \str{\tc{a}{c}{d}}{\varnothing}, 
    %\str{\tc{\gamma}{c}{d}}{\varnothing},
    \str{\tc{a}{c}{y}, \tc{\gamma}{y}{d}}{y}
    }{
    \vlin{\tcrule , \wk}{}{
    \str{\nf{a}(c,e_1)}{\varnothing},
    \str{\nf{a}(e_1,e)}{\varnothing},
    \red{\str{\cotc{\nf{\alpha}}{e}{d}}{\varnothing}},  \str{a(c,z),\tc{a}{z}{d}}{z}, 
   % \str{\tc{\gamma}{c}{d}}{\varnothing},
    \str{\tc{a}{c}{y}, \tc{\gamma}{y}{d}}{y}
    }{
    \vlin{\instrule , \id}{}{
    \str{\nf{a}(c,e_1)}{\varnothing},
    \str{\nf{a}(e_1,e)}{\varnothing},
    \red{\str{\cotc{\nf{\alpha}}{e}{d}}{\varnothing}},  \str{a(c,z),\tc{a}{z}{d}}{z}, 
  % \str{\tc{\gamma}{c}{d}}{\varnothing},
    \str{a(c,z_1),\tc{a}{z_1}{y}, \tc{\gamma}{y}{d}}{y,z_1}
    }{
    \vlin{\tcrule , \wk}{}{
    % \str{\nf{a}(c,e_1)}{\varnothing},
    \str{\nf{a}(e_1,e)}{\varnothing},
    \red{\str{\cotc{\nf{\alpha}}{e}{d}}{\varnothing}},  \str{\tc{a}{e_1}{d}}{\varnothing}, 
  % \str{\tc{\gamma}{c}{d}}{\varnothing},
    \str{\tc{a}{e_1}{y}, \tc{\gamma}{y}{d}}{y}
    }{
    \vlin{\tcrule}{}{
    % \str{\nf{a}(c,e_1)}{\varnothing},
    \str{\nf{a}(e_1,e)}{\varnothing},
    \red{\str{\cotc{\nf{\alpha}}{e}{d}}{\varnothing}},  \str{a(e_1,z_2),\tc{a}{z_2}{d}}{z_2}, 
  % \str{\tc{\gamma}{c}{d}}{\varnothing},
    \str{\tc{a}{e_1}{y}, \tc{\gamma}{y}{d}}{y}
    }{
    \vlin{\instrule, \id}{}{
   %  \str{\nf{a}(c,e_1)}{\varnothing},
    \str{\nf{a}(e_1,e)}{\varnothing},
    \red{\str{\cotc{\nf{\alpha}}{e}{d}}{\varnothing}},  \str{a(e_1,z_2),\tc{a}{z_2}{d}}{z_2}, 
  % \str{\tc{\gamma}{c}{d}}{\varnothing},
    \str{a(e_1,y), \tc{\gamma}{y}{d}}{y},
    \str{a(e_1,z_3),\tc{a}{z_3}{y}, \tc{\gamma}{y}{d}}{y,z_3}
    }{
    \vlin{}{\circ}{
     % \str{\nf{a}(c,e_1)}{\varnothing},
%    \str{\nf{a}(e_1,e)}{\varnothing},
    \red{\str{\cotc{\nf{\alpha}}{e}{d}}{\varnothing}},  \str{\tc{a}{e}{d}}{\varnothing}, 
  % \str{\tc{\gamma}{c}{d}}{\varnothing},
    \str{ \tc{\gamma}{e}{d}}{\varnothing},
    \str{\tc{a}{e}{y}, \tc{\gamma}{y}{d}}{y}
    }{\vlhy{\vdots}}
    }
    }
    }
    }
   }
    }
    }
    }
    }
    }
   }
   }
   }
   }
    $
\end{adjustbox}
    
    \
    
    We do not show the finite derivations of hypersequents $\qq_1$ and $\qq_2$; here follows the infinite derivation branch starting at $\qq_3$.
    
    \
     
    \begin{adjustbox}{max width = \textwidth}
    \begin{tabular}{c c l}
    $\qq_1$ & = &  $ \str{\nf{\alpha}(c,d)}{\varnothing},
    \str{\nf{\alpha}(c,d),\cotc{\nf{\alpha}}{e}{d}}{\varnothing}, \str{\tc{a}{c}{d}}{\varnothing}, 
    \str{\tc{\gamma}{c}{d}}{\varnothing},
    \str{\tc{a}{c}{y}, \tc{\gamma}{y}{d}}{y}$ \\[0.3cm]
    $\qq_2$ & = & $\str{\nf{\alpha}(c,e)}{\varnothing},
    \str{\nf{\alpha}(c,e)}{\varnothing},  \str{\tc{a}{c}{d}}{\varnothing}, 
    \str{\tc{\gamma}{c}{d}}{\varnothing},
    \str{\tc{a}{c}{y}, \tc{\gamma}{y}{d}}{y}$\\[0.3cm]
    $\qq_3$ & = & $\str{\nf{aba}(c,e)}{\varnothing},
    \str{\cotc{\nf{\alpha}}{e}{d}}{\varnothing},  \str{\tc{a}{c}{d}}{\varnothing}, 
    \str{\tc{\gamma}{c}{d}}{\varnothing},
    \str{\tc{a}{c}{y}, \tc{\gamma}{y}{d}}{y}$\\
    \end{tabular}
    \end{adjustbox}

    \ 
    
    \begin{adjustbox}{max width = \textwidth}
    $
    \vlderivation{
    \vlin{}{}{
    \str{\nf{aba}(c,e)}{\varnothing},
    \str{\cotc{\nf{\alpha}}{e}{d}}{\varnothing},  \red{\str{\cotc{\nf{\alpha}}{e}{d}}{\varnothing}}, 
    \str{\tc{\gamma}{c}{d}}{\varnothing},
    \str{\tc{a}{c}{y}, \tc{\gamma}{y}{d}}{y}
    }{
    \vlin{\wk}{}{
    \str{\nf{aba}(c,e)}{\varnothing},
   \red{\str{\cotc{\nf{\alpha}}{e}{d}}{\varnothing}},  %\str{\tc{a}{c}{d}}{\varnothing}, 
    %\str{\tc{\gamma}{c}{d}}{\varnothing},
    \str{\tc{a}{c}{y}, \tc{\gamma}{y}{d}}{y}
    }{
    \vlid{\forall, \vlor}{}{
    \str{\nf{aba}(c,d)}{\varnothing},
   \red{\str{\cotc{\nf{\alpha}}{e}{d}}{\varnothing}},  %\str{\tc{a}{c}{d}}{\varnothing}, 
    %\str{\tc{\gamma}{c}{d}}{\varnothing},
    \str{\tc{a}{c}{y}, \tc{\gamma}{y}{d}}{y}
    }{
    \vlid{\forall, \vlor}{}{
     \str{\nf{a}(c,e_1)}{\varnothing},
    \str{\nf{ba}(e_1,e)}{\varnothing},
    \red{\str{\cotc{\nf{\alpha}}{e}{d}}{\varnothing}},  %\str{\tc{a}{c}{d}}{\varnothing}, 
    %\str{\tc{\gamma}{c}{d}}{\varnothing},
    \str{\tc{a}{c}{y}, \tc{\gamma}{y}{d}}{y}
    }{
    \vlin{}{}{
      \str{\nf{a}(c,e_1)}{\varnothing},
    \str{\nf{b}(e_1,e_2)}{\varnothing},
    \str{\nf{ba}(e_2,e)}{\varnothing},
    \red{\str{\cotc{\nf{\alpha}}{e}{d}}{\varnothing}},  %\str{\tc{a}{c}{d}}{\varnothing}, 
    %\str{\tc{\gamma}{c}{d}}{\varnothing},
    \str{\tc{a}{c}{y}, \tc{\gamma}{y}{d}}{y}
    }{
    \vlin{\tcrule,\wk}{}{
    \str{\nf{a}(c,e_1)}{\varnothing},
    \str{\nf{b}(e_1,e_2)}{\varnothing},
    \str{\nf{a}(e_2,e)}{\varnothing},
    \red{\str{\cotc{\nf{\alpha}}{e}{d}}{\varnothing}},  %\str{\tc{a}{c}{d}}{\varnothing}, 
    %\str{\tc{\gamma}{c}{d}}{\varnothing},
    \str{\tc{a}{c}{y}, \tc{\gamma}{y}{d}}{y}
    }{
    \vlin{\instrule,\id}{}{
    \str{\nf{a}(c,e_1)}{\varnothing},
    \str{\nf{b}(e_1,e_2)}{\varnothing},
    \str{\nf{a}(e_2,e)}{\varnothing},
    \red{\str{\cotc{\nf{\alpha}}{e}{d}}{\varnothing}},  %\str{\tc{a}{c}{d}}{\varnothing}, 
    %\str{\tc{\gamma}{c}{d}}{\varnothing},
    \str{a(c,y), \tc{\gamma}{y}{d}}{y}
    }{
    \vlin{\tcrule}{}{
    %    \str{\nf{a}(c,e_1)}{\varnothing},
    \str{\nf{b}(e_1,e_2)}{\varnothing},
    \str{\nf{a}(e_2,e)}{\varnothing},
    \red{\str{\cotc{\nf{\alpha}}{e}{d}}{\varnothing}},  %\str{\tc{a}{c}{d}}{\varnothing}, 
    %\str{\tc{\gamma}{c}{d}}{\varnothing},
    \str{ \tc{\gamma}{y}{d}}{y}
    }{
    \vlid{\exists, \vlan}{}{
    %    \str{\nf{a}(c,e_1)}{\varnothing},
    \str{\nf{b}(e_1,e_2)}{\varnothing},
    \str{\nf{a}(e_2,e)}{\varnothing},
    \red{\str{\cotc{\nf{\alpha}}{e}{d}}{\varnothing}},  %\str{\tc{a}{c}{d}}{\varnothing}, 
    %\str{\tc{\gamma}{c}{d}}{\varnothing},
    \str{\st{y, d}{ba^+}}{y}
    \str{\st{y, z}{ba^+} ,\tc{\gamma}{z}{d}}{y,z}
    }{
    \vlin{\instrule , \id}{}{
     %    \str{\nf{a}(c,e_1)}{\varnothing},
    \str{\nf{b}(e_1,e_2)}{\varnothing},
    \str{\nf{a}(e_2,e)}{\varnothing},
   \red{\str{\cotc{\nf{\alpha}}{e}{d}}{\varnothing}},  %\str{\tc{a}{c}{d}}{\varnothing}, 
    %\str{\tc{\gamma}{c}{d}}{\varnothing},
    \str{b (y, y_1), \tc{a}{y_1}{d} }{y,y_1}
    \str{ b (y, z_1), \tc{a}{z_1}{z},\tc{\gamma}{z}{d}}{y,z,z_1}
    }{
    \vlin{\tcrule ,\wk}{}{
     %    \str{\nf{a}(c,e_1)}{\varnothing},
%    \str{\nf{b}(e_1,e_2)}{\varnothing},
    \str{\nf{a}(e_2,e)}{\varnothing},
    \red{\str{\cotc{\nf{\alpha}}{e}{d}}{\varnothing}},  %\str{\tc{a}{c}{d}}{\varnothing}, 
    %\str{\tc{\gamma}{c}{d}}{\varnothing},
    \str{ \tc{a}{e_2}{d} }{\varnothing}
    \str{  \tc{a}{e_2}{z},\tc{\gamma}{z}{d}}{z}
    }{
    \vlin{\tcrule}{}{
         %    \str{\nf{a}(c,e_1)}{\varnothing},
%    \str{\nf{b}(e_1,e_2)}{\varnothing},
    \str{\nf{a}(e_2,e)}{\varnothing},
    \red{\str{\cotc{\nf{\alpha}}{e}{d}}{\varnothing}},  %\str{\tc{a}{c}{d}}{\varnothing}, 
    %\str{\tc{\gamma}{c}{d}}{\varnothing},
    \str{ a(e_2, w)\tc{a}{w}{d} }{w}
    \str{  \tc{a}{e_2}{z},\tc{\gamma}{z}{d}}{z}
    }{
    \vlin{\instrule,\id}{}{
         %    \str{\nf{a}(c,e_1)}{\varnothing},
%    \str{\nf{b}(e_1,e_2)}{\varnothing},
    \str{\nf{a}(e_2,e)}{\varnothing},
   \red{\str{\cotc{\nf{\alpha}}{e}{d}}{\varnothing}},  %\str{\tc{a}{c}{d}}{\varnothing}, 
    %\str{\tc{\gamma}{c}{d}}{\varnothing},
    \str{ a(e_2, w)\tc{a}{w}{d} }{w}
    \str{ a(e_2, z) ,\tc{\gamma}{z}{d}}{z},
    \str{ a(e_2, k), \tc{a}{k}{z},\tc{\gamma}{z}{d}}{z,k}
    }{
    \vlin{}{\circ}{
           %    \str{\nf{a}(c,e_1)}{\varnothing},
%    \str{\nf{b}(e_1,e_2)}{\varnothing},
 %   \str{\nf{a}(e_2,e)}{\varnothing},
   \red{\str{\cotc{\nf{\alpha}}{e}{d}}{\varnothing}},  %\str{\tc{a}{c}{d}}{\varnothing}, 
    %\str{\tc{\gamma}{c}{d}}{\varnothing},
    \str{ \tc{a}{e}{d} }{\varnothing}
    \str{\tc{\gamma}{e}{d}}{\varnothing},
    \str{  \tc{a}{e}{z},\tc{\gamma}{z}{d}}{z}
    }{\vlhy{\vdots}}
    }
    }
    }
    }
    }
    }
    }
    }
    }
    }
    }
    }
    }
    }
    $
    \end{adjustbox}
    \end{example}

Finally, it is pertinent to revisit the `counterexample' \eqref{eq:counterexample-formula-pdltr} that witnessed incompleteness of $\ltcl$ for $\PDLtr$.
The following result is, in fact, already implied by our later completeness result, Thm.~\ref{thm:completeness-htcl-pdltr}, but we shall present it nonetheless:

\begin{proposition}
\label{prop:counterexample-in-htcl}
$\htcl \dercyc \st{c,d}{(aa \cup aba)^+} \vljm \st{c,d}{a^+((ba^+)^+\cup a)}$.
\end{proposition}
\begin{proof}
We use the following abbreviations: $\aredf{c}{d}   =  \st{c, d}{a  a \union a  b  a}$ 
%$\aredlf{c}{d}   =  \st{c, d}{a  a  }$, $ \aredrf{c}{d}  =  \st{c, d}{a  b  a}$ 
and $\bbluef{c}{d}   = \st{c, d}{(b  \plus{a})^+ \cup a}$. The progressing hypertrace is marked in blue.

\begin{adjustbox}{max width = \textwidth}
$
\vlderivation{
	\vlid{\vlor, \exists, \vlan}{}{ 
	\blue{\str{\nf{\st{c,d}{(aa \cup aba)^+}} \vlor  \st{c, d}{a^+ ;((ba^+)^+ \cup a )} }{\varnothing}} }{
	\vlin{\cotcrule}{\bullet}{ 
	\blue{\str{ \cotc{\coared}{ c}{ d} }{\varnothing}}, \str{ \tc{{a}}{ c}{  y}, \bblue({y}, {d}) }{y} }{
	\vliin{\cup}{}{  
	\str{\coaredf{c}{d}, \coaredf{c}{e} }{\varnothing}, \blue{\str{\coaredf{c}{d}, \cotc{\coared}{ e}{ d} }{\varnothing}}, \str{ \tc{{a}}{ c}{ y}, \bblue({y}, {d}) }{y}  }{
	\vlhy{\mathbf{R}}}
	{
	\vliin{\cup}{}{ 
	\str{\coaredf{c}{e} }{\varnothing}, \blue{\str{\coaredf{c}{d}, \cotc{\coared}{ e}{ d} }{\varnothing}}, \str{ \tc{{a}}{\c}{  y}, \bblue({y}, {d})  }{y} }{
	\vlhy{\mathbf{R}'}
	}{
	\vlid{\vlan}{}{ \str{\coaredf{c}{e} }{\varnothing}, 
	\blue{\str{ \cotc{\coared}{e}{ d} }{\varnothing}}, \str{ \tc{{a}}{ c}{  y}, \bblue({y}, {d}) }{y} }{
	\vliin{\cup}{}{\str{\nf{aa}(c,e), \nf{aba}(c,e) }{\varnothing}, \blue{\str{ \cotc{\coared}{ e}{ d} }{\varnothing}}, \str{ \tc{{a}}{ c}{ y}, \bblue({y}, {d})  }{y}  }{
	\vlhy{\mathbf{P}}
	}{
	\vlid{\forall, \vlor}{}{\str{\nf{aa}(c,e) }{\varnothing}, 
	\blue{\str{ \cotc{\coared}{ e}{ d} }{\varnothing}}, \str{ \tc{{a}}{ c}{  y}, \bblue({y},{d})  }{y}  }{
	%\vlid{\forall, \vlor}{}{ \str{\red{\nf{a}}(\red c, \red f) }{\varnothing}, \str{\red{\nf{a\comp a}}(\red f, \red e) }{\varnothing}, \str{ \cotc{\coared}{\red e}{\red d} }{\varnothing}, \str{ \tc{\blue{a}}{\blue c}{ \blue y}, \bblue(\blue{y}, \blue{d})  }{y} }{
	\vlin{\tcrule}{}{ \str{ {\nf{a}}({c}, {f}) }{\varnothing}, \str{{\nf{a}}( f,  e) }{\varnothing}, \blue{\str{ \cotc{\coared}{ e}{ d} }{\varnothing}}, \str{ \tc{{a}}{ c}{ y}, \bblue({y}, {d})  }{y}   }{
	\vlin{\wk}{}{ \str{ {\nf{a}}({c}, {f}) }{\varnothing}, \str{{\nf{a}}( f,  e) }{\varnothing}, \blue{ \str{ \cotc{\coared}{ e}{ d} }{\varnothing}},   \str{\abluef{c}{y}, \bblue({y}, {d})  }{y},  \str{ \abluef{c}{z}, \tc{{a}}{ z}{  y},  \bblue({y}, {d}) }{y,z} }{
	\vlin{\instrule}{[f/z]}{\str{ {\nf{a}}({c}, {f}) }{\varnothing}, \str{{\nf{a}}( f,  e) }{\varnothing},  \blue{\str{ \cotc{\coared}{ e}{ d} }{\varnothing}},    \str{ \abluef{c}{z}, \tc{{a}}{ z}{  y}, \bblue({y}, {d}) }{y,z} }{
	\vlin{\id}{}{ \str{ {\nf{a}}({c}, {f}) }{\varnothing}, \str{{\nf{a}}( f,  e) }{\varnothing},  \blue{\str{ \cotc{\coared}{ e}{ d} }{\varnothing}},    \str{ \abluef{c}{f}, \tc{{a}}{ f}{  y}, \bblue({y}, {d})  }{y} }{
	\vlin{\tcrule}{}{ \str{{\nf{a}}( f,  e) }{\varnothing},  \blue{\str{ \cotc{\coared}{ e}{ d} }{\varnothing}},    \str{\tc{{a}}{ f}{  y}, \bblue({y}, {d}) }{y} }{
	\vlin{\wk}{}{  \str{{\nf{a}}( f,  e) }{\varnothing},  \blue{\str{ \cotc{\coared}{ e}{\ d} }{\varnothing}}, \str{\abluef{ f}{  y}, \bblue({y}, {d}) }{y},  \str{\abluef{f}{k},\tc{{a}}{ k}{  y}, \bblue({y}, {d}) }{y,k} }{
	\vlin{\instrule}{[e/k]}{ \str{{\nf{a}}( f,  e) }{\varnothing}, \blue{ \str{ \cotc{\coared}{ e}{ d} }{\varnothing}},  \str{\abluef{f}{k},\tc{{a}}{ k}{  y}, \bblue({y}, {d}) }{y,k} }{
	\vlin{\id}{}{\str{{\nf{a}}( f,  e) }{\varnothing},  \blue{\str{ \cotc{\coared}{ e}{ d} }{\varnothing}},  \str{\abluef{f}{g},\tc{{a}}{ g}{  y}, \bblue({y}, {d}) }{y}  }{
	%
	%\vlin{\tc}{}{  \str{\red{\nf{a}}(\red{g}, \red{e}) }{\varnothing}, \str{ \cotc{\coared}{\red e}{\red d} }{\varnothing},  \str{\tc{\blue{a}}{\blue g}{ \blue y}, \bblue(\blue{y}, \blue{d}) }{y}  }{
	%
	%\vlin{\wk}{}{\str{\red{\nf{a}}(\red{g}, \red{e}) }{\varnothing}, \str{ \cotc{\coared}{\red e}{\red d} }{\varnothing}, \str{\abluef{\blue g}{ \blue y}, \tc{\bblue}{\blue y}{\blue d}}{y} , \str{ \abluef{g}{j}, \tc{\blue{a}}{\blue j}{ \blue y}, \tc{\bblue}{\blue y}{\blue d}}{y,j}   }{
	%
	%\vlin{\instrule}{[e/j]}{ \str{\red{\nf{a}}(\red{g}, \red{e}) }{\varnothing}, \str{ \cotc{\coared}{\red e}{\red d} }{\varnothing}, \str{ \abluef{g}{j} ,\tc{\blue{a}}{\blue j}{ \blue y}, \tc{\bblue}{\blue y}{\blue d}}{y,j} }{
	%
	%\vlin{\id}{}{\str{\red{\nf{a}}(\red{g}, \red{e}) }{\varnothing}, \str{ \cotc{\coared}{\red e}{\red d} }{\varnothing}, \str{ \abluef{g}{e} \tc{\blue{a}}{\blue e}{ \blue y}, \tc{\bblue}{\blue y}{\blue d}}{y} }{
	\vlin{}{\bullet}{ \blue{ \str{ \cotc{\coared}{ e}{ d} }{\varnothing}}, \str{ \tc{{a}}{ j}{  e}, \bblue({ y}{ d})}{y} }{\vlhy{\vdots}}
	}
	}
	}
	}
	}
	}
	}
	}
	}
	}
	}
	}
	}
	}
	}
%	}
%	}
%	}
%	}
%	}
}
$
\end{adjustbox}

\

\noindent In the above derivation, $ \mathbf{R} , \mathbf{R}'$ and $\mathbf P $ are the following hypersequents:
% , derivable by means of finite derivation branches (which we do not show):
\begin{eqnarray*}
\mathbf{R} & = & \str{\coaredf{c}{d} }{\varnothing}, \str{\coaredf{c}{d}, \cotc{\coared}{ e}{ d} }{\varnothing}, \str{ \tc{{a}}{ c}{  y}, \bblue({y}, {d}) }{y}\\
\mathbf{R}' & = & \str{ \coaredf{c}{d} }{\varnothing}, \str{\coaredf{c}{d} }{\varnothing}, \str{ \tc{{a}}{ c}{  y}, \bblue({y}, {d}) }{y}\\
\mathbf P & = & \str{\nf{aba}(c,e)}{\varnothing}, \str{ \cotc{\coared}{ e}{ d} }{\varnothing}, \str{ \tc{{a}}{ c}{  y}, \bblue({y}, {d}) }{y}
\end{eqnarray*}
$\mathbf R,\mathbf R'$ have finitary proofs, while
 $\mathbf{P}$ has a cyclic proof, shown below, %(see App.~\ref{sec:appendix-examples-of-htcl-proofs}). 
 which uses \Cref{ex_htcl_appendix_1} above. %, we show a cyclic proof of the following hypersequent:
% 	$$ 
% 	\mathbf{P} = \str{ \nf{aba}(c,e) }{\varnothing}, \str{ \cotc{\coared}{ e}{ d} }{\varnothing}, \str{ \tc{{a}}{ c}{  y}, \bblue({y}, {d}) }{y}  
% 	$$
% 	In the derivation below, $\qq =  \str{\coaredrf{c}{e} }{\varnothing}, \str{ \cotc{\coared}{ e}{ d} }{\varnothing}$. The progressing hypertrace follows $\str{ \cotc{\coared}{ e}{ d} }{\varnothing}$. 
	We use the following abbreviations:
	\[
	\begin{array}{r c l}
		 \aredf{c}{d}  &  = & \st{c, d}{a  a \union a  b  a} \\
		 %\aredlf{c}{d} &  = & \st{c, d}{a \comp a  } \\
		 %\aredrf{c}{d} & = & \st{c, d}{a \comp b \comp a} \\ 
		 \bbluef{c}{d} &  = & \st{c, d}{(b  \plus{a})^+ \cup a}\\
		 \gamma(c,d) &  = & \st{c, d}{b  \plus{a}}.
	\end{array}
	\]
	
	\begin{adjustbox}{max width = \textwidth}
	$
	\vlderivation{
	\vlid{\vlor}{}{
	\str{ \nf{aba}(c,e) }{\varnothing}, 
	\str{ \cotc{\coared}{ e}{ d} }{\varnothing}, 
	\str{ \tc{{a}}{ c}{  y}, \bblue({y}, {d}) }{y} 
	}{
	\vlin{\wk}{}{
	\str{ \nf{aba}(c,e) }{\varnothing}, 
	\str{ \cotc{\coared}{ e}{ d} }{\varnothing}, 
	\str{ \tc{{a}}{ c}{  y}, \tc{\gamma}{y}{d} }{y} ,
	\str{ \tc{{a}}{ c}{  y}, a(y,d) }{y} 
	}{
	\vlid{\forall, \vlor}{}{
	\str{ \nf{aba}(c,e) }{\varnothing}, 
	\str{ \cotc{\coared}{ e}{ d} }{\varnothing}, 
	\str{ \tc{{a}}{ c}{  y}, \tc{\gamma}{y}{d} }{y} 
	}{
	\vlid{\forall, \vlor}{}{
	\str{ \nf{a}(c,e_1) }{\varnothing}, 
	\str{ \nf{ba}(e_1,e) }{\varnothing}, 
	\str{ \cotc{\coared}{ e}{ d} }{\varnothing}, 
	\str{ \tc{{a}}{ c}{  y}, \tc{\gamma}{y}{d} }{y} 
	}{
	\vlin{\tcrule, \wk}{}{
	\str{ \nf{a}(c,e_1) }{\varnothing}, 
	\str{ \nf{b}(e_1,e_2) }{\varnothing}, 
	\str{ \nf{a}(e_2,e) }{\varnothing}, 
	\str{ \cotc{\coared}{ e}{ d} }{\varnothing}, 
	\str{ \tc{{a}}{ c}{  y}, \tc{\gamma}{y}{d} }{y} 
	}{
	\vlin{\instrule, \id}{}{
	\str{ \nf{a}(c,e_1) }{\varnothing}, 
	\str{ \nf{b}(e_1,e_2) }{\varnothing}, 
	\str{ \nf{a}(e_2,e) }{\varnothing}, 
	\str{ \cotc{\coared}{ e}{ d} }{\varnothing}, 
	\str{ a(c,y), \tc{\gamma}{y}{d} }{y} 
	}{
	\vlin{\tcrule}{}{
	%\str{ \nf{a}(c,e_1) }{\varnothing}, 
	\str{ \nf{b}(e_1,e_2) }{\varnothing}, 
	\str{ \nf{a}(e_2,e) }{\varnothing}, 
	\str{ \cotc{\coared}{ e}{ d} }{\varnothing}, 
	\str{ \tc{\gamma}{e_1}{d} }{\varnothing} 
	}{
	\vlid{\exists, \vlan}{}{
	%\str{ \nf{a}(c,e_1) }{\varnothing}, 
	\str{ \nf{b}(e_1,e_2) }{\varnothing}, 
	\str{ \nf{a}(e_2,e) }{\varnothing}, 
	\str{ \cotc{\coared}{ e}{ d} }{\varnothing},
	\str{ \st{e_1, d}{ba^+} }{\varnothing}, 
	\str{\st{e_1, z}{ba^+}, \tc{\gamma}{z}{d} }{z} 
	}{
	\vlin{\instrule, \id}{}{
	%\str{ \nf{a}(c,e_1) }{\varnothing}, 
	\str{ \nf{b}(e_1,e_2) }{\varnothing}, 
	\str{ \nf{a}(e_2,e) }{\varnothing}, 
	\str{ \cotc{\coared}{ e}{ d} }{\varnothing},
	\str{b(e_1,k), \tc{a}{k}{d}  }{k}, 
	\str{b(e_1,z_1), \tc{a}{z_1}{z}, \tc{\gamma}{z}{d} }{z,z_1} 
	}{
	\vlin{\tcrule}{}{
	%\str{ \nf{a}(c,e_1) }{\varnothing}, 
	%\str{ \nf{b}(e_1,e_2) }{\varnothing}, 
	\str{ \nf{a}(e_2,e) }{\varnothing}, 
	\str{ \cotc{\coared}{ e}{ d} }{\varnothing},
	\str{ \tc{a}{e_2}{d}  }{\varnothing}, 
	\str{ \tc{a}{e_2}{z}, \tc{\gamma}{z}{d} }{z}
	}{
	\vlin{\tcrule,wk}{}{
	%\str{ \nf{a}(c,e_1) }{\varnothing}, 
	%\str{ \nf{b}(e_1,e_2) }{\varnothing}, 
	\str{ \nf{a}(e_2,e) }{\varnothing}, 
	\str{ \cotc{\coared}{ e}{ d} }{\varnothing},
	\str{ \tc{a}{e_2}{d}  }{\varnothing}, 
	\str{a(e_2, z), \tc{\gamma}{z}{d} }{z},
	\str{a(e_2, w), \tc{a}{w}{z}, \tc{\gamma}{z}{d} }{z,w}
	}{
	\vlin{\instrule, \id}{}{
	%\str{ \nf{a}(c,e_1) }{\varnothing}, 
	%\str{ \nf{b}(e_1,e_2) }{\varnothing}, 
	\str{ \nf{a}(e_2,e) }{\varnothing}, 
	\str{ \cotc{\coared}{ e}{ d} }{\varnothing},
	\str{a(e_2, j), \tc{a}{j}{d}  }{j}, 
	\str{a(e_2, z), \tc{\gamma}{z}{d} }{z},
	\str{a(e_2, w), \tc{a}{w}{z}, \tc{\gamma}{z}{d} }{z,w}
	}{
	\vliq{}{}{
	%\str{ \nf{a}(c,e_1) }{\varnothing}, 
	%\str{ \nf{b}(e_1,e_2) }{\varnothing}, 
	%\str{ \nf{a}(e_2,e) }{\varnothing}, 
	\str{ \cotc{\coared}{ e}{ d} }{\varnothing},
	\str{ \tc{a}{e}{d}  }{\varnothing}, 
	\str{ \tc{\gamma}{e}{d} }{\varnothing},
	\str{ \tc{a}{e}{z}, \tc{\gamma}{z}{d} }{z}
	}{\vlhy{\Cref{ex_htcl_appendix_1}}}
	}
	}
	}
	}
	}
	}
	}
	}
	}
	}
	}
	}
	}
	$
	\end{adjustbox}

\end{proof}

\subsection{On cyclic-proof checking}
In usual cyclic systems, checking that a regular preproof is progressing is decidable by straightforward reduction to the universality of nondeterministic $\omega$-word-automata,
% \footnote{The precise acceptance condition depends on the logic at hand. 
% E.g., B\"uchi will do for $\ltcl$, but in general parity/Rabin/Muller is needed when fixed points alternate.} 
with runs `guessing' a progressing thread along an infinite branch.
Our notion of progress exhibits an extra quantifier alternation: we must \emph{guess} an infinite hypertrace in which \emph{every} trace is progressing.
Nonetheless, by appealing to determinisation or alternation, 
we can still decide our progressing condition:
\begin{proposition}
\label{prop:checking-progress-htcl-automata}
    Checking whether a $\htcl$ preproof is a proof is decidable by reduction to universality (or emptiness) of $\omega$-regular languages.
    % (cf.~\cite{MiyanoHayashi83:alt-aut-omega-words}). 
\end{proposition}
\anupam{give proof in appendix? add citation.}

%We show that the progressing condition for cyclic proofs in $\htcl$ is decidable using automata-theoretic techniques:
\begin{proof}[Proof sketch]
%[of Prop.~\ref{prop:checking-progress-htcl-automata}, sketch]
The result is proved using using automata-theoretic techniques. 
Fix a cyclic $\htcl$ preproof $\derd$.
First, using standard methods from cyclic proof theory, it is routine to construct a nondeterministic B\"uchi automaton recognising non-progressing hypertraces of $\derd$.
The construction is similar to that recognising progressing branches in cyclic sequent calculi, e.g.\ as found in \cite{daxhofmannlange06} or \cite{Simpson17,Das20:cyclic-arithmetic}, since we are asking that there \emph{exists} a non-progressing trace within a hypertrace.
By B\"uchi's complementation theorem and McNaughton's determinisation theorem (see, e.g., \cite{thomas1997languages} for details), we can thus construct a deterministic parity automaton $\mathcal P_H$ recognising progressing hypertraces. 
(This is overkill, but allows us to easily deal with the issue of alternation in the progressing condition.)

Now we can construct a nondeterministic parity automaton $\mathcal P$ recognising progressing branches of $\derd$ similarly to the previous construction, but further keeping track of states in $\mathcal P_H$:
\begin{itemize}
    \item $\mathcal P$ essentially guesses a `progressing' hypertrace along the branch input;
    \item at the same time, $\mathcal P$ runs the hypertrace-in-construction along $\mathcal P_H$ and keeps track of the state therein;
    \item acceptance for $\mathcal P$ is inherited directly from $\mathcal P_H$, i.e.\ the hypertrace guessed is accepting for $\mathcal P$ just if it is accepted by $\mathcal P_H$. 
\end{itemize}
Now it is clear that $\mathcal P$ accepts a branch of $\derd$ if and only if it is progressing.
Assuming that $\mathcal P$ also accepts any $\omega$-words over the underlying alphabet that are not branches of $\derd$ (by adding junk states), we have that $\derd$ is a proof (i.e.\ is progressing) if and only if $\mathcal P$ is universal. 
\end{proof}

\subsection{Simulating Cohen-Rowe}
As we mentioned earlier, cyclic proofs of $\htcl$ indeed are at least as expressive as those of Cohen and Rowe's system by a routine local simulation of rules:

\begin{theorem}
% [Simulating Cohen-Rowe]
	\label{simulation-of-cohen-rowe}
	If $\ltcl \dercyc A$ then $\htcl \dercyc A$.
%	We have that:
%	\begin{enumerate}
%		\item if $\ltcl \dercyc A$ then $\htcl \dercyc A$; and,
%		\item if $\lrtcl \dercyc A$ then $\hrtcl \dercyc A$.
%	\end{enumerate}
\end{theorem}

\begin{proof}[Proof sketch]
%	[of Thm.~\ref{simulation-of-cohen-rowe}, sketch]
	Let $\derd$ be a $\ltcl$ cyclic proof. 
	We can convert it to a $\htcl$ cyclic proof by simply replacing each sequent $A_1, \dots, A_n$ by the hypersequent $\str{A_1}\emptyset,\dots, \str{A_n}\emptyset$ and applying some local corrections.
	In what follows, if $\Gamma = A_1, \dots, A_n$, let us simply write $\sq_\Gamma$ for $\str{A_1}\emptyset, \dots, \str{A_n}\emptyset$.
	\anupam{with current notation, we can actually just write $\Gamma $ instead of $\sq_\Gamma$.}
	\begin{itemize}
		\item Any $\id$ step of $\derd$ must be amended as follows:
		\[
		\vlinf{\id}{}{\Gamma, p(t),\bar p(t)}{}
		\quad \leadsto \quad
		\vlderivation{
			\vliq{\wk}{}{\sq_\Gamma, \str{p(t)}\emptyset, \str{\bar p(t)}\emptyset}{
				\vlin{\reset}{}{\str{p(t)}\emptyset, \str{\bar p(t)}\emptyset}{
					\vlin{\init}{}{\str{\, }\emptyset}{\vlhy{}}
				}
			}
		}
		\]
		\item
		Any $\vlor $ step of $\derd$ becomes a correct $\vlor$ step of $\htcl$ or $\hrtcl$
		\item Any $\vlan$ step of $\derd$ must be amended as follows:
		\[
		\vliinf{\vlan}{}{\Gamma, A \vlan B}{\Gamma, A}{\Gamma, B}
		\quad \leadsto\quad
		\vlderivation{
		\vlin{\vlan}{}{\sq_\Gamma, \str{A\vlan B}\emptyset}{
		\vliin{\cup}{}{\sq_\Gamma, \str{A,B}\emptyset}{
			\vlhy{\sq_\Gamma, \str A \emptyset}
		}{
			\vlhy{\sq_\Gamma, \str B \emptyset}
		}
		}
		}
		\]
		\item Any $\exists $ step of $\derd$ must be amended as follows:
		\[
		\vlinf{\exists}{}{\Gamma, \exists x A(x)}{\Gamma, A(t)}
		\quad \leadsto \quad
		\vlderivation{
		\vlin{\exists}{}{\sq_\Gamma, \str{\exists xA(x)}\emptyset}{
		\vlin{\instrule}{}{\sq_\Gamma, \str{A(x)}x}{
		\vlhy{\sq_\Gamma, \str{A(t)}\emptyset}
		}
		}
		}
		\]
		\item Any $\forall$ step of $\derd$ becomes a correct $\forall$ step of $\htcl$ or $\hrtcl$
		\item Any $\TC_0$ step of $\derd$ becomes a correct $\TC_0$ step of $\htcl$.
		\item Any $\TC_1$ step of $\derd$ must be amended as follows:
		\[
		\vliinf{\TC_1}{}{\Gamma, \tc Ast}{\Gamma , A(s,r)}{\Gamma, \tc Art}
		\quad \leadsto \quad
		\vlderivation{
		\vlin{\TC_1}{}{\sq_\Gamma, \str{\tc Ast}\emptyset}{
		\vlin{\instrule}{}{\sq_\Gamma, \str{A(s,x), \tc Axt}x}{
		\vliin{\cup}{}{\sq_\Gamma, \str{A(s,r),\tc Art}\emptyset }{
			\vlhy{\sq_\Gamma, \str{A(s,r)}\emptyset}
		}{
			\vlhy{\sq_\Gamma, \str{\tc Art}\emptyset}
		}
		}
		}
		}
		\]
		
		\item Any $\coTC$ step of $\derd$ must be amended as follows:
		\[
		\begin{array}{rl}
		    & \vliinf{\coTC}{}{\Gamma, \blue{\cotc A st}}{\Gamma, A(s,t)}{\Gamma,A(s,c),\blue{\cotc A ct}} \\
		    \noalign{\medskip}
		\leadsto & 
		\vlderivation{
		\vlin{\coTC}{}{\sq_\Gamma,\blue{\str{\cotc Ast}\emptyset}}{
		\vliin{2\cup}{}{\sq_\Gamma, \str{A(s,t),A(s,c)}\emptyset , \blue{\str{A(s,t), \cotc A ct}\emptyset} }{
		\vlhy{\sq_\Gamma, \str{A(s,t)}\emptyset}
		}{
		\vlhy{\sq_\Gamma, \str{A(s,c)}\emptyset, \blue{\str{\cotc Act}\emptyset} }
		}
		}
		}
		\end{array}
		\]
% 		\item Any $\RTC_0$ step of $\derd$ must be amended as follows:
% 		\[
% 		\vlinf{\RTC_0}{}{\Gamma, \rtc Ass}{}
% 		\quad \leadsto \quad
% 		\vlderivation{
% 		\vliq{\wk}{}{\sq_\Gamma, \str{\rtc Ass}\emptyset}{
% 		\vlin{\RTC_0}{}{\str{\rtc Ass}\emptyset}{
% 		\vlin{=}{}{\str{s=s}\emptyset}{
% 		\vlin{\init}{}{\str{\, }\emptyset}{\vlhy{}}
% 		}
% 		}
% 		}
% 		}
% 		\]
% 		\item Any $\RTC_1$ step of $\derd $ must be amended as follows:
% 		\[
% 		\vliinf{\RTC_1}{}{\Gamma, \rtc Ast}{\Gamma , A(s,r)}{\Gamma, \rtc Art}
% 		\quad \leadsto \quad
% 		\vlderivation{
% 			\vlin{\RTC_1}{}{\sq_\Gamma, \str{\rtc Ast}\emptyset}{
% 				\vlin{\instrule}{}{\sq_\Gamma, \str{A(s,x), \rtc Axt}x}{
% 					\vliin{\cup}{}{\sq_\Gamma, \str{A(s,r),\rtc Art}\emptyset }{
% 						\vlhy{\sq_\Gamma, \str{A(s,r)}\emptyset}
% 					}{
% 						\vlhy{\sq_\Gamma, \str{\rtc Art}\emptyset}
% 					}
% 				}
% 			}
% 		}
% 		\]
% 		\item Any $=$ step of $\derd$ must be amended as follows:
% 		\[
% 		\vlinf{=}{}{\Gamma, t=t}{}
% 		\quad \leadsto \quad
% 		\vlderivation{
% 		\vliq{\wk}{}{\sq_\Gamma, \str{t=t}\emptyset}{
% 		\vlin{=}{}{\str{t=t}\emptyset}{
% 		\vlin{\init}{}{\str{\, }\emptyset}{\vlhy{}}
% 		}
% 		}
% 		}
% 		\]
% 		\item Any $\neq$ step of $\derd $ becomes a correct $\neq$ step of $\hrtcl$. \anupam{amend rules so that this is true.}
	\end{itemize}
Particular inspection of the $\coTC$ case shows that progressing traces of $\ltcl$ induce progressing hypertraces of $\htcl$.\anupam{explain more?}
\end{proof}

\section{Soundness of $\htcl$}
\label{sec:soundness_htcl}

%\todo{Justify that we don't use the substitution rule in the soundness proof. }

This section is devoted to the proof of the first of our main results:

\begin{theorem}[Soundness]
	\label{thm:soundness_tcl}
	If $\htcl \dernwf \sq $ then $  \models  \sq $.
\end{theorem}
%\anupam{7/2: why not state the theorem for non-wellfoudned provability rather than cyclic provability? Then also we can assume substitution does not occur. (Eliminating subsstitution does not preserve regularity).}
% 
%\anupam{write simply $\models \sq$ instead of $\models \fmi \sq$.}
%
The argument is quite technical due to the alternating nature of our progress condition.
In particular the treatment of traces within hypertraces requires a more fine grained argument than usual, bespoke to our hypersequential structure.

Throughout this section, we shall fix a $\htcl$ preproof $\der$ of a hypersequent $\sq$.
We start by introducing some additional definitions and propositions. 

\anupam{added some section structure}

\subsection{Some conventions on (pre)proofs and semantics}
First, we work with proofs without substitution, in order to control the various symbols occurring in a proof.

\begin{proposition}
    If $\htcl \dernwf \sq $ then there is also a $\htcl$ proof of $\sq$ that does not use the substitution rule.
\end{proposition}
\begin{proof}
    [Proof sketch]
    We proceed by a coinductive argument, applying a meta-level substitution operation on proofs to admit each substitution step.
    Productivity of the translation is guaranteed by the progressing condition: each infinite branch must, at the very least, have infinitely many $\coTC$ steps.
\end{proof}
\todo{at some point make this proof formal, can wait for next revision. Also requires proper typing of the substitution rule.}

The utility of this is that we can now carefully control the occurrences of eigenfunctions in a proof so that, bottom-up, they are never `re-introduced', thus facilitating the definition of the interpretation $\deltah$ on them.

% Along with a standard assumption on the naming of bound variables, we henceforth adopt the following conventions:

% \begin{remark}
% \label{conv:naming-of-variables-and-eigenfunctions}
% Henceforth we shall assume the following for fixed proofs $\der$:
% \begin{itemize}
%     \item All quantifiers in a formula bind distinct variables. As a result, along any hypertrace of $\der$, a variable can be instantiated at most once by a term in a $\instrule$-step.
%     \item Any eigenfunction, i.e.\ the fresh symbol $f$ indicated in the premiss of the $\forall$-rule or the $\coTC$-rule, is distinct from all function symbols occurring below it in $\der$.
% \end{itemize}
% \anupam{put more things here as needed}
% \end{remark}

Throughout this section, we shall allow interpretations to be only partially defined, i.e.\ they are now \emph{partial} maps from the set of function symbols of our language to appropriately typed functions in the structure at hand.
Typically our interpretations will indeed interpret the function symbols in the context in which they appear, but as we consider further function symbols it will be convenient to extend an interpretation `on the fly'.
	This idea is formalised in the following definition:
	
	\begin{definition}[Interpretation extension]
	%\anupam{why not just call it an `extension'? Like, for sets, if $A\supseteq B$ then $A$ \emph{extends} $B$.}
	    Let $\mathcal{M}$ be a structure and $\rho$, $\rho'$ be two interpretations over $\domain{\mathcal{M}}$. We say that $\rho '$ is an \emph{extension} of $\rho$, written $\rho \subseteq \rho'$, if $\rho'(f) = \rho(f)$, for all $f$ in the domain of $\rho$.
	\end{definition}

%For practical reasons we shall assume that $\der$ is substitution-free (at the cost of regularity) and
Finally, we assume that each quantifier in $\sq$ binds a distinct variable.
Note that this convention means we can simply take $y=x$ in the $\exists$ rule in Fig.~\ref{fig:rules:htcl}.

\subsection{Constructing a `countermodel' branch}
Recall that we have fixed at the beginning of this section a $\htcl$ preproof $\der$ of a hypersequent $\sq$.
Let us fix some structure $\fmd $ and an interpretation $\rho_0$ such that $\rho_0 \not\models \sq$ (within $\fmd$). 
Since each rule is locally sound, by contraposition we can continually choose `false premisses' to construct an infinite `false branch':
\begin{lemma}
[Countermodel branch]
\label{lem:countermodel-branch}
There is a branch $\fbr = (\sq_i)_{i < \omega}$ of $\der$ and an interpretation $\fint$ such that, with respect to $\fmd$:
\begin{enumerate}
    \item\label{item:fint-falsifies-Si} $\fint \not\models \sq_i$, for all $i<\omega$;
    \item\label{item:fint-assigns-minimal-n} Suppose that $\sq_i$ concludes a $\coTC$ step,
%     	$$
%     	\scriptsize
% 	\vlinf{\coTC}{}{\sq_i \, = \, \qq, \str{\Gamma, \cotc{A}{s}{t}}{\bv}}{ \sq_{i+1} \, = \, \qq, \str{\Gamma, {A}(s,t),{A}(s,f(\bv))}{\bv}, \str{\Gamma, {A}(s,t),\cotc{A}{f(\bv)}{t}}{\bv} } 
% 	$$
%	\anupam{All $\bar A$ above should be $A$.}
as typeset in Fig.~\ref{fig:rules:htcl},
and $\fint \models \tc{\bar A}st \, [\ntuple/\bv]$.
If $n$ is 
minimal 
such that $\fint \models \bar A(d_i,d_{i+1})$ 
for all 
$i\leq n$, $\fint(s)=d_0$ and $\fint(t)=d_n$, 
and $n>1$, 
then $\fint(f)(\vec d) = d_1$\footnote{To be clear, we here choose an arbitrary such minimal `$\bar A$-path'.
	% to set $d_1$.
}
% 	\[
% 	1) \quad \rho_{i} \models \bigvee \nf{\Gamma} \, [\ntuple / \bv]
% 	\quad \text{ or } \quad 
% 	2) \quad 
% 	\rho_i \models \tc{\bar A}st \, [\ntuple/\bv]
% 	\]
% 		We define $ \rho_{i+1} $ to conservatively extend $ \rho_{i} $ by defining $\rho_{i+1}(f)$ as follows.
% 	Let $ \ntuple \subseteq \fmd $.
% 	If $1)$ holds, then we may set $\rho_{i+1}(\ntuple)$ to be an arbitrary element of $\domain \fmd$.
% 	Otherwise, $2)$ must hold, so 
% 	By the truth conditions for $\TC$ there is a $ \nf A $-path between $ \rho_{i}(s) $ and $ \rho_{i}(t) $ of length $\geq 1$, i.e.\ there are elements $d_0, \dots, d_{n}$, with $n>0$ and $\rho_i(s)=d_0 $ and $\rho_i(t)=d_{n}$, such that $\rho_i \models \bar A(d_i,d_{i+1})$ for all $i< n$. 
% 	We select a \emph{shortest} such path, i.e.\ one with smallest possible $n>0$.
% 	If $n>1$ then $\rho_{i+1}(f)(\vec d) = d_1$, 
	so that $\rho_{i+1} \models \bar A(s,f(\bv)) [\ntuple/\bv]$ and $\fint \models \tc{\bar A}{f(\vec x)}t [\ntuple/\bv]$.
\end{enumerate}
\end{lemma}
%Unpacking this a little, 

Intuitively, our interpretation $\fint$ is going to be defined at the end of the construction as the limit of a chain of `partial' interpretations $(\rho_i)_{i<\omega}$, with each $\rho_i \not\models \sq_i$ (
%inside 
within $\fmd$).
Note in particular that, by \Cref{item:fint-assigns-minimal-n}, whenever some $\coTC$-formula is principal, we choose $\rho_{i+1}$ to always assign to it 
a 
%`least' 
falsifying path 
of minimal length 
(if one exists at all), with respect to the assignment to variables in its annotation.
It is crucial at this point that our definition of $\fint$ is parametrised by such assignments.

\begin{proof}[Proof of \Cref{lem:countermodel-branch}]
% We show how to construct an interpretation $\fint$ and select a branch $\fbr$ in $\der$ such that both \cref{item:counter_1} and \cref{item:counter_2} are satisfied. 
To construct $\fint$, we extend the interpretation $\rho_0$ at every step of the derivation. Thus, we shall define a chain of interpretations $\rho_0 \subseteq \rho_1\subseteq \rho_2 \subseteq \cdots $ such that, for each $i$, $\fmd, \rho_i \models \sq_i $. We will define $\fint$ as the limit of this chain. 
%\marianna{check that $\fint$ is coherent with the main text}
%In the following, we sometimes omit specifying the model $\fmd$ for satisfaction, writing simply $ \rho_i \models C $ instead of $\fmd, \rho_i \models C $. 
In the following, we shall write $\rho \models \sq$ instead of $\rho \models \fmi{\sq}$. 
We distinguish cases according to the last rule $ r_i $ applied to $ \sq_i $. 
For the case of weakening, $\sq_{i+1}$ is the unique premiss of the rule, and $\rho_{i+1} = \rho_i$. %Here follows all the other cases.

\noindent \textbf{$\triangleright$ Case $ (\cup) $}
$$
\vliiinf{\cup}{}{\sq_i \, = \, \qq , \str{\Gamma_1, \Gamma_2}{\bv_1,\bv_2} }{\sq^1 \, = \, \qq, \str{\Gamma_1}{\bv_1}}{}{\sq^2 \, = \, \qq , \str{ \Gamma_2}{\bv_2}} 
$$
By assumption, $ \rho_i \not\models \qq $ and  $ \rho_i \models \forall \bv_1 \,\forall \bv_2(\bigvee \nf{\Gamma}_1 \vlor \bigvee \nf{\Gamma}_2) $. 
%We shall define a  $ \rho_{i+1} $ such that  $ \rho_{i+1}\not\models \fmi{\mathcal{Q}} $ and either $ \rho_{i+1} \models \forall \bv_1( \bigvee \nf{\Gamma}_1) $ or $ \rho_{i+1} \models\forall \bv_2 ( \bigvee \nf{\Gamma}_2 ) $. 
Set $ \rho_{i+1} = \rho_i $. 
By the truth condition associated to $ \forall $, we have that 
for all $m$-tuples $ \ntuple_1 \in \domain{\fmd}$ and $n$-tuples $ \ntuple_2 \in \domain{\fmd}$, 
for $ n = \card{\bv_1} $, $ m = \card{\bv_2} $, 
it holds that\footnote{Recall that $ \bv_1 $, $ \bv_2 $ are defined as \emph{sets} of variables, but we are here considering them as \emph{lists}, assuming their elements to be ordered according to some fixed canonical ordering of variables.}:
$$
\rho_{i+1} \models \bigvee \nf{\Gamma}_1 \vlor \bigvee \nf{\Gamma}_2 \quad [\ntuple_1 / \bv_1][\ntuple_2 / \bv_2]
$$
By the truth condition associated to $ \vlor $ we can conclude that, for all 
$\ntuple_1$, $\ntuple_2$, either:
$$
\rho_{i+1} \models \bigvee \nf{\Gamma}_1 \quad [\ntuple_1 / \bv_1][\ntuple_2 / \bv_2]
\quad \text{ or } \quad 
\rho_{i+1} \models  \bigvee \nf{\Gamma}_2\quad [\ntuple_1 / V_1][\ntuple_2 / V_2]
$$
Since $ \bv_1  \cap \freev{\Gamma_2} = \emptyset$ and $ \bv_2 \cap \freev{\Gamma_1}=\emptyset$, the above is equivalent to:
$$
\rho_{i+1} \models \bigvee \nf{\Gamma}_1 \quad [\ntuple_1 / \bv_1]
\quad \text{ or } \quad 
\rho_{i+1} \models  \bigvee \nf{\Gamma}_2\quad [\ntuple_2 / \bv_2]
$$
And, since this holds 
%for all $ \vec{d}_n $ and $ \vec{d}_m $, we can conclude that:
for all choices of $ \ntuple_1$ and $\ntuple_2 $, we can conclude that:
$$
\rho_{i+1}\models \forall \bv_1( \bigvee \nf{\Gamma}_1)
\quad \text{ or } \quad 
\rho_{i+1} \models\forall \bv_2(  \bigvee \nf{\Gamma}_2)
$$
Take $ \sq_{i+1} $ to be the $ \sq^k $ such that $ \rho_{i+1} \models \forall \bv_k( \bigvee \nf{\Gamma}_k)$, for $ k = \{1,2\} $. 

\

\noindent For all the remaining cases, $ \sq_{i+1} $ is the unique premiss of the rule $ r_i $. 
Moreover, for $ \bv $ the unique (possibly empty) annotation explicitly reported in the remaining rules, let $ n = \card{\bv} $ and $ \ntuple \in \domain{\fmd} $ be a $ n $-tuple of elements of the domain. 

\

%%%%%%%%%%%%%%%%%%%%%%%%%%%%%%%%%%%%%%%%%%%%%%%%%%%%%%%%%%%%%%%%%%%%%%%%%%%%%%%%%%%%%%%%%%%%%%%%%%%%%
% Formula interpr. conclusion -> formula interp. premiss
%%%%%%%%%%%%%%%%%%%%%%%%%%%%%%%%%%%%%%%%%%%%%%%%%%%%%%%%%%%%%%%%%%%%%%%%%%%%%%%%%%%%%%%%%%%%%%%%%%%%%

\noindent \textbf{$\triangleright$  Cases $(  \vlan) $, $(  \vlor) $, $ (\exists) $, $(\reset)$ and $ (\tcrule) $}

\ 

	\begin{adjustbox}{max width = \textwidth}
	\begin{tabular}{c}
		$	\vlinf{\vlan}{}{\sq_i\,  = \, \qq, \str{\Gamma, A \vlan B}{\bv}}{\sq_{i+1}\,  = \, \qq , \str{\Gamma, A , B}{\bv}}
		\quad 
		\vlinf{\vlor}{}{\sq_i \, = \, \qq, \str{\Gamma, A \vlor B}{\bv}}{\sq_{i+1} \, = \, \qq, \str{\Gamma, A }{\bv}, \str{\Gamma, B}{\bv}}
		\quad
		\vlinf{\exists}{}{ \sq_i \, = \, \qq , \str{\Gamma, \exists x (A(x))}{\bv} }{ \sq_{i+1} \, = \,  \qq, \str{\Gamma, A(x)}{\bv, x} }$\\[0.5cm]
		$\vlinf{\reset}{}{\sq_i = \qq, \str{\Gamma, A}{\bv},\str{\nf A }{\varnothing}}{\sq_{i+1} = \qq , \str{\Gamma}{\bv}} 
		\quad  
		\vlinf{\tcrule}{
			%\textcolor{purple}{z \notin FV(\Gamma,A)}
		}{\sq_i \, = \, \qq, \str{\Gamma, \tc{A}{s}{t}}{\bv}}{\sq_{i+1} \, = \, \qq, \str{\Gamma, A(s,t)}{\bv}, \str{\Gamma,A(s,z), \tc{A}{z}{t}}{\bv,z} } $
	\end{tabular}
\end{adjustbox}

\ 

\noindent For all these cases set $\rho_{i+1} = \rho_i$. In all these cases, the formula interpretation of the conclusion logically implies the formula interpretation of the premiss. Thus, form $\fmd, \rho_{i} \not \models \sq_i$  we have that $ \fmd, \rho_{i+1} \not \models \sq_{i+1} $. We explicitly show the construction for $ (\exists) $, $ (\reset) $ and $ (\mathsf{tc}) $. 

\ 

\noindent $ (\exists) $ 
By assumption, $ \rho_i\not\models \qq  $ and $ \rho_i \models \forall \bv (\bigvee \nf{\Gamma} \vlor \forall x (\nf A(x))) $. We have that $ \rho_{i+1}\not\models \qq $. Moreover, by prenexing the quantifier we obtain that $ \rho_{i+1} \models \forall \bv \,\forall x ( \bigvee \nf{\Gamma} \vlor \nf{A}(x)) $. 

\ 

\noindent $ (\reset) $
By assumption, $ \rho_i\not\models \qq $ and $ \rho_i \models \forall \bv (\bigvee \nf{\Gamma} \vlor \nf{A}) $ and $\rho_i \models A $. By the forcing condition associated to $ \forall $ we have that, for all choices of $ \ntuple $, it holds that:
%choices of $\delta:\bv \to D$,
$$
\rho_{i+1} \models \bigvee \nf{\Gamma} \vlor \nf{A} \quad [\ntuple / \bv]
$$
By the truth condition associated to $ \vlor $, for every choice of $\ntuple $ it holds that either:
$$
\rho_{i+1}  \models \bigvee \nf{\Gamma} \quad [\ntuple / \bv]\quad \text{ or } \quad  \rho_{i+1}  \models\nf{A} \quad [\ntuple / \bv]
$$
Since $\freev{A} \cap \bv = \emptyset $, the above is equivalent to:
$$
\rho_{i+1} \models \bigvee \nf{\Gamma} \quad [\ntuple / \bv] \quad \text{ or } \quad  \rho_{i+1} \models\nf{A}
$$
By assumption, $ \rho_{i+1} \models A $. Thus, the second disjunct cannot hold, and we have that
$
\rho_{i+1}  \models \bigvee \nf{\Gamma} \quad [\ntuple / \bv]
$. 
Since this holds for all choices of $\ntuple $, we conclude that 
$
\rho_{i+1} \models \forall \bv (\bigvee \nf{\Gamma} )
$. 

\ 

	\noindent $ (\mathsf{tc}) $
By assumption, $ \rho_i\not\models \qq $ and $ \rho_i \models \forall \bv (\bigvee \nf{\Gamma} \vlor \nf{\tc{A}{s}{t})} $. Recall that $\nf{\tc{A}{s}{t})} \defsym \cotc{\nf A}{s}{t}$. We reason as follows:
\begin{align*}
&\rho_{i+1} \models \forall \bv (\bigvee \nf{\Gamma} \vlor \cotc{\nf A}{s}{t}) &  \\
&\rho_{i+1} \models \forall \bv (\bigvee \nf{\Gamma} \vlor (\nf{A}(s,t) \vlan \forall z(\nf{A}(s,z) \vlor \cotc{\nf A}{z}{t})) ) & (\star)\\
&\rho_{i+1} \models \forall \bv ( (\bigvee \nf{\Gamma} \vlor \nf{A}(s,t)) \vlan (\bigvee \nf{\Gamma} \vlor \forall z(\nf{A}(s,z) \vlor \cotc{\nf A}{z}{t})) )& %\text{by distributivity}
\\	 
&\rho_{i+1} \models \forall \bv ( (\bigvee \nf{\Gamma} \vlor \nf{A}(s,t) )\vlan \forall z(\bigvee \nf{\Gamma} \vlor \nf{A}(s,z) \vlor \cotc{\nf A}{z}{t}) )&
%\text{by prenexing} 
\\	 
&\rho_{i+1} \models \forall \bv ( \bigvee \nf{\Gamma} \vlor \nf{A}(s,t) ) \vlan \forall \bv \forall z(\bigvee \nf{\Gamma} \vlor \nf{A}(s,z) \vlor \cotc{\nf A}{z}{t}) & (*)  \\	 
&\rho_{i+1} \models \forall \bv ( \bigvee \nf{\Gamma} \vlor \nf{A}(s,t) ) \text{ and }  \rho_{i+1} \models\forall \bv \forall z(\bigvee \nf{\Gamma} \vlor \nf{A}(s,z) \vlor \cotc{\nf A}{z}{t}) & % \text{by truth of } \vlan 
%\\
%&\rho_{i+1} \models \forall \bv ( \bigvee \nf{\Gamma} \vlor \nf{A}(s,t) ) \text{ and }  \rho_{i+1} \models\forall \bv \forall z(\bigvee \nf{\Gamma} \vlor \nf{A}(s,z) \vlor \nf{\tc{ A}{z}{t}}) & % \text{by definition}  
\end{align*}
In the above, step $(\star)$ follows from the inductive definition of $\coTC$, and step $ (*) $ is obtained by distributing $ \forall $ over $ \vlan $, i.e., by means of the classical theorem $ \forall x (A \vlan B) \vljm (\forall x (A) \vlan \forall x (B) )$. The other steps are standard theorems or follows from the truth conditions of the operators.

%%%%%%%%%%%%%%%%%%%%%%%%%%%%%%%%%%%%%%%%%%%%%%%%%%%%%%%%%%%%%%%%%%%%%%%%%%%%%%%%%%%%%%%%%%%%%%%%%%%%%
% Inst
%%%%%%%%%%%%%%%%%%%%%%%%%%%%%%%%%%%%%%%%%%%%%%%%%%%%%%%%%%%%%%%%%%%%%%%%%%%%%%%%%%%%%%%%%%%%%%%%%%%%%

\

For the three remaining case of $ (\mathsf{inst}) $, $ (\forall) $ and $ (\nf{\mathsf{tc}}) $, $\rho_{i+1}$ extends $\rho_i$ by interpreting the function symbols introduced, bottom-up, by the rules.

\ 

\noindent \textbf{$\triangleright$ Case	$ (\mathsf{inst}) $}
$$
\vlinf{\mathsf{inst}}{}{\sq_i \, = \, \qq, \str{\Gamma(y)}{\bv ,y}}{\sq_{i+1} \, = \, \qq , \str{\Gamma(t)}{\bv~~}} 
$$
By assumption, $ \rho_i\not\models \qq $ and $ \rho_i \models \forall \bv\, \forall x( \bigvee \nf{\Gamma}(x)) $. Thus, for all choices of $ \ntuple$, we have that 
$
\rho_{i} \models \forall x (\bigvee \nf{\Gamma}(x)) \, [\ntuple / \bv] 
$. By the truth condition associated to $ \forall $, this means that, for all  $ d \in \domain{\fmd}$, 	$
\rho_{i} \models \bigvee \nf{\Gamma}(x) \, [\ntuple / \bv] [d / x]
$. Take $ \rho_{i+1} $ to be any extension of $ \rho_i $ that is defined on the language of $\sq_{i+1}$. That is, if $ f $ is a function symbol in $t$ to which $ \rho_{i} $ already assigns a map, then $ \rho_{i+1} $ assigns to it that same map. Otherwise, $ \rho_{i+1} $ assigns an arbitrary map to $ f $. 
% 	From the assumption that all cedents are closed, it follows that $ t $ cannot contain free variables. 
It follows that $ \rho_{i+1}\not\models \fmi{\mathcal{Q}} $ and $ \rho_{i+1} \models \bigvee \nf{\Gamma}(t) [\ntuple / \bv]$ and, since this holds for all $ \ntuple $, we have that $ \rho_{i+1} \models\forall \bv ( \bigvee \nf{\Gamma}(t))$. Thus  $ \rho_{i+1} \not \models \sq_{i+1} $.

%\marianna{Add a side condition: $ \freev{t}\notin \bv $, to avoid capturing a free variable within the scope of $ \bv $, and require that $ t $ can't be a variable}

\

%%%%%%%%%%%%%%%%%%%%%%%%%%%%%%%%%%%%%%%%%%%%%%%%%%%%%%%%%%%%%%%%%%%%%%%%%%%%%%%%%%%%%%%%%%%%%%%%%%%%%
% Forall
%%%%%%%%%%%%%%%%%%%%%%%%%%%%%%%%%%%%%%%%%%%%%%%%%%%%%%%%%%%%%%%%%%%%%%%%%%%%%%%%%%%%%%%%%%%%%%%%%%%%%

\noindent \textbf{$\triangleright$ Case $ (\forall) $}
$$
\vlinf{\forall}{}{ \sq_i \, = \, \qq, \str{\Gamma, \forall x (A(x)) }{\bv } }{\sq_{i+1} \, = \, \qq , \str{\Gamma, A(f(\bv))}{\bv}}
$$
By assumption, $ \rho_i\not\models \qq $ and $ \rho_i \models \forall \bv (\bigvee \nf{\Gamma} \vlor \exists x (\nf{A}(x)) $.  
By the truth condition associated to $ \forall $ and to $ \vlor $, for all choices of $ \ntuple$ we have that:
$$
\rho_i \models  \bigvee \nf{\Gamma} \quad [\ntuple / \bv] \quad  \text{ or } \quad  	\rho_i \models \exists x(\nf{A}(x)) \quad [\ntuple / \bv]
$$
We define $ \rho_{i+1} $ to extend $\rho_i$ by defining $\rho_{i+1}(f)$ as follows. 
Let $\vec d \subseteq \domain{\fmd}$.
If $ \rho_{i} \models  \bigvee \nf{\Gamma} \, [\ntuple / \bv]$
then we may set $\rho_{i+1}(f)(\vec d)$ to be arbitrary. 
We still have $ \rho_{i+1} \models  \bigvee \nf{\Gamma} \, [\ntuple/\bv]$, as required.
% 	, and the value that $ \rho_{i+1} $ assigns to $ f(\bv) $ can be arbitrary. Set: 
% 	\marianna{Fixed the type of $\rho_{i+1}(f)$ (I think)}
% 	\begin{eqnarray*}
% 		\rho_{i+1}(f) : \domain{\fmd}^n \to \domain{\fmd} \quad \text{ s.t.} \quad  \rho_{i+1}(f)(\bv) =  c, \text{ for an arbitrary } c \in \domain{\fmd}
% 		%, \text{ for any } \vec{d}_n \subseteq D
% 	\end{eqnarray*}
Otherwise, it holds that $ \rho_i \models \exists x(\nf{A}(x)) \, [\ntuple / \bv]$. 
By the truth condition associated to $ \exists $, there is a $ d \in \domain{\fmd} $ such that 
$
\rho_i  \models \nf{A}(x) \, [\ntuple / \bv][d / x]
$. 
%Take an arbitrary $ \ntuple $.  
%Define $ \rho_{i+1} $ to agree with $ \rho_{i} $ on all free variables and function symbols occurring in $ \rho_{i} $. 
In this case, we define $ \rho_{i+1}(f) (\vec d) =d$, 
% 	as follows:
% 	\begin{eqnarray*}
% 		\rho_{i+1}(f) : D^n \to D \quad \text{ s.t.} \quad  \rho_{i+1}(f)(\bv) =  d
% 		%, \text{ for any } \vec{d}_n \subseteq D
% 	\end{eqnarray*}
% 	Therefore, 
so that $ 	\rho_{i+1}  \models \nf{A}(f(\bv)) \, [\ntuple / \bv] $. 
So, for all $\vec d$, we have that
$
\rho_{i+1} \models  \bigvee \nf{\Gamma} \vlor  \nf{A}(f(\bv)) \, [\ntuple/\bv]
$, 
and so
$	\rho_{i+1}  \models \forall \bv (  \bigvee \nf{\Gamma} \vlor  \nf{A}(f(\bv)))
$. 
Thus, $ \rho_{i+1} \not \models \sq_{i+1} $, as required.

\

%%%%%%%%%%%%%%%%%%%%%%%%%%%%%%%%%%%%%%%%%%%%%%%%%%%%%%%%%%%%%%%%%%%%%%%%%%%%%%%%%%%%%%%%%%%%%%%%%%%%%
% Not tc
%%%%%%%%%%%%%%%%%%%%%%%%%%%%%%%%%%%%%%%%%%%%%%%%%%%%%%%%%%%%%%%%%%%%%%%%%%%%%%%%%%%%%%%%%%%%%%%%%%%%%

\noindent \textbf{$\triangleright$  Case $ (\cotcrule) $} 
%	\anupam{8/2: I reworked this case considerably. I think it is the correct structure now; earlier the meta-binding of parameters and applications of the inductive invariant were occurring in the wrong order.}
$$
\vlinf{\cotcrule}{}{\sq_i \, = \, \qq, \str{\Gamma, \cotc{A}{s}{t}}{\bv}}{ \sq_{i+1} \, = \, \qq, \str{\Gamma, {A}(s,t),{A}(s,f(\bv))}{\bv}, \str{\Gamma, {A}(s,t),\cotc{A}{f(\bv)}{t}}{\bv} } 
$$
%	\anupam{All $\bar A$ above should be $A$.}

By assumption, $ \rho_i\not\models \qq$ and $ \rho_i \models \forall \bv (\bigvee \nf{\Gamma} \vlor \nf{\cotc{A}{s}{t})} $ which, by definition of duality, means $ \rho_i \models \forall \bv (\bigvee \nf{\Gamma} \vlor \tc{\nf A}{s}{t}) $. 
By the truth conditions for $\vlor$ we have, for all $\ntuple$:
\[
1) \quad \rho_{i} \models \bigvee \nf{\Gamma} \, [\ntuple / \bv]
\quad \text{ or } \quad 
2) \quad 
\rho_i \models \tc{\bar A}st \, [\ntuple/\bv]
\]
We define $ \rho_{i+1} $ to extend $ \rho_{i} $ by defining $\rho_{i+1}(f)$ as follows.
Let $ \ntuple \subseteq \fmd $.
If $1)$ holds, then we may set $\rho_{i+1}(\ntuple)$ to be an arbitrary element of $\domain \fmd$.
Otherwise, $2)$ must hold, so by the truth conditions for $\TC$ there is a $ \nf A $-path between $ \rho_{i}(s) $ and $ \rho_{i}(t) $ of length greater or equal than 1, i.e.\ there are elements $d_0, \dots, d_{n}$, with $n>0$ and $\rho_i(s)=d_0 $ and $\rho_i(t)=d_{n}$, such that $\rho_i \models \bar A(d_i,d_{i+1})$ for all $i< n$. 
We select a \emph{shortest} such path, i.e.\ one with smallest possible $n>0$.
There are two cases:
\begin{enumerate}[i)]
	\item if $n=1$, then already $\rho_i \models \bar A(s,t) [\ntuple/\bv]$, so we may set $\rho_{i+1}(f)(\ntuple)$ to be arbitrary;
	\item otherwise $n>1$ and we set $\rho_{i+1}(f)(\vec d) = d_1$, so that $\rho_{i+1} \models \bar A(s,f(\bv)) [\ntuple/\bv]$ and $\rho_{i+1} \models \tc{\bar A}{f(\vec x)}t [\ntuple/\bv]$.
\end{enumerate}

We have considered all the rules; the construction of $\fbr$ is complete.
From here, note that we have $\rho_i \subseteq \rho_{i+1}$, for all $i<\omega$. 
Thus we can construct the limit $\fint = \bigcup_{i<\omega} \rho_i$, that we shall call the interpretation \emph{induced} by $\fmd$ and $\rho_0$.
\end{proof}

\subsection{Canonical assignments along countermodel branches}
Let us now fix $\fbr$ and $\fint$ as provided by \Cref{lem:countermodel-branch} above.
Moreover, let us henceforth assume that $\der$ is a proof, i.e.\ it is progressing, and fix a progressing hypertrace $ \fhy = (\str{\Gamma_i}{\bv_i})_{i< \omega} $ along $\fbr$.
In order to carry out an infinite descent argument, we will need to define a particular trace along this hypertrace that `preserves' falsity, bottom-up.
This is delicate since the truth values of formulas in a trace depend on the assignment of elements to variables in the annotations.
A particular issue here is the instantiation rule $\mathsf{inst}$, which requires us to `revise' whatever assignment of $y$ we may have defined until that point.
Thankfully, our earlier convention on substitution-freeness and uniqueness of bound variables in $\derd$ facilitates the convergence of this process to a canonical such assignment:

	\begin{definition}[Assignment]
	\label{def:assignment_h}
% 	Let $\rho$ and $ \fbr = (\sq_i)_{i< \omega} $ be an interpretation and a derivation branch constructed according to \Cref{def:false_branch_tcl}.% and having the properties indicated in \Cref{lemma:failed_branch}. 
% %	i.e., for which there exists a model and a sequence $ (\rho_i)_{i<\omega} $ such that for all $ i $, $	\fmd, \rho_{i} \not\models \sq_i $. 
% 	\anupam{formally, $(\rho_i)_i$ needs to be fixed for this definition.}
% 	\marianna{done} 
% 	For  an infinite hypertrace along $ \fbr $,  let $ \fv{\hyp} = \bigcup_i \bv_i$ and 
% 	Let $d\in \domain \fmd$ be arbitrary.
	We define $\deltah : \bigcup\limits_{i<\omega}\vec x_i \to \domain{\fmd}$ by $	\deltah(x) \defsym \rho(t)$ if $x$ is instantiated by $t$ in $\fhy$; otherwise $\deltah(x)$ is some arbitrary $d\in \domain \fmd$.
% 	\begin{eqnarray*}
% 		\deltah : \fv{\hyp} \to \domain{\fmd} \; \text{ s.t. } \;
% 				\left\{
% 		\begin{array}{l l}
% 		\deltah(x) = \rho(t)
% % 		, \text{ for some } i< \omega, 
% 		& \text{if } x \text{ appears in } \fhy  \text{ and } \\[0.1cm]
% 		& x \text{ is instantiated by $t$ along $\fhy$} ;\\[0.3cm]
% 		\deltah(x) = d
% % 		,  \text{ for an arbitrary } c \in \domain{\fmd}, 
% 		& \text{ otherwise.}
% 		\end{array}
% 		\right.
% 	\end{eqnarray*}
	\end{definition}

Note that $\deltah$ is indeed well-defined, thanks to the convention that each quantifier in $\sq$ binds a distinct variable.
In particular we have that each variable $ x $ is instantiated at most once along a hypertrace.
% : this is because we assumed proofs to be substitution-free, because in our language we assume each bound variable in our language is bound by at most one quantifier, and because the rules only introduce fresh variables. 
% 	\anupam{do we need to say this earlier?}
% 	\anupam{oh, remember, no rules introduce `variables'. I stop this pedantry now.}
% 	\marianna{There's a remark in Section 4}
	Henceforth we shall simply write $\rho,\deltah \models A(\vec x)$ instead of $\rho \models A(\deltah(\vec x))$.
% 	For a formula $A(\vec x)$ with all free variables displayed, and an assignment $\delta$ mapping each $x\in \vec x$ to an element of $\domain \fmd$, we shall write $\rho,\delta \models A(\vec x)$ if $\rho \models A(\delta(\vec x))$.
% 	\anupam{why not use this notation in previous subsection too, instead of $[\ntuple/\bv]$ everywhere?}
	Working with such an assignment ensures that false formulas along $\fhy$ always have a false immediate ancestor:
	
\begin{lemma}[Falsity through $\fhy$]
		\label{lemma:failed_trace}
	   % Let $\rho$ and $ \fbr = (\sq_i)_{i< \omega} $ be an interpretation and a derivation branch constructed according to \Cref{def:false_branch_tcl} and having the properties indicated in \Cref{lemma:failed_branch}.
	   % Let $ \fhy = (\str{\Gamma_i}{\bv_i})_{i< \omega} $ be infinite hypertrace in $ \fbr $ and let $\deltah$ be as defined in \Cref{def:assignment_h}. 
	 	 %   %We inductively define the \emph{trace $\ftr = (F_i)_{i < \omega}$ induced by $\fmd$, $\rho$}, which occurs in $\fhy$ and satisfies the following invariant: 
	 	 If $\fint, \deltah \not \models F$ for some $F\in \Gamma_i$, then $F$  has an immediate ancestor $F'\in \Gamma_{i+1}$ with $\fint,\deltah \not\models F'$. 
	 	 %\anupam{change notation to $\deltah^\times$?}
	 	 %\marianna{yes! }
	   % For all inference steps $ r_i $ along $ \fbr $, and for all formulas $ F $ occurring in some cedent $ \str{\Gamma_i}{\bv_i} $ in $\sq_i$, we have:
    %     $$
	   % \text{if} \quad \fmd, \rho , \deltah \not \models F
	   % \quad 
	   % \text{then there exists a formula } F' \text{ s.t.}
	   % \quad
	   % \fmd, \rho, \deltah \not \models F'
	   % $$
	   % where $ F' $ occurs in $ \str{\Gamma_{i+1}}{\bv_{i+1}} $ and %\anupam{
    %     is an immediate ancestor of $F$ in $\Gamma_{i+1}$.
    %     \todo{in original definition of trace earlier, need to explicitly use terminology `immediate ancestor'.}
\end{lemma}	

\begin{proof}
%[of Lem.~\ref{lemma:failed_trace}]
	Suppose that $F$ is a formula  occurring in some cedent $ \str{\Gamma_i}{\bv_i} $ in $\sq_i$, such that $\fint , \deltah \not \models F$. We show how to choose a $F'$ satisfying the conditions of the Lemma. 
	We distinguish cases according to the rule $ \infrule r_i $ applied. Propositional cases are immediate, by setting $F'=F$, as well as $ \cup $, since the failed branch has been chosen during the construction of $ \fbr $. 
	%\anupam{mention that we can set $F'=F$?}\marianna{done} 
	For the rule of weakening, observe that we could not have chosen a hypertrace going through the structure which gets weakened, as by assumption the hypertrace is infinite. 
	%Thus, it suffices to set $F = F'$.  
	%The case of substitution is immediate. \anupam{better to assume no substitution, otherwise $\delta_h$ won't be well-defined.}
	We show the remaining cases. It can be easily checked that, given a formula $ F $ such that $ \fmd, \fint , \deltah \not \models F $, the formula $ F' $ is an immediate ancestor of $F$.
	%	which we select belongs to a trace where $ F $ belongs.
	
	\
	
	\noindent \textbf{$ \triangleright $ Case $ (\exists) $.}
	Suppose $ \str{\Gamma_i}{\bv_i} = \str{\Gamma, \exists x (A(x))}{\bv}$ and $ \str{\Gamma_{i+1}}{\bv_{i+1}} = \str{\Gamma ,A(x)}{\bv,x}$. 
	By assumption, $\fint \models \forall \bv (\bigvee \nf{\Gamma} \vlor \forall x (\nf A(x))) $.
	Applying the truth condition associated to $ \forall $ we obtain that, for all $ n $-tuples $ \ntuple $ of elements of $ \domain{\fmd} $, for $ n = \card{\bv} $, it holds that:
	\begin{equation}
	\label{eq:lm_trace_exists}
	\fint \models \bigvee \nf{\Gamma} \vlor \forall x (\nf A(x)) \quad  [\ntuple / \bv]
	\end{equation}
	By definition, $ \deltah $ assigns a value in the domain for all the variables in $ \fv{\hyp} $, and $ \bv \subseteq \fv{\hyp}$. 
	%, %each choice of 
	%$ \deltah $ selects $ n $-tuples of elements of $ \domain{\fmd} $. 
	From \eqref{eq:lm_trace_exists} and from the truth condition associated to $ \vlor $ it follows that %, \marianna{for all choices of assignment, and thus also for 	$ \deltah $, 
	either:
	$$
	\fint, \delta_\hyp \models \bigvee \nf{\Gamma} \quad  \text{ or } \quad  \fint, \delta_\hyp \models \forall x (\nf A(x))
	$$
	%\anupam{here and throughout, why ``if $F=C$ for some $C\in \Gamma$...'' instead of ``if $F\in \Gamma$...''}
	If $F \in \Gamma$, by hypothesis we have that $\fint, \delta_\hyp \models \nf{F} $. Choose $ F' = F $. 
	Otherwise, $F = \exists x (A(x))$ and $ \fint, \delta_\hyp \models \forall x  (\nf A(x)) $. 
	By the truth condition associated to $ \forall $, and since  $ x \in \fv{\fhy} $, we conclude that 
	$
	\fint, \delta_\hyp \models \nf A(x) 
	$. 
	%	However, since  $ x \in \fv{\fhy} $, $ \deltah $ already assigns a value to $ x $. 
	Thus, we set $F' = A(x) $.% and conclude that $ \fint, \delta_\hyp \models  \nf A(x) $. 

	\
	
	\noindent \textbf{$ \triangleright $ Case $ (\mathsf{inst}) $.}
	%	$$
	%	\vlinf{\mathsf{inst}}{}{\sq_i \, = \, \qq, \str{\Gamma(x)}{\bv ,x}}{\sq_{i+1} \, = \, \qq , \str{\Gamma(t)}{\bv~~}} 
	%	$$
	Suppose $ \str{\Gamma_i}{\bv_i} = \str{\Gamma(y)}{\bv,y}$ and $ \str{\Gamma_{i+1}}{\bv_{i+1}} = \str{\Gamma(t)}{\bv}$. 
	By construction, $ \fint \models \forall \bv\, \forall y( \bigvee \nf{\Gamma}(y)) $. 
	Reasoning as in the previous case, from the truth condition associated to $ \forall $ it follows that, for all choices of $\vec{d} \in \domain{\fmd}$, and thus also for the $n$-tuple selected by $ \delta_\hyp $, $ \fint, \deltah \models \bigvee \nf{\Gamma}(y) $. 
	%	\anupam{quantification/fixing of $\deltah$ here is confused. $\deltah$ is already defined. (similar in what follows)}\marianna{changed}
	If $F$ does not contain $ y $, then $\fint, \delta_\hyp \models  \nf{F}	$, and we set $F' = F$. 
	%If $F$  contains $ y $ $ F = C(y) $,  for some $ C(y) \in \Gamma $, 
	Otherwise, if $F$ contains $y$, then $\fint, \delta_\hyp \models  \nf{F}(y)$. By \Cref{lem:countermodel-branch}, $ \fint $ assigns a value to $ t $, and by \Cref{def:assignment_h}, since $y$ is instantiated with $t$ along $\fhy$, $ \delta_\hyp(y) = \fint(t) $. Therefore, set $F' = F(t)$ and conclude that $ \fint, \delta_\hyp \models  \nf{F}(t) $.
	
	\
	
	\noindent \textbf{$ \triangleright $ Case $ (\reset) $.} Suppose $ \str{\Gamma_i}{\bv_i} = \str{\Gamma, A}{\bv}$ and $ \str{\Gamma_{i+1}}{\bv_{i+1}} = \str{\Gamma}{\bv}$. 
	Observe that the hypertrace $ \fhy $ could not have gone through the structure $ \str{\nf{ A}}{\varnothing} $ occurring in the conclusion of the rule, because by assumption $ \fhy $ is infinite. 
	Moreover, by construction the formula interpretation of all the cedents along $ \fbr $ is not valid, 
	and thus 
	% 	we have that $ \fint \not \models \str{\Gamma, A}{\bv} $ and 
	$  \fint \not \models \str{\nf{ A}}{\varnothing} $. 
	This implies that $  \fint \models A $ and so: %for any choice of $ \delta_\hyp$, 
	$$
	\fint, \delta_\hyp \models \bigvee \nf{\Gamma} 
	% 	\vlor \nf{A}
	$$
	% 	However, the latter disjunct cannot hold, because $  \fint \models A $. 	
	%	By \Cref{lemma:failed_branch}, $ \fint_{i+1} = \fint_{i} $. %Fix a $ \deltah $. 
	So
	we have that $F\in \Gamma$ and $ \fint, \delta_\hyp \models \nf{F} $. Set $F' = F$.% for some $ F\in \Gamma $, conclude that $ \fint_{i+1}, \delta_\hyp \models \nf{F} $. 
	
	\
	
	\noindent \textbf{$ \triangleright $ Case $ (\tcrule) $. } 
	Suppose $ \str{\Gamma_i}{\bv_i} = \str{\Gamma, \tc{A}{s}{t}}{\bv} $. By assumption, $ \fint \models \forall \bv (\bigvee \nf{\Gamma} \vlor \nf{\tc{A}{s}{t})} $. 
	Thus, %for all choices of $ \delta_\hyp $, 
	we have that: 
	$$
	\fint, \delta_\hyp \models \bigvee \nf{\Gamma} \,  \text{ or } \,  \fint, \delta_\hyp \models \cotc{\nf A}{s}{t}
	$$
	%According to \Cref{lemma:failed_branch}, $ \fint_{i+1} = \fint_{i} $. %Fix a $ \deltah $. 
	%Using \eqref{eq:cotc:fixed-point-formula}
	From the inductive definition of $\coTC$ 
	and the truth condition for $ \vlan $, the second disjunct is equivalent to: % $\fint, \delta_\hyp \models \cotc{\nf A}{s}{t}$ is  equivalent to:
	\begin{equation}
	\label{eq:tc_not}
	\fint, \delta_\hyp \models \nf{A}(s,t) 
	\quad  
	\text{ and } 
	\quad 
	\fint, \delta_\hyp \models \forall z (\nf{A}(s,z) \vlor \cotc{\nf A}{z}{t} )
	\end{equation}	
	There are two cases to consider, since the premiss of the rule has two cedents that the hypertrace $ \fhy $ could follow: 
	\begin{enumerate}[$ i $)]
		\item 
		$\str{\Gamma_{i+1}}{\bv_{i+1}} =   \str{\Gamma, A(s,t)}{\bv}  $. By construction, 
		$
		\fint \models \forall \bv(\bigvee \nf{\Gamma} \vlor \nf{A}(s,t))
		$, and thus $\fint, \delta_\hyp \models \bigvee \nf{\Gamma} \vlor \nf{A}(s,t)$. 
		If $ F  \in \Gamma$, then $ \fint, \delta_\hyp \models \nf{F}  $. Set $F = F'$. Otherwise, $ F = \tc{A}{s}{t} $ and $ \fint, \delta_\hyp \models \nf{F}  $. By (\ref{eq:tc_not}) we have that $	\fint, \delta_\hyp \models \nf{A}(s,t) $. Set $ F' = A(s,t) $.
		
		\item 
		$\str{\Gamma_{i+1}}{\bv_{i+1}} =  \str{\Gamma, A(s,z), \tc{A}{z}{t}}{\bv,z}  $. 
		By construction, we have that $ \fint \models \forall \bv \forall z (\bigvee \nf{\Gamma} \vlor ( \nf{A}(s,z) \vlor \cotc{\nf A}{z}{t})) $. Since $z$ does not occur free in $\Gamma$, this is equivalent to $ \fint, \deltah \models \bigvee \nf{\Gamma} \vlor  \forall z( \nf{A}(s,z) \vlor \cotc{\nf A}{z}{t}) $.
		If $ F  \in \Gamma $ and $ \fint, \delta_\hyp \models \nf{F}  $, set $F' = F$. Suppose $ F = \tc{A}{s}{t} $ and  $ \fint, \delta_\hyp \models \nf{F}  $. From (\ref{eq:tc_not}) and since  $ z \in \fv{\hyp}$ we have that:
		% 		$$
		% 			\fint, \delta_\hyp \models \forall z (\nf{A}(s,z) \vlor \cotc{\nf A}{z}{t} )
		% 		$$
		% 		and, since $ z \in \fv{\hyp}$: %,  $ \delta_\hyp $ assigns a value to $ z $, and thus: 
		$$
		\fint, \delta_\hyp \models \nf{A}(s,z) \vlor \cotc{A}{z}{t} 
		$$
		If %it holds that 
		$ \fint, \delta_\hyp \models \nf{A}(s,z) $, set $ F' = A(s,z) $. Otherwise, if  $ \fint, \delta_\hyp \models \cotc{A}{z}{t}  $, set $ F'= \tc{A}{z}{t} $. 
	\end{enumerate}
	
	\
	
	\noindent \textbf{$ \triangleright $ Case $ (\forall) $.}
	%	$$
	%	\vlinf{\forall}{}{ \sq_i \, = \, \qq, \str{\Gamma, \forall x (A(x)) }{\bv } }{\sq_{i+1} \, = \, \qq , \str{\Gamma, A(f(\bv))}{\bv}}
	%	$$
	Suppose $ \str{\Gamma_i}{\bv_i} = \str{\Gamma, \forall x (A(x))}{\bv}$ and $ \str{\Gamma_{i+1}}{\bv_{i+1}} = \str{\Gamma ,A(f(\bv))}{\bv}$. 
	By assumption, $ \fint \models \forall \bv (\bigvee \nf{\Gamma} \vlor \exists x (\nf A(x))) $ and, from the truth conditions associated to $ \forall $ and $ \vlor $: %, and since $ \delta_\hyp $ assigns a value elements of $ \bv $, we have that: %, for any choice of $ \delta_\hyp $:
	%for all choices of $ \delta: V \to D $,  
	%we have that for all $ \vec{d}_n $, where $ n = \card{\bv} $:
	%$$
	%\fint_i, \delta \models \bigvee \nf{\Gamma} \vlor \exists x \nf{(A(x))}
	%$$
	%Since $ \delta_\hyp $ ranges over all elements of $ V $, it holds that either:
	$$
	\fint, \delta_\hyp \models \bigvee \nf{\Gamma} \quad  \text{ or } \quad  \fint, \delta_\hyp \models \exists x (\nf A(x))
	$$
	%	Fix an arbitrary $ \deltah $. 
	%	By \Cref{lemma:failed_branch}, $ \fint_{i+1} $ is a conservative extension of $\fint_{i} $.
	If $F \in \Gamma$ and $ \fint, \delta_\hyp \models \nf{F} $, set $F' = F$, since $ x $ does not occur in $ \Gamma $.
	%    First check if $ \fint, \delta_\hyp \models \nf{F} $ for some $ F \in \Gamma $. If this is the case, conclude that $ \fint_{i+1}, \delta_\hyp \models \nf{F} $, since $ x $ does not occur in $ \Gamma $. 
	Otherwise, $F = \forall x(A (x))$ and $ \fint, \delta_\hyp \models \exists x  (\nf A(x)) $. 
	%By the truth condition of $\exists$, there exists a $ d \in \domain{\fmd} $ such that $ \fint, \delta_\hyp \models \nf{ A}(x) \, [d / x] $. 
	Suppose $\deltah(\bv) = \vec{d} \subseteq \domain{\fmd}$. By definition of $ \fint $ at step $(\forall)$, we have that $\fint(f)(\vec{d}) $ selects a $ d \in \domain{\fmd}$ such that  $ \fint, \delta_\hyp \models \nf{ A}(f(\bv)) $. Set $F' = A(f(\bv))$.
	
	% 	then by the truth condition associated to $ \exists $, there exists a $ d \in \domain{\fmd} $ such that $ \fint, \delta_\hyp \models \nf{ A}(x) \, [d / x] $. %, and consequently an element $ \delta(x) = c \in D$ such that $ \fint_{i}, \delta_\hyp \models \nf{ A(c)} $. 
	% 	Recall that, by definition of $ \fint $ at step $(\forall)$ (\Cref{def:false_branch_tcl}), 
	% 	if $\fint \models \exists x  \, (\nf A(x)) [\vec{d} / \bv] $, then 
	% 	for every $ n $-tuple $ \ntuple $ of elements of $ \domain{\fmd} $, where $ n = \card{V} $, $ \fint (f)(\bv) $ chooses an element $ d \in \domain{\fmd}  $ such that $ 	\fint  \models \nf{A}(x) \, [\ntuple / \bv][d / x] $, and thus $ 	\fint  \models \nf{A}(f(\bv)) \, [\ntuple / \bv] $. Since $ \deltah $ assigns a value to elements element in $ \bv $, and since the above holds for \emph{any} $n$-tuple of elements, we conclude that under $ \deltah $, $ \fint, \delta_\hyp \models \nf{ A}(f(\bv)) $.  

	\
	
	\noindent \textbf{$ \triangleright $ Case $ (\cotcrule) $.}
	%	\anupam{there are many typos here, and it is the most important case, so please do a pass of this. I'll highlight some.}
	%	\marianna{fixed} 
	Suppose $ \str{\Gamma_i}{\bv_i} = \str{\Gamma, \cotc{A}{s}{t}}{\bv} $. By assumption, $ \fint \models \forall \bv (\bigvee \nf{\Gamma} \vlor \nf{\cotc{A}{s}{t}) }$, that is, $ \fint \models \forall \bv (\bigvee \nf{\Gamma} \vlor \tc{\nf A}{s}{t}) $. %\anupam{should be $\bar A$} 
	Thus, %for all choices of $ \delta_\hyp $, 
	we have that: 
	\begin{equation}
	\label{eq:lemma_trace_cotc}
	\fint, \delta_\hyp \models \bigvee \nf{\Gamma} 
	\quad  
	\text{ or }
	\quad  
	\fint, \delta_\hyp \models \tc{\nf A}{s}{t}
	\end{equation}
	%	\anupam{should be $\bar A$}
	%	By \Cref{lemma:failed_branch}, $ \fint_{i+1} $ is a conservative extension of $ \fint_{i} $. 
	%    \anupam{...this is not stated in lemma. Also, why not define $\fint = \lim_i \fint_i$ in advance to have a fixed common interpretation throughout the branch?}
	%    \marianna{what is not stated in the Lemma?} 
	We need to consider two cases, depending on which cedent the hypertrace $ \fhy $ follows:
	\begin{enumerate}[$ i )$]
		\item $ \str{\Gamma_{i+1}}{\bv_{i+1}} = \str{\Gamma, A(s,t), A(s, f(\bv)) }{\bv} $.
		%\anupam{that is not one of the auxiliary cedents of the $\coTC$ rule, and many polarities are wrong. Should be $\Gamma, A(s,t),A(s,f(V))$.}\marianna{ops.. sorry}
		By construction, $	\fint \models \forall \bv (\bigvee \nf{\Gamma} \vlor (\nf A(s,t) \vlor \nf{A}(s, f(\bv))$, that is,  %i.e., for any choice of $ \delta_\hyp $, it holds that: 
		\begin{equation}
		\label{eq:lemma_trace_cotc_left}
		\fint, \delta_\hyp \models \bigvee \nf{\Gamma} 
		\quad
		\text{ or }
		\quad 
		\fint, \delta_\hyp \models \nf A(s,t)
		\quad 
		\text{ or }
		\quad 
		\fint, \delta_\hyp \models \nf A(s,f(\bv))
		\end{equation}
		%		Fix an arbitrary $ \deltah $. 
		Consider \eqref{eq:lemma_trace_cotc}. If $F \in \Gamma$ and $ \fint, \delta_\hyp \models \nf{F}$, then set $F = F'$. Otherwise,   $ F = \cotc{ A}{s}{t} $ and it holds that $ \fint, \delta_\hyp \models \tc{\nf A}{s}{t} $. By the inductive definition of $\TC$ %\eqref{eq:cotc:fixed-point-formula} 
		and the truth condition associated to $ \vlor $, this is equivalent to:
		$$
		\fint, \delta_\hyp \models \nf A(s,t) 
		\quad  
		\text{ or } 
		\quad 
		\fint, \delta_\hyp \models \exists z (\nf A(s,z) \vlan \tc{\nf A}{z}{t} )
		$$
		First check if $ \fint, \delta_\hyp \models \nf A(s,t) $. If this is the case, set $ F' =  A(s,t)$ and conclude by \eqref{eq:lemma_trace_cotc_left} that $ \fint, \delta_\hyp \models\nf  A(s,t) $. 
		Otherwise, $ \fint, \delta_\hyp \models \exists z (\nf A(s,z) \vlan \tc{\nf A}{z}{t} ) $.
		Let $\deltah(\bv) = \vec{d}$. 
		According to the definition of $ \fint $ at the $(\cotcrule)$ step, since $\fint, \deltah \not \models \bigvee \nf{\Gamma} $ and $\fint, \deltah \not \models \nf{A}(s,t)$, then $\fint(f)(\vec{d})$ is defined as in case $2)$, subcase ii). Thus, $\fint(f)(\vec{d})$ is an element $ d \in \domain{\fmd}$ such that $ \fint, \delta_\hyp \models \nf A(s,f(\bv)) $ and $ \fint, \delta_\hyp \models \tc{\nf A}{f(\bv)}{t}  $. 
		Set $ F' =   A(s,f(\bv)) $ and conclude, by \eqref{eq:lemma_trace_cotc_left}, that $ \fint_{i+1}, \delta_\hyp \models \nf  A(s,f(\bv))$. 
		\item $ \str{\Gamma_{i+1}}{\bv_{i+1}} = \str{\Gamma, A(s,t), \cotc{A}{f(\bv)}{t}}{\bv} $.
		By construction, it holds that either:
		\begin{equation}
		\label{eq:lemma_trace_cotc_right}
		\fint, \delta_\hyp \models \bigvee \nf{\Gamma} 
		\quad
		\text{ or }
		\quad 
		\fint, \delta_\hyp \models \nf A(s,t)
		\quad 
		\text{ or }
		\quad 
		\fint, \delta_\hyp \models \tc{\nf A}{f(\bv)}{t}
		\end{equation}
		At this point the proof proceeds exactly as in the previous case except for the very last step, where $ F'  $ is set to be $ \cotc{A}{f(\bv)}{t} $, and we conclude by \eqref{eq:lemma_trace_cotc_right} that $  \fint, \delta_\hyp \models \tc{\nf A}{f(\bv)}{t}$.  
		%: if $ \fint_{i+1}, \delta_\hyp \models A(s,t)  $ holds choose $ C' = \nf{A(s,t)} $ and, if $ \fint_{i+1}, \delta_\hyp \models A(s,t)  $ does not hold, set $ C' = \nf{\tc{A}{f(V)}{t}} $. 
		
	\end{enumerate}
\end{proof}

\subsection{Putting it all together}
%\noindent
Note how the $\instrule$  of Fig.~\ref{fig:rules:htcl} is handled in the proof above: 
%In particular, regarding the $\mathsf{inst}$ rule of Fig.~\ref{fig:rules:htcl}, note that 
if $F \in \Gamma(y)$ then we can choose $F' = F[t/y]$ which, by definition of $\deltah$, has the same truth value.
By repeatedly applying \Cref{lemma:failed_trace} we obtain:

\begin{proposition}
[False trace]
	\label{prop:failed_trace}
		%Let $ \fbr = (\sq_i)_{i< \omega} $ and $ \fhy = (\str{\Gamma_i}{\bv_i})_{i< \omega}  $ be defined as in \Cref{def:assignment_h}. 
% 		Let $\rho$ and $ \fbr = (\sq_i)_{i< \omega} $ be an interpretation and a derivation branch constructed according to \Cref{def:false_branch_tcl} and having the properties indicated in \Cref{lemma:failed_branch}.  
	   % Let $ \fhy = (\str{\Gamma_i}{\bv_i})_{i< \omega} $ be infinite hypertrace in $ \fbr $ and let $ \delta_\hyp $ be as in  \Cref{def:assignment_h}. 
		There exists an infinite trace $ \ftr = (F_i)_{i<\omega} $ through $ \fhy $ such that, for all $ i $, it holds that $ \fmd, \fint, \delta_\hyp \not \models F_i $. 
		
%		Let $\rho$ and $ \fbr = (\sq_i)_{i< \omega} $ be an interpretation and a derivation branch constructed according to \Cref{def:false_branch_tcl} and having the properties indicated in \Cref{lemma:failed_branch}.  
%	    Let $ \fhy = (\str{\Gamma_i}{\bv_i})_{i< \omega} $ be infinite hypertrace in $ \fbr $. Fix an arbitrary assignment $\deltah$ defined as in \Cref{def:assignment_h}. 
%	 	 We inductively define the \emph{trace $\ftr = (F_i)_{i < \omega}$ induced by $\fmd$, $\rho$}, which occurs in $\fhy$ and satisfies the following invariant: 
	\end{proposition}

	\begin{proof}
	%[of Prop.~\ref{prop:failed_trace}]
		 We inductively define $\ftr $ as follows. 
		By assumption, $\fmd, \fint \not \models \sq_0$, and thus
		in particular $ \fint \models \forall \bv_0 (\bigvee \nf \Gamma_0)$. Thus, for all $n$-tuples $\vec{d}$ of elements of $\domain{\fmd}$, for $n = \card{\bv_0}$, we have that $ \fint \models  \bigvee \nf \Gamma_0 \, [\vec{d}/ \bv_0]$. 
		Since $\deltah$ assigns a value to all variables occurring in annotations of cedents in $\fhy$,  $ \fint , \deltah \models  \bigvee \nf \Gamma_0 $. 
		Take $F_0 \in \Gamma_0$ such that $\fmd, \fint, \deltah \not \models \nf{F_0}$. 
		
		For the inductive step, we define $ F_{i+1} $ by inspecting the rule $ \infrule_i $ applied and at the corresponding case in \Cref{lemma:failed_trace}. Since by assumption $ \fhy $ is infinite, and since \Cref{lemma:failed_trace} ensures that for every formula $ F_i $ such that $ \fmd, \fint, \delta_\hyp \not \models F_i $ it is possible to find a formula $ F_{i+1} $ which is an immediate ancestor of $ F_i $ and such that $ \fmd, \fint, \delta_\hyp \not \models F_{i+1} $. 
		% 	$ \ftr $ is infinite, and it is contains only formulas which are false in $ \fmd $, under $ \delta_\hyp $ and assignment $ \fint $. 
	\end{proof}

\noindent
	We are now ready to prove our main soundness result. 
	\begin{proof}[Proof of Thm.~\ref{thm:soundness_tcl}]
% 		We proceed by contradiction. Let us assume that that there exists a non-wellfounded proof $\fdr$ of $\sqzero$, but $ \tcl \not \models \sqzero $, i.e., there is a model $ \fmd $ and an assignment $ \rho_0 $ such that $ \fmd, \rho_0 \not\models \sqzero $. Using \Cref{def:false_branch_tcl}, we construct a branch $ \fbr= (\sq_{i})_{i < \omega} $ in $ \fdr $ and an induced interpretation $\rho$ such that, for each $ i $, $ \fmd, \rho \not \models \sq_i $. 
% 		If $ \fbr $ is finite we reach a contradiction: by definition of  $ \htcl $ proofs, every leaf of $ \fbr $ is labelled by the hypersequent $ \str{\,}{\varnothing} $, and by definition of $ \fbr $, we have that  $ \fmd, \rho \not \models \fmi{\str{\,}{\varnothing}} = \top $.
% 		Suppose $ \fbr $ is infinite. By definition of $ \htcl $ proofs (\Cref{def:progress_tcl}), $ \fbr $ needs to contain a \emph{progressing hypertrace}, that is, an hypertrace which is infinite and such that all infinite traces going through it are progressing. Let $ \fhy = (\str{\Gamma_i}{\bv_i})_{i < \omega}$ be such an hypertrace. Fix an  assignment $ \delta_\hyp $, defined in \Cref{def:assignment_h}. 		
		Fix the infinite trace $ \ftr = (F_i)_{i<\omega} $ through $ \fhy $ obtained by Prop.~\ref{prop:failed_trace}.
		Since $\ftr$ is infinite, by definition of $\htcl$ proofs, it needs to be {progressing}, i.e., it is infinitely often $\coTC$-principal
% 		$ \cotc{A}{s}{t} $ which is infinitely often principal. Moreover, by inspection of the rules, starting from a position $ k $, $\ftr$ becomes constant, i.e., it holds that 
% By a simple well-foundedness argument, and by inspection of the rules, 
and
there 
% must be
is
some $k\in \Nat$ s.t.\ for $i>k$ we have that $F_i = \cotc A{s_i}{t_i}$ for some terms $s_i,t_i$.
%		\anupam{replace $k<i$ subscript with $i>k$. It is standard to have the indexed parameter to occur as the left argument of any infix relation.}
% 		$
% 		(F_i)_{i>k} = \cotc{A}{s_i}{t_i}
% 		$.
%		This happens because, if the trace were to follow some other subformulas of a formula of the kind $ \cotc{A}{s}{t} $ then, starting from a position $ k $, there would be no way to have formula $ \cotc{A}{s}{t} $ to occur in the trace again. \anupam{?}
%		While formula $ \cotc{A}{s_i}{t_i}$ needs to be infinitely often principal in $ (F_i)_{k<i} $, this does not mean that all occurrences of it are principal. 
		 
		To each $ F_i $, for $i>k$, we associate the natural number $n_i$ measuring the `$\bar A$-distance between $s_i$ and $t_i$'. 
		Formally, $n_i\in \Nat$ is least such that there are $d_0, \dots, d_{n_i} \in \domain{\fmd}$ with $\fint (s) = d_0, \fint(t) = d_{n_i}$ and, for all $i<n_i$, $\fint ,\deltah \models \bar A(d_i,d_{i+1})$.
% 		natural number $n_i$ as follows. Define the \emph{$ \nf A $-distance between $ s_i $ and $ t_i $} as the smallest $ n_i \in \mathbb{N}$ such that:
% 		\begin{align*}
% 		\text{there exists } d_0, \dots, d_{n_i} \in \domain{\fmd} \, \text{ s.t. } \, & \rho(s_i) = d_0 \text{ and } 
% 		\rho (t_i) = d_{n_i} \text{ and}\\
% 		& \text{for all } 0 \leq j < n_i, 
% 		\fmd, \rho, \delta_\hyp \models \nf{A}(d_j, d_{j+1})
% 		%& \text{ for all } 0 \leq j \leq n_i -1\\	
% %		& \text{and } \fmd, \rho_i, \delta_\hyp \models A(d_1, d_2) \\	
% %		& \hspace{1.5cm}\vdots\\
% %		& \text{and } \fmd, \rho_i, \delta_\hyp \models A(d_{n_{i}-1}, d_{n_i}).	
% 		\end{align*}
% 		Thus, each $ n_i $ is counting the number of $ \nf A $-edges occurring in one of the shortest $ \nf A $-path between $ \rho(s_i) $ and $ \rho(t_i) $.	
		Our aim is to show that $ (n_i)_{i > k} $ has no minimal element, contradicting wellfoundness of $ \Nat $. 
For this, we establish the following two local properties:
		\begin{enumerate}
			\item $ (n_i)_{i> k} $ is \emph{monotone decreasing}, i.e., for all $ i > k $, we have $ n_{i+1}\leq n_i $; \label{proof:mon_dec}
			\item Whenever $ F_i$ is principal, we have
			 $ n_{i+1} < n_i $.  \label{proof:str_dec}
		\end{enumerate} 

		We prove \cref{proof:mon_dec} and \cref{proof:str_dec} by inspection on $ \htcl $ rules. We start with \cref{proof:str_dec}. Suppose $ F_i = \cotc{A}{s}{t}$ is the principal formula in an occurrence of $ \nf{\mathsf{tc}} $ so $ F_{i+1} = \cotc{A}{f(\bv)}{t} $, for some $\bv$. 
		Moreover, by construction $ \fint, \delta_\hyp \models \tc{\nf A}{s}{t} $ and $ \fint, \delta_\hyp \models \tc{\nf A}{f(\bv)}{t}$. We have to show that $n_i$, the $\nf{A}$-distance between $f(\bv)$ and $t$,  is strictly smaller than $n_{i+1}$, the $\nf{A}$-distance between $s$ and $t$, under $\fint$ and $\deltah$. 
		
		Let $\deltah(\bv) = \vec{d}$. By case $(\cotcrule)$ of  Lem.~\ref{lem:countermodel-branch}, and since $ \fint, \delta_\hyp \models \tc{\nf A}{f(\bv)}{t}$, 
		there is a shortest $\nf A$-path between $\fint(s)$ and $\fint(t)$, composed of $n > 1$ elements $d_0, \dots, d_n$, with $\fint(s) = d_0$ and $\fint(t) = d_n$, and 
		$\fint(f)(\vec{d}) = d_1$. Consequently, it holds that $\fint , \deltah \models \nf{A}(s, f(\bv))$ and $\fint, \deltah  \models \tc{\nf A}{ f(\bv)}{t}$. Thus, there is an $ \nf A $-edge between $ \fint(s) $ and $ \fint(f)(\vec{d}) $, and $ \fint(f)(\vec{d}) $ is one edge closer to $ \fint(t) $ in one of the shortest $ \nf A $-paths between $ \fint(s)  $ and $ \fint(t) $. We conclude that $ n_{i+1} $ is exactly $ n_i -1 $, and $ n_{i+1} < n_i $.
		
		To prove \cref{proof:mon_dec}, suppose that $  F_i  $ is not principal in the occurrence $ r_i $ of a $ \htcl $ rule. 
		Suppose $ \infrule_i  $ is $ \mathsf{inst} $, $ F_i = \cotc{A}{s}{x} $, $ F_{i+1} = \cotc{A}{s}{t} $, and $ x $ gets instantiated with $ t $ by $ \mathsf{inst} $. 
		By construction, $ \fint, \delta_\hyp \models \tc{\nf A}{s}{x} $. Let $ n_i $ be the distance between $ \fint(s)$ and $ \delta_\hyp(x) $. By definition, $  \delta_\hyp(x) = \fint(t) $. Thus, the distance between $ \fint(s) $ and $ \fint(t) $ is $n_{i+1} =  n_i $. 
		In all the other cases, $ F_i = \cotc{A}{s}{t} = F_{i+1}  $, and thus $ n_{i+1} = n_i $.   

%		\noindent
		So $ (n_i)_{i > k} $ is monotone decreasing (by point 1) but cannot converge by point 2, and by definition of progressing trace.
		Thus $ (n_i)_{k <i} $ has no minimal element,
%		is \emph{strictly decreasing} \anupam{i think you mean `does not converge'. `strictly decreasing' means everywhere locally strictly decreasing.},
        % and we obtain a contradiction with the 
        yielding the required contradiction.
        \end{proof}

\section{Completeness for $\PDLtr$, over the standard translation}
\label{sec:htcl_simulates_pdltr}

In this section we give our next main result:
\begin{theorem}
    [Completeness for $\PDLtr$]
    \label{thm:completeness-htcl-pdltr}
    For a $\PDLtr$ formula $A$, if $\models A$ then $\htcl \dercyc \st c A$.
\end{theorem}
The proof is by a direct simulation of a cut-free cyclic system for $\PDLtr$ that is complete.
We shall briefly sketch this system 
%now.
below. 

\subsection{Circular system for $\PDLtr$}
\label{ssec:sequents_pdltr}
%%%%%%%%%%%%%%%%%%%%%%%%%%%%%%%%%%%%%%%%%%%%%%%%%%%%%%%%%%%%%%%%%%%%%%%%%%%

    The system $\lpdltr$, shown in \Cref{eq:pdl-sequent-system}, is the natural extension of the usual sequent calculus for basic multimodal logic $K$
    \anupam{any possible reference here?}
    \marianna{No idea. What about a ref for PDL?}
    by rules for programs. 
    In \Cref{eq:pdl-sequent-system}, $\Gamma,\Delta$ etc.\ range over sets of $\PDLtr$ formulas, and we write $\diaP a \Gamma$ is shorthand for $\{\diaP a B : B \in \Gamma \} $.
%   Rules such as $\casezero \vlor$ and $\caseone\vlor$ are 
%     often 
%     abbreviated to $\vlor$ when it is unambiguous.
    (Regular) preproofs for this system are defined just like for $\htcl$ or $\ltcl$.

\begin{figure}[t]
	%\begin{equation}
	%   \scriptsize
	\begin{center}
		\begin{tabular}{|@{\hspace{0.2cm}}c@{\hspace{0.2cm}}|}
			\hline  \\
			$  \vlinf{\mathsf{id}}{\,}{p, \bar{p}}{} $
			\quad 
			$\vlinf{\wk}{}{\magenta \Gamma, A}{\magenta\Gamma}$
			\quad 
			$  \vlinf{\krule a}{}{ \magenta{\diaP{a}\Gamma},\blue{\boxP{a}A}}{ \magenta\Gamma, \blue A}   $
			\quad 
			$\vliinf{\vlan}{}{\magenta \Gamma, \blue{A \vlan B}}{\magenta\Gamma, \blue A}{\magenta\Gamma, \blue B}$ 
			\quad 
			$ \vlinf{\casezero \vlor}{}{\magenta \Gamma, \blue{A_0 \vlor A_1}}{\magenta \Gamma, \blue{A_0}} $
			\quad 
			$  \vlinf{\caseone \vlor}{}{\magenta \Gamma, \blue{A_0 \vlor A_1}}{\magenta\Gamma, \blue{A_1}} $
			\\[0.5cm]
			$ \vlinf{\diaP{ \comp }}{}{ \magenta\Gamma, \blue{\diaP{\alpha \comp \beta}A}}{ \magenta\Gamma, \blue{\diaP{\alpha}\diaP{\beta}A}} $
			\quad 
			$ \vlinf{\casezero{\diaP{ \union }}}{}{ \magenta \Gamma, \blue{\diaP{\alpha_0 \cup \alpha_1}A}}{ \magenta \Gamma, \blue{\diaP{\alpha_0}A}} $
			\quad 
			$ \vlinf{\caseone{\diaP{ \union }}}{}{\magenta \Gamma, \blue{\diaP{\alpha_0 \cup \alpha_1}A}}{ \magenta\Gamma, \blue{\diaP{\alpha_1}A}} $
			\quad 
			$ \vliinf{\boxP{ \union }}{}{ \magenta\Gamma ,\blue{\boxP{\alpha \cup \beta}A}}{ \magenta\Gamma, \blue{\boxP{\alpha}A}}{ \magenta\Gamma, \blue{\boxP{\beta}A}} $\\[0.5cm]
			$\vlinf{\boxP{ \comp }}{}{ \magenta \Gamma, \blue{\boxP{\alpha \comp \beta}A}}{ \magenta \Gamma, \blue{\boxP{\alpha}\boxP{\beta}A}}$
			\quad 
			$\vlinf{\casezero{\diaP{+}}}{}{ \magenta\Gamma, \blue{\diaP{{\alpha^+}}A}}{ \magenta\Gamma, \blue{\diaP{\alpha}A} } $
			\quad 
			$\vlinf{\caseone{\diaP{+}}}{}{ \magenta\Gamma, \blue{\diaP{{\alpha^+}}A}}{ \magenta\Gamma, \blue{\diaP{\alpha}\diaP{{\alpha^+}}A} } $
			\quad 
			$\vliinf{\boxP{+}}{}{ \magenta \Gamma, \blue{\boxP{{\alpha^+}}A} }{ \magenta\Gamma, \blue{\boxP\alpha A}}{ \magenta \Gamma,\blue{\boxP{\alpha}\boxP{{\alpha^+}}A} }  $\\[0.5cm]
			\hline 
		\end{tabular}
	\end{center}
	\caption{Rules of $\lpdltr$.}
	\label{eq:pdl-sequent-system}
\end{figure}
%\end{equation}
% \noindent
    
 The notion of %`immediate ancestor' 
    ancestry for formulas is colour-coded in  \Cref{eq:pdl-sequent-system} as before: a formula $C$ in a premiss is an immediate ancestor of a formula $C'$ in the conclusion if they have the same colour; if $C,C'\in \Gamma$ then we furthermore require $C=C'$.
More formally (and without relying on colours):
\anupam{again, put colour-definition first for intuition}
\begin{definition}
	[Immediate ancestry]
% 	[Ancestry, formulas]
\label{def:ancestry_fml_pdltr}
	\label{dfn:pdl-ancestry}
	Fix a preproof $\derd$. We say that a formula occurrence $C$ is an \emph{immediate ancestor} of a formula occurrence $D$ in $\derd$ if $C$ and $D$ occur in a premiss and the conclusion, respectively, of an inference step $\infrule$ of $\derd$ and:
	\begin{itemize}
	\item If $\infrule$ is a $k$ step then, as typeset in \Cref{eq:pdl-sequent-system}:
	\begin{itemize}
	    \item $D$ is $\diaP a B$ for some $B\in \Gamma$ and $C$ is $B$; or,
	    \item $D$ is $\boxP a A$ and $C$ is $A$.
	\end{itemize} 
	    \item If $\infrule$ is not a $k$-step then:
	    \begin{itemize}
	        \item $C$ and $D$ are occurrences of the same formula; or,
	        \item $D$ is principal and $C$ is auxiliary in $\infrule$, i.e.\ as typeset in \Cref{eq:pdl-sequent-system}, $C$ and $D$ are the (uniquely) distinguished formulas in a premiss and conclusion, respectively;
	    \end{itemize}
	\end{itemize}
\end{definition}

\begin{definition}
    [Non-wellfounded proofs]
    \label{dfn:pdl-threads-proofs}
    Fix a preproof $\derd$ of a sequent $\Gamma$. 
    A \emph{thread} is a maximal path in its graph of immediate ancestry. 
    We say a thread is \emph{progressing} if it has a smallest infinitely often principal formula of the form $\boxP{\alpha^+}A$.
    $\derd$ is a \emph{proof} if every infinite branch has a progressing thread.
    If $\derd$ is regular, we call it a \emph{cyclic proof}
    and we may write $\lpdltr \dercirc \Gamma$.
%    
    % If there is a circular proof in $\lpdltr$ of $A$ then we write $\lpdltr \dercirc A$.
\end{definition}

Soundness of cyclic-$\lpdltr$ is established by a standard infinite descent argument, but is also implied by the soundness of cyclic-$\htcl$ (Thm.~\ref{thm:soundness_tcl}) and the simulation we are about to give (Thm.~\ref{thm:completeness-htcl-pdltr}), though this is somewhat overkill.
Completeness may be established by the game theoretic approach of Niwinsk\'i and Walukiewicz \cite{NIWINSKI1996}, as carried out by \cite{lange2003games}, or by the purely proof theoretic techniques of Studer \cite{studer2008proof}.
Either way, both results follow immediately by a standard embedding of $\PDLtr$ into the (guarded) $\mu$-calculus and its known completeness results \cite{NIWINSKI1996,studer2008proof}, by way of a standard `proof reflection' argument: $\mu$-calculus proofs of the embedding are `just' step-wise embeddings of $\lpdltr$ proofs. 

\marianna{added ref to Lange}
\anupam{Great, that's perfect!}
\begin{theorem}
	[Soundness and completeness, \cite{lange2003games}]
	\label{pdl-soundness-completeness}
	Let $A$ be a $\PDLtr$ formula.
	$ \models A$ iff $\lpdltr \dercirc A$.
\end{theorem}

\subsection{Examples of cyclic proofs in $\lpdltr$}
Before giving our main simulation result, let us first see some examples of proofs in $\lpdltr$, in particular addressing the `counterexample' from \Cref{ssec:counterexample}.

\begin{example}
	\label{ex:pdltr_1}
	We show a cyclic $\lpdltr$  proof of the $ \lpdltr $ sequent:
	\begin{equation*}
	%\label{eq:counterexample-formula-pdltr}
	\boxP{(aa\cup aba)^+} \nf p ,  \diaP{a^+ }p , \diaP{(ba^+)^+ }p ,  \diaP{a^+}\diaP{(ba^+)^+ }p
	\end{equation*}
	We use the following abbreviations: $\alpha =(aa\cup aba)^+$ and $\beta  = (ba^+)^+$. Moreover, we sometimes use rule $ \diaP{+} $, which is derivable from rules $\diaP{+}_0$ and $ \diaP{+}_1 $, keeping in mind that $ \lpdltr $ sequents are \emph{sets} of formulas:
	$$
	\vlinf{\diaP{+}}{}{\Gamma, \diaP{\alpha^+}A}{\Gamma, \diaP{\alpha}A, \diaP{\alpha} \diaP{\alpha^+}A}
	$$
	Similarly for rule $\vlor$. The progressing threads (one for each infinite branch) are highlighted in blue. $ \Gamma $ is sequent $  \boxP{aa \cup aba} \nf p,  \diaP{a^+ }p ,  \diaP{\beta}p ,  \diaP{a^+}\diaP{\beta }p $, derivable by means of a finite derivation (which we do not show). 
	%\vspace{-0.5cm}
	
	\begin{adjustbox}{max width= \textwidth}
	$
	\vlderivation{
	%\vlin{\vlor}{}{ \violet{\boxP{\alpha} \nf p \vlor  \diaP{a^+ }p \vlor  \diaP{\beta }p \vlor  \diaP{a^+}\diaP{\beta}p} }{
	\vliin{\boxP{+}}{\circ}{
		\blue{\boxP{\alpha} \nf p} ,  \diaP{a^+ }p ,  \diaP{\beta}p ,  \diaP{a^+}\diaP{\beta }p }{	\vlhy{\Gamma} }{
	\vliin{\boxP{\cup}}{}{ \blue{\boxP{aa \cup aba}\boxP{\alpha} \nf p} ,  \diaP{a^+ }p ,  \diaP{\beta}p ,  \diaP{a^+}\diaP{\beta }p }{
	\vlin{\boxP{\comp} }{}{ \blue{\boxP{aa }\boxP{\alpha} \nf p} ,  \diaP{a^+ }p ,  \diaP{\beta}p ,  \diaP{a^+}\diaP{\beta }p }{
	\vlin{\diaP{+}_1}{}{ \blue{\boxP{a}\boxP{a }\boxP{\alpha} \nf p} ,  \diaP{a^+ }p ,  \diaP{\beta}p ,  \diaP{a^+}\diaP{\beta }p  }{
	\vlin{k_a}{}{ \blue{\boxP{a}\boxP{a }\boxP{\alpha} \nf p },  \diaP{a}\diaP{a^+ }p ,  \diaP{\beta}p ,  \diaP{a}\diaP{a^+}\diaP{\beta }p  }{
	\vlin{\diaP{+}}{}{  \blue{\boxP{a }\boxP{\alpha} \nf p} ,  \diaP{a^+ }p , \diaP{a^+}\diaP{\beta }p  }{
	\vlin{\diaP{+}_1}{}{ \blue{ \boxP{a }\boxP{\alpha} \nf p} ,  \diaP{a^+ }p ,\diaP{a}\diaP{\beta }p, \diaP{a}\diaP{a^+}\diaP{\beta }p  }{
	\vlin{k_a}{}{ \blue{\boxP{a }\boxP{\alpha} \nf p} , \diaP{a} \diaP{a^+ }p ,\diaP{a}\diaP{\beta }p, \diaP{a}\diaP{a^+}\diaP{\beta }p }{
	\vlin{}{\circ}{ \blue{\boxP{\alpha} \nf p} , \diaP{a^+ }p ,\diaP{\beta }p, \diaP{a^+}\diaP{\beta }p }{\vlhy{\vdots}}
	}
	}
	}
	}
	}
	}
	}{
	\vlin{2\boxP{\comp}}{}{{\boxP{aba}\boxP{\alpha} \nf p} ,  \diaP{a^+ }p ,  \diaP{\beta}p ,  \diaP{a^+}\diaP{\beta }p }{
	\vlin{\diaP{+}_0}{}{{\boxP{a}\boxP{b}\boxP{a}\boxP{\alpha} \nf p} ,  \diaP{a^+ }p ,  \diaP{\beta}p ,  \diaP{a^+}\diaP{\beta }p}{
	\vlin{k_a}{}{ {\boxP{a}\boxP{b}\boxP{a}\boxP{\alpha} \nf p} ,  \diaP{a^+ }p ,  \diaP{\beta}p ,  \diaP{a}\diaP{\beta }p  }{
	\vlin{\diaP{+}}{}{ {\boxP{b}\boxP{a}\boxP{\alpha} \nf p} ,  \diaP{\beta }p  }{
	\vlin{2\diaP{\comp}}{}{ {\boxP{b}\boxP{a}\boxP{\alpha} \nf p} , \diaP{b a^+ }p, \diaP{b a^+}\diaP{\beta }p   }{
	\vlin{k_b}{}{  {\boxP{b}\boxP{a}\boxP{\alpha} \nf p} , \diaP{b}\diaP{ a^+ }p, \diaP{b}\diaP{ a^+}\diaP{\beta }p     }{
	\vlin{\diaP{+}}{}{   {\boxP{a}\boxP{\alpha} \nf p} , \diaP{ a^+ }p, \diaP{ a^+}\diaP{\beta }p     }{
	\vlin{\diaP{+}_1}{}{  {\boxP{a}\boxP{\alpha} \nf p} , \diaP{ a^+ }p,  \diaP{ a}\diaP{\beta }p,\diaP{a}\diaP{ a^+}\diaP{\beta }p     }{
	\vlin{k_a}{}{  {\boxP{a}\boxP{\alpha} \nf p} , \diaP{a}\diaP{ a^+ }p,  \diaP{ a}\diaP{\beta }p,\diaP{a}\diaP{ a^+}\diaP{\beta }p     }{
	\vlin{}{\circ}{  {\boxP{\alpha} \nf p} , \diaP{ a^+ }p,  \diaP{\beta }p,\diaP{ a^+}\diaP{\beta }p   }{\vlhy{\vdots}}
	}
	}
	}
	}
	}
	}
	}
	}
	}
	}
	}
	}
%	}
	$
	\end{adjustbox}
\end{example}

\begin{example}
We show a cyclic  $\lpdltr$ proof of formula \eqref{eq:counterexample-formula-pdltr}, which witnesses the incompleteness of $\ltcl$ without cut:
\begin{equation*}
%\label{eq:counterexample-formula-pdltr}
\diaP{(aa\cup aba)^+} p \,  \vljm \, \diaP{a^+((ba^+)^+ \cup a)}p
\end{equation*}
We employ the same shorthands as in the previous example, i.e., $\alpha =(aa\cup aba)^+$ and $\beta  = (ba^+)^+$. The progressing threads are highlighted in blue; in the rightmost branch, the thread continues into the thread shown in \Cref{ex:pdltr_1}, where progress can be observed. 
%\[
%\begin{array}{ c c c}
%\alpha & = & (aa\cup aba)^+\\
%\beta & = & (ba^+)^+ 
%\end{array}
%\]
$$
\vlderivation{
	\vlin{\vlor}{}{\blue{\boxP{\alpha} \nf{p} \vlor \diaP{a^+ (\beta \cup a)}p} }{
		\vlin{\diaP{\comp }}{}{\blue{\boxP{\alpha} \nf{p}} , \diaP{a^+ (\beta \cup a)}p}{
			\vliin{\boxP{+}}{\bullet}{ \blue{\boxP{\alpha} \nf{p}} , \diaP{a^+}\diaP{ \beta \cup a}p }{
				\vlhy{\Gamma'
					%\boxP{aa \cup aba} \nf{p} , \diaP{a^+}\diaP{ (\beta \cup a)}p
				}
			}{
				\vliin{\boxP{\cup}}{}{\blue{\boxP{aa \cup aba}\boxP{\alpha} \nf{p}} , \diaP{a^+}\diaP{ \beta \cup a}p}{
					\vlin{\boxP{\comp}}{}{ \blue{\boxP{aa}\boxP{\alpha} \nf{p} }, \diaP{a^+}\diaP{ \beta \cup a}p }{
						\vlin{\diaP{+}_1}{}{\blue{\boxP{a}\boxP{a}\boxP{\alpha} \nf{p} }, \diaP{a^+}\diaP{ \beta \cup a}p}{
							\vlin{k_a}{}{ \blue{\boxP{a}\boxP{a}\boxP{\alpha} \nf{p} }, \diaP{a}\diaP{a^+}\diaP{ \beta \cup a}p }{
								\vlin{\diaP{+}_1}{}{\blue{\boxP{a}\boxP{\alpha} \nf{p}} , \diaP{a^+}\diaP{ \beta \cup a}p}{
									\vlin{k_a}{}{\blue{\boxP{a}\boxP{\alpha} \nf{p}} , \diaP{a}\diaP{a^+}\diaP{ \beta \cup a}p}{
										\vlin{}{\bullet}{\blue{ \boxP{\alpha} \nf{p}} , \diaP{a^+}\diaP{ \beta \cup a}p }{\vlhy{\vdots}}
									}
								}
							}
						}
					}
				}{ 
					\vlin{2\boxP{\comp}}{}{\blue{\boxP{aba }\boxP{\alpha} \nf{p}} , \diaP{a^+}\diaP{ \beta \cup a}p}{  
					\vlin{\diaP{+}_0}{}{ \blue{\boxP{a}\boxP{b}\boxP{a }\boxP{\alpha} \nf{p} }, \diaP{a^+}\diaP{ \beta \cup a}p }{
					\vlin{k_a}{}{\blue{ \boxP{a}\boxP{b}\boxP{a }\boxP{\alpha} \nf{p}} , \diaP{a}\diaP{ \beta \cup a}p  }{
					\vlin{\diaP{\cup}_0}{}{\blue{\boxP{b}\boxP{a }\boxP{\alpha} \nf{p} }, \diaP{ \beta \cup a}p}{
					\vlin{\diaP{+}}{}{ \blue{\boxP{b}\boxP{a }\boxP{\alpha} \nf{p} }, \diaP{ \beta }p }{
					\vlin{2\diaP{\comp}}{}{\blue{\boxP{b}\boxP{a }\boxP{\alpha} \nf{p} }, \diaP{ ba^+ }p,\diaP{ ba^+ }\diaP{ \beta }p  }{
					\vlin{k_b}{}{\blue{\boxP{b}\boxP{a }\boxP{\alpha} \nf{p}} , \diaP{ b}\diaP{a^+ }p,\diaP{ b}\diaP{a^+ }\diaP{ \beta }p }{
					\vlin{\diaP{+}}{}{\blue{ \boxP{a }\boxP{\alpha} \nf{p} }, \diaP{a^+ }p,\diaP{a^+ }\diaP{\beta }p }{
					\vlin{\diaP{+}_1}{}{\blue{ \boxP{a }\boxP{\alpha} \nf{p} }, \diaP{a^+ }p,\diaP{a}\diaP{\beta }p,\diaP{a}\diaP{a^+ }\diaP{\beta }p }{
					\vlin{k_b}{}{\blue{ \boxP{a }\boxP{\alpha} \nf{p}} , \diaP{a}\diaP{a^+ }p,\diaP{a}\diaP{\beta }p,\diaP{a}\diaP{a^+ }\diaP{\beta }p}{
					\vliq{}{}{ \blue{\boxP{\alpha} \nf{p} }, \diaP{a^+ }p,\diaP{\beta }p,\diaP{a^+ }\diaP{\beta }p  }{\vlhy{\text{\Cref{ex:pdltr_1}}}}
					}
					}
					}
											}
										}
									}
								}
							}
						}
					}
				}
			}
		}
	}
}
$$
Here follows the finite derivation of $\Gamma'$, i.e., sequent $\boxP{aa \cup aba} \nf{p} , \diaP{a^+}\diaP{ \beta \cup a}p$. 
$$
\vlderivation{
\vliin{\boxP{\cup}}{}{ \boxP{aa \cup aba} \nf{p} , \diaP{a^+}\diaP{ \beta \cup a}p }{
\vlin{\boxP{\comp}}{}{ \boxP{aa } \nf{p} , \diaP{a^+}\diaP{ \beta \cup a}p }{
\vlin{\diaP{+}_0}{}{\boxP{a}\boxP{a } \nf{p} , \diaP{a^+}\diaP{ \beta \cup a}p}{
\vlin{k_a}{}{ \boxP{a}\boxP{a } \nf{p} , \diaP{a}\diaP{ \beta \cup a}p  }{
\vlin{\diaP{\cup}}{}{ \boxP{a } \nf{p} , \diaP{ \beta \cup a}p  }{
\vlin{k_a}{}{ \boxP{a } \nf{p} , \diaP{  a}p}{
\vlin{\init}{}{  \nf{p} , p  }{\vlhy{}}
}
}
}
}
}
}{
\vlin{2\boxP{\comp}}{}{ \boxP{ aba} \nf{p} , \diaP{a^+}\diaP{ \beta \cup a}p  }{
\vlin{\diaP{+}_0}{}{ \boxP{ a}\boxP{b}\boxP{a} \nf{p} , \diaP{a^+}\diaP{ \beta \cup a}p   }{
\vlin{k_a}{}{ \boxP{ a}\boxP{b}\boxP{a} \nf{p} , \diaP{a}\diaP{ \beta \cup a}p    }{
\vlin{\diaP{\cup}}{}{ \boxP{b}\boxP{a} \nf{p} , \diaP{\beta \cup a}p    }{
\vlin{\diaP{+}_0}{}{  \boxP{b}\boxP{a} \nf{p} , \diaP{ \beta }p     }{
\vlin{\diaP{\comp}}{}{  \boxP{b}\boxP{a} \nf{p} , \diaP{ b a^+ }p    }{
\vlin{k_b}{}{  \boxP{b}\boxP{a} \nf{p} , \diaP{ b}\diaP{ a^+ }p }{
\vlin{\diaP{+}_0}{}{   \boxP{a} \nf{p} , \diaP{ a^+ }p}{
\vlin{k_a}{}{  \boxP{a} \nf{p} , \diaP{ a }p}{
\vlin{\init}{}{\nf p, p}{\vlhy{}}
}
}
}
}
}
}
}
}
}
}
}
$$

\end{example}

\subsection{A `local' simulation of $\lpdltr$ by $\htcl$}
\label{ssec:pdltr_to_htcl}

In this subsection we show that $\lpdltr$-preproofs can be stepwise transformed into $\htcl$-proofs, with respect to the standard translation. 
In order to produce our local simulation, we need a more refined version of the standard translation that incorporates the structural elements of hypersequents.

\marianna{commented the 'simplified' version of HT. Maybe put a footnote in the IJCAR paper observing the difference between the two?}
	
% Fix a $ \PDLtr $ formula $A =\boxP{\alpha_1} \dots \boxP{\alpha_n}\diaP{\beta_1}\dots \diaP{\beta_m} B $, for  $ n,m \geq 0 $. %\anupam{$\geq$}
% 	The \emph{hypersequent translation} of $ A $, written $ \ct{A}{\cons{c}} $, is defined as: 
% 	\begin{align*}
% % 	 \ct{A}{\cons{c}} \, \defsym \quad & 
% 	 & \str{ \nf{\st{\cons{c},\cons{d}_1}{\alpha_1}} }{\varnothing}, \str{ \nf{\st{\cons{d}_1,\cons{d}_2}{\alpha_2}} }{\varnothing}, \dots,
% 	 \str{ \nf{\st{\cons{d}_{n-1},\cons{d}_n}{\alpha_n}} }{\varnothing}, \\
% % 	 &  
% 	 & \str{\st{\cons{d}_n,y_1}{\beta_1},\st{y_2,y_3}{\beta_2},\dots, \st{y_{m-1},y_m}{\beta_m}, \st{y_m}{B}}{y_1, \dots, y_m} 
% 	\end{align*}
% 	For $ \Gamma = A_1, \dots, A_k $, we write
% 	$ \ct{\Gamma}{\cons{c}} \defsym \ct{A_1}{\cons{c}}, \dots, \ct{A_k}{\cons{c}} $. 

We start by refining the definition of the hypersequent translation of a $\PDLtr$ formula from \cref{ssec:pdltr_to_htcl}, including the explicit definition of the cedent translation, $\mathsf{CT}$.  
In the definitions below we are often using arbitrarily chosen variables (e.g. ``fresh variables'') and constant symbols. We assume these choices do not break our assumptions on variables and constants occurring in hypersequents.

%\anupam{below, and throughout, make some point somewhere early that many definitions involve arbitrary choices, including choices of `fresh' variables. We assume these are made canonically.} \marianna{done}

%\marianna{14/02: fixed this definition and the proposition}

\begin{definition}[Hypersequent translation for formulas]
\label{def:hyp_transl}
	For $ x $ a variable and $ A $  a $   \PDLtr $ formula, we define the \emph{hypersequent translation} of $ A $, denoted by $ \ct{A}{\cons{c}} $, by induction on the complexity of $A$ as follows, for $\circ \in \{\vlan, \vlor\}$, $d$ fresh constant symbol and $y$ fresh variable: 
\begin{eqnarray*}
A = p \text{ or } A = B \circ C &   \ct{A}{x} & \defsym \quad  \str{\st{x}{A}}{\varnothing}\\
 A = \boxP{\alpha}B &  \ct{A}{x} & \defsym \quad  \str{\nf{\st{x,\cons d}{\alpha}}}{\varnothing}, \ct{B}{\cons{d}}\\
A = \diaP{\alpha}B &   \ct{A}{x} & \defsym \quad  \str{\st{x,y}{\alpha},\diat{B}{y}}{\bvvar{B},y}
\end{eqnarray*}
where the \emph{cedent translation} $ \diat{B}{t} $  and $\bvvar{B}$ are  defined as follows:
\begin{itemize}
    \item 
    \label{prop:diat_1}
    if $B = p$ or $B = F \circ G $ or $B = \boxP{\beta}F$, then:
    \begin{eqnarray*}
    \diat{B}{t} \, \defsym \, \st{y}{B} & \text{and} \quad \bvvar{B} = \varnothing %& \text{ if } t = y%\\
%    \diat{B}{t} \, \defsym \, \st{c}{B} & \text{and} \quad \bvvar{B} = \varnothing \,\, & \text{ if } t = c
    \end{eqnarray*}
    \item 
    \label{prop:diat_2}
    if $B = \diaP{\beta}C$ then, for $z$ fresh variable:
    $$
    \diat{B}{y} \, \defsym \, \st{y, z}{\beta}, \diat{C}{z} \quad \text{and} \quad \bvvar{B}  = \bvvar{C}, z.
    $$
\end{itemize}
\end{definition}

%\anupam{in general, the above approach is good and formal.}

\begin{proposition}
\label{prop:ht_st_tcl}
    For $ \sq $ a hypersequent, and $ A$ a $\PDLtr$ formula:
    \begin{enumerate}
    % 	\item If  $ \htcl \dercyc \sq, \ct{A}{\cons{c}}$, then  there is a finite and cut-free proof of $ \sq, \str{\diat{A}{c}}{\bvvar{A}}$;
    	\item There is a finite cut-free $\htcl$-derivation from $\sq, \ct{A}{\cons{c}}$ to $ \sq, \str{\diat{A}{c}}{\bvvar{A}}$; and,
    % 	\item If  $ \htcl \dercyc \sq, \str{\diat{A}{c}}{\bvvar{A}}$, then there is a finite and cut-free proof of $ \sq, \str{\st{c}{A}}{\varnothing}$.
    	\item There is a finite cut-free derivation from $\sq, \str{\diat{A}{c}}{\bvvar{A}}$ to $ \sq, \str{\st{c}{A}}{\varnothing}$.
    \end{enumerate} 
\end{proposition}

\begin{proof}
    By induction on the complexity of $A$. %, and by straightforward application of $\htcl$ rules. 
    If $A $ is atomic, a conjunction or a disjunction, then  $\ct{A}{\cons{c}} = \str{\diat{A}{c}}{\bvvar{A}} = \str{\st{\cons{c}}{A}}{\varnothing} 	 $. 
    The remaining cases are shown below. If $A= \boxP{\alpha} B$, then $\str{\diat{A}{c}}{\bvvar{A}} = \str{\st{c}{A}}{\varnothing}$. 
    % in $ \sq, \str{\diat{A}{y}}{\bvvar{A}}$ we have that $ \bvvar{A} = \{y\}$.    
    %it holds that $ \str{\diat{A}{\cons{c}}}{\bvvar{A}} = \str{\st{\cons{c}}{A}}{\varnothing} $. 
    If $A = \diaP{\alpha_1}\dots \diaP{\alpha_n}B$, for $ 1 \leq n$,  
    then $  \str{\diat{A}{y}}{\bvvar{A}} = \ct{A}{c}$, and $ \bvvar{A} = \{y_1, \dots, y_n\}$, for $ y_1, \dots, y_n $ variables that do not occur in $ B $. The double line in the derivation below denotes $ n -1$ occurrences of rules $ \exists, \vlan $.    
    
    \ 

   \begin{adjustbox}{max width = \textwidth}
 % \begin{center}
    \begin{tabular}{c @{\hspace{0.3cm}} c}
    $    \vlderivation{
  %  \vlin{\instrule}{}{\sq,\str{\diat{\boxP{\alpha}B}{y}}{y}}{
    \vlid{}{}{
   	    \sq,\str{\st{\cons{c}}{\boxP{\alpha}B}}{\varnothing}
   	}{
    \vlid{}{}{
	\sq,\str{\diat{\boxP{\alpha}B}{\cons{c}}}{\varnothing}
	}{
    \vlin{\forall}{}{\sq,\str{ \forall y (  \nf{\st{\cons{c}, y }{\alpha}} \vlor \st{y}{B} )  }{\varnothing}}{
    \vlin{\vlor}{}{\sq, \str{ \nf{\st{\cons{c}, \cons{d} }{\alpha}} \vlor \st{\cons{d}}{B}}{\varnothing} }{
   % \vlin{\sigma}{}{ \sq, \str{ \nf{\st{\cons{c}, y }{\alpha}}}{\varnothing}, \str{\st{y}{B}}{\varnothing} }{
    \vlid{}{}{\sq, \str{ \nf{\st{\cons{c}, \cons{d} }{\alpha}}}{\varnothing}, \str{\st{\cons{d}}{B}}{\varnothing}   }{
    \vlhy{\sq,\ct{\boxP{\alpha} B}{\cons{c}} }
    }
    }
    }
    }
    }
	}
%	}
    $
    &
    $
     \vlderivation{
    	\vlid{}{}{\sq,\str{\st{\cons{c}}{\diaP{\alpha_1}\dots \diaP{\alpha_n}B}}{\varnothing}}{
    		\vlin{\exists}{}{\sq,\str{ \exists y_1 (  \st{\cons{c}, y_1 }{\alpha_1} \vlan \st{y_1}{\diaP{\alpha_2}\dots \diaP{\alpha_n}B} )  }{y_1}}{
    			\vlin{\vlan}{}{\sq, \str{ \st{\cons{c}, y_1 }{\alpha} \vlan \st{y_1}{\diaP{\alpha_2}\dots \diaP{\alpha_n}B}}{y_1} }{
    				\vliq{\exists, \vlan}{}{ \sq,\str{ \st{\cons{c}, y_1 }{\alpha} , \st{y_1}{\diaP{\alpha_2}\dots \diaP{\alpha_n}B}}{y_1}     }{
    				\vlid{}{}{ \sq,\str{ \st{\cons{c}, y_1 }{\alpha_1}, \dots, \st{y_{n-1},y_n}{\alpha_n} , \st{y_n}{B}}{y_1, \dots, y_n} }{
    			%	\vlid{}{}{ \sq,\str{ \st{\cons{c}, y }{\alpha} , \diat{B}{y}}{y} }{
    				\vlid{}{}{ \sq,\str{ \diat{\diaP{\alpha_1}\dots \diaP{\alpha_n}B}{y_n}}{y_1, \dots,y_n}}{
    					\vlhy{\sq,\ct{\diaP{\alpha_1}\dots \diaP{\alpha_n}B}{\cons{c}} }
    				}
    			}
    		}
    	}
    }
}
}
%}
    $\\
    \end{tabular}
%      \end{center}
\end{adjustbox}

\end{proof}

%\anupam{minor point: $\mathsf{HT}$ is only defined up to the choice of fresh constant symbols $\vec{\cons d}$ and bound variables $\vec y$. Perhaps indicating these explicitly would be overkill, we can just gloss over it.}
%	\marianna{not sure I understand. You mean putting $\vec{\cons d}$ and $\vec y$ as indexes in HT(A)(c)?}
%	\marianna{30/12: now I understand. I also don't think that I'm using the above definition correctly in the proof. I tired to update the definition below. }

%
%\begin{proposition}
%\label{prop:ht_st_tcl}
%    For $ \sq $ $ \htcl $ hypersequent, $ A$ $\lpdl$ formula and $\cons{c}$ constant symbol, if $ \htcl \dercyc \sq, \ct{A}{\cons{c}}$, then $ \htcl \dercyc \sq, \str{\st{\cons{c}}{A}}{\varnothing}$. 
%\end{proposition}

% In the definitions below we are often using arbitrarily chosen variables (e.g. ``fresh variables'') and constant symbols. We assume these choices do not break our assumptions on variables and constants occurring in hypersequents.

% \begin{lemma}
% 	\label{lemma:pdlplus_to_htcl}
% % 	Let $ A $ be a $ \PDLtr $ formula and $ \cons{c} $ be a constant. 
% 	If $ \lpdltr \dercyc A $, then it holds that $ \htcl \dercyc \ct{A}{\cons{c}} $. 
% \end{lemma}

% \begin{proof}[Definition of translation]	
\begin{definition}
    [$\HT$-translation]
    \label{dfn:ht-translation}
	Let $ \mathcal{D} $ be a $ \PDLtr $ preproof.
	We shall define a $ \htcl $ preproof $ \htr c \der $
	%\anupam{why not $\mathsf{HT}(\der)$ (or even $\ct \der c$)?} 
	of the hypersequent $  \ct{A}{\cons{c}} $ by a local translation of inference steps.
	Formally speaking, the well-definedness of $\htr c \der$ is guaranteed by coinduction: each rule of $\der$ is translated into a (nonempty) derivation.
%	We give only a few of the important cases here, but a full definition can be found in App.~\ref{ssec_appendix_pdltr_to_htcl}.
	
			\[
	\vlinf{\id}{}{\Gamma,p, \nf{p}}{} 
	\quad \rightsquigarrow \quad 
	\vlderivation{
		\vlid{}{}{\strg, \ct{p}{\cons{c}}, \ct{\nf{p}}{\cons{c}} }{
			\vlin{\wk}{}{\strg, \str{p(\cons{c})}{\varnothing}, \str{\nf{p}(\cons{c})}{\varnothing}}{
				\vlin{\id}{}{ \str{p(\cons{c})}{\varnothing}, \str{\nf{p}(\cons{c})}{\varnothing}}{
					\vlin{\init}{}{\str{\, }{\varnothing}}{\vlhy{}
					}
				}	
			}
		}
	}
	\]

	\todo{add weakening}
	
%	\vspace{0.3cm}
	
		\[
	\vlderivation{
		\vlin{\vlor_i}{i \in\{0, 1\}}{\Gamma, A_0 \vlor A_1}{
			%\vlpr{\mathcal{D}}{}{\Gamma, A_i}
			\vlhy{\Gamma, A_i	%\vltreeder{\der}{\Gamma, A_i}{\ }{}{\ }
			}
		}
	} 
	\quad \rightsquigarrow \quad 
	\vlderivation{
		\vlid{}{}{\strg, \ct{A_0 \vlor A_1}{\cons{c}} }{
		\vlin{\vlor}{}{\strg, \str{ \st{\cons{c}}{A_0} \vlor \st{\cons{c}}{A_1} }{\varnothing}
			}{
		\vliq{\doubleline}{}{ \strg, \str{ \st{\cons{c}}{A_i}}{\varnothing} }{
		\vlhy{	 \strg, \ct{A_i}{\cons{c}}
		}
		}
				}
			}
		}
	%}
	\]

%	\vspace{0.3cm}
	
	\[
	\vlderivation{
		\vliin{\vlan}{}{\Gamma,  A \vlan B}{
			%\vlpr{\mathcal{D}_1}{}{\Gamma, A}
			\vlhy{	\Gamma, A	}
		}{
			%\vlpr{\mathcal{D}_2}{}{\Delta,B}
			\vlhy{	\Gamma,B
			}
		}
	} 
	\quad \rightsquigarrow \quad 
	\vlderivation{
		\vlid{}{}{\strg,  \ct{A \vlan B}{\cons{c}} }{
			\vlin{\vlan}{}{\strg,  \str{ \st{\cons{c}}{A} \vlan \st{\cons{c}}{B} }{\varnothing} }{
				\vliin{\cup}{}{\strg,  \str{ \st{\cons{c}}{A} , \st{\cons{c}}{B} }{\varnothing} }{
					%\vlin{\wk}{}{\strg, \strd ,\str{ \st{\cons{c}}{A} }{\varnothing} }{
					\vliq{\doubleline}{}{\strg ,\str{ \st{\cons{c}}{A} }{\varnothing}}{
					%	\vlpr{t(\mathcal{D}_1)}{}{\strg, \ct{A}{\cons{c}}  }
					\vlhy{	\strg, \ct{A}{\cons{c}} %\vltreeder{\tderp}{  }{\quad }{}{\quad }
					}
					}
					%}
				}{
					\vliq{\doubleline}{}{  \strg , \str{ \st{\cons{c}}{B} }{\varnothing}}{
					%	\vlpr{t(\mathcal{D}_2)}{}{ \strd ,  \ct{B}{\cons{c}}  }
					\vlhy{\strg, \ct{B}{\cons{c}} %	\vltreeder{\tderp}{ }{\quad }{}{\quad }
					}
					}
				}
			}
		}
	}
	\]

  %  \vspace{0.3cm}

    $
    \vlinf{\krule a}{}{ \diaP{a} B_1, \dots, \diaP{a} B_k, \boxP a A}{B_1, \dots , B_k, A} \, \leadsto \, 
    $
    
    \vspace{0.2cm}

    \begin{adjustbox}{max width = \textwidth}
	$
	\vlderivation{ 
		\vlid{=}{}{ \ct{\diaP{a}B_1}{ \cons{c}},\dots, \ct{\diaP{a}B_k}{ \cons{c}}, \ct{\boxP{a}A}{ \cons{c}}  }{
		\vliq{\instrule}{}{ \str{\st{ \cons{c},y}{a} , \diat{B_1}{y}}{\bvvar{B_1}, y}, \dots ,\str{\st{ \cons{c},y}{a} , \diat{B_k}{y}}{\bvvar{B_k}, y}, \str{ \nf{ \st{ \cons{c}, \cons{d}}{a}}  }{\varnothing}, \ct{A}{\cons{d}}
		}
	%	{
%		\vlin{\cup}{}{ \str{\st{ \cons{c}, \cons{d}}{a} , \diat{B_1}{ \cons{d}} }{\bvvar{B_1}}, \dots ,\str{\st{ \cons{c}, \cons{d}}{a} , \diat{B_k}{ \cons{d}}}{\bvvar{B_k}}, \str{ \nf{ \st{ \cons{c}, \cons{d}}{a}}  }{\varnothing}, \ct{A}{\cons{d}}
		%\str{\st{ \cons{d}}{A} }{\varnothing}
%		}
% 		{
% 		 \vlin{\wk, \id}{}{ \str{\st{ \cons{c}, \cons{d}}{a} }{\varnothing} ,  \str{ \nf{ \st{ \cons{c}, \cons{d}}{a}}  }{\varnothing} }{
% 		 \vlin{\init}{}{\str{\,}{\varnothing}}{\vlhy{}}
% 	 	}
% 		}
		{
		\vliq{\cup}{}{
		 \str{ \diat{B_1}{ \cons{d}} }{\bvvar{B_1}}, \dots ,\str{\st{ \cons{c}, \cons{d}}{a} , \diat{B_k}{ \cons{d}}}{\bvvar{B_k}},\str{ \nf{ \st{ \cons{c}, \cons{d}}{a}}  }{\varnothing}, 
		 %\str{\st{ \cons{d}}{A} }{\varnothing}
		 \ct{A}{\cons{d}}
		}{
		\vlin{\wk}{}{  \str{ \diat{B_1}{ \cons{d}}}{\bvvar{B_1}}, \dots ,\str{\diat{B_k}{ \cons{d}}}{\bvvar{B_k}},\str{ \nf{ \st{ \cons{c}, \cons{d}}{a}}  }{\varnothing}, 	 \ct{A}{\cons{d}}
		%\str{\st{ \cons{d}}{A} }{\varnothing}
		}{
		\vliq{\vlor, \forall}{}{ 	\str{ \diat{B_1}{ \cons{d}}}{\bvvar{B_1}}, \dots ,\str{\diat{B_k}{ \cons{d}}}{\bvvar{B_k}},\ct{A}{\cons{d}}
		%\str{\st{\cons{c}}{A} }{\varnothing} 
		}{
		\vlin{[d/c]}{}{\ct{B_1}{\cons{d}}, \dots, \ct{B_k}{\cons{d}}, \ct{A}{\cons{d}}  
		}{
% 		\vlpr{t(\mathcal{D})}{}{
%		 \vlin{\sub{[c/d]}}{}{\ct{B_1}{\cons{d}}, \dots, \ct{B_k}{\cons{d}}, \ct{A}{\cons{d}} }{
\vlhy{\ct{B_1}{\cons{c}}, \dots, \ct{B_k}{\cons{c}}, \ct{A}{\cons{c}} }}
 		}
		}
		}
		}
		}
%		}
%		}	
%		}
%		}
%		}
%		}
%	}
}
	$
	\end{adjustbox}
	
	\ 
	
	where (omitted) left-premisses of $\cup$ steps are simply proved by $\wk,\id,\init$.
%	In this and the following cases, we use the notation $\mathsf{CT}(A)(c)$ and $\vec x_A$ for the appropriate sets of formulas and variables forced by the definition of $\mathsf{HT}$. (A detailed definition of $\mathsf{CT}(A)(c)$ and $\vec x_A$ can be found in the appendix, Dfn.~\ref{def:hyp_transl}).

	\vspace{0.5cm}
	
%\anupam{for all action rules, same subtlety as above that $A$ below can be more complex.}
	$
	\vlderivation{
		\vliin{\boxP{\union}}{}{ \Gamma,  \boxP{\alpha \cup \beta}A}{
			%\vlpr{\mathcal{D}_1}{}{\Gamma, \boxP{\alpha}A}
			\vlhy{	\Gamma, \boxP{\alpha}A
			%\vltreeder{\der_1}{\Gamma, \boxP{\alpha}A}{\ }{}{\ }
			}
		}{
		%	\vlpr{\mathcal{D}_2}{}{\Delta, \boxP{\beta}A}
			\vlhy{	\Gamma, \boxP{\beta}A
%			\vltreeder{\der_2}{\Delta, \boxP{\beta}A}{\ }{}{\ }
			}
		}
	} 
	\quad \rightsquigarrow \quad 
	$
	\[
	\vlderivation{
		\vlid{}{}{\strg,  \ct{\boxP{\alpha \cup \beta}A}{\cons{c}} }{
%		\vlid{}{}{\strg, \strd, \str{ \forall y ( \nf{\st{\cons{c}, y}{ \alpha \cup \beta}} \vlor \st{y}{A} ) }{\varnothing}}{
%		\vlin{\forall}{}{\strg, \strd,\str{ \forall y ( (\nf{\st{\cons{c},y}{\alpha}} \vlan \nf{\st{\cons{c},y}{\beta}})   \vlor \st{y}{A} ) }{\varnothing}}{
%		\vlin{\vlor}{}{
%		\strg, \strd,\str{(\nf{\st{\cons{c},\cons{d}}{\alpha}} \vlan \nf{\st{\cons{c},\cons{d}}{\beta}})   \vlor \st{\cons{d}}{A}  }{\varnothing}
%		}{
		\vlin{\vlan}{}{
		\strg, \str{\nf{\st{\cons{c},\cons{d}}{\alpha}} \vlan \nf{\st{\cons{c},\cons{d}}{\beta}}}{\varnothing}, 
		\ct{A}{\cons{d}}
		%\str{\st{\cons{d}}{A}  }{\varnothing}
		}{
		\vliin{\cup}{}{
		\strg, \str{\nf{\st{\cons{c},\cons{d}}{\alpha}} , \nf{\st{\cons{c},\cons{d}}{\beta}}}{\varnothing},  
		\ct{A}{\cons{d}}
		%\str{\st{\cons{d}}{A}  }{\varnothing}
		}{
		\vlid{}{}{\strg, \str{\nf{\st{\cons{c},\cons{d}}{\alpha}} }{\varnothing},  
		\ct{A}{\cons{d}}
		%\str{ \st{\cons{d}}{A}  }{\varnothing}
		}{
		%\vlpr{t(\mathcal{D}_1)}{}{\strg,  \ct{\boxP{\alpha}A}{\cons{c}}
		\vlhy{	\strg,  \ct{\boxP{\alpha}A}{\cons{c}}
%		\vltreeder{\tderp}{ \strg,  \ct{\boxP{\alpha}A}{\cons{c}} }{\quad }{}{\quad }
		}
		}
		}{
		\vlid{}{}{\strg, \str{ \nf{\st{\cons{c},\cons{d}}{\beta}}}{\varnothing},  
			\ct{A}{\cons{d}}
		%\str{\st{\cons{d}}{A} ) }{\varnothing}
		}{
		%\vlpr{t(\mathcal{D}_2)}{}{	\strd, \ct{\boxP{\beta}A}{\cons{c}}}
		\vlhy{	\strg, \ct{\boxP{\beta}A}{\cons{c}}
%		\vltreeder{\tderp}{ \strd, \ct{\boxP{\beta}A}{\cons{c}} }{\quad }{}{\quad }
		}
		}
		}
		}
		}
%		}
%		}
%		}
	}
	\]

%	\vspace{0.5cm}

\[
\vlinf{\diaP{\cup}_i}{i \in \{0,1\}}{\Gamma, \diaP{\alpha_0 \cup \alpha_1}A }{\Gamma, \diaP{\alpha_i}A }
\, 
\leadsto 
\, 
\vlderivation{
	\vlid{=}{}{\strg, \ct{ \diaP{\alpha_0 \cup \alpha_1}A }{\cons{c}} }{
	\vlin{\vlor}{}{	\strg, \str{ \st{\cons{c}, y }{ \alpha_0} \vlor \st{\cons{c}, y}{\alpha_1} , \diat{A}{y}  }{\bvvar{A},y}
	}{
	\vlid{=}{}{\strg, \str{ \st{\cons{c}, y }{ \alpha_i},  \diat{A}{y}  }{\bvvar{B},y}
	}{
	\vlhy{	\strg, \ct{\diaP{\alpha_i}A}{\cons{c}}}
	}
	}
	}
}
\]

%		\[
%	\vlderivation{
%		\vlin{\diaP{\comp}}{}{ \Gamma, \diaP{\alpha \comp \beta}A}{
%%			\vlpr{\mathcal{D}}{}{
%				\vlhy{\Gamma, \diaP{\alpha}\diaP{\beta}A}
%%				}
%		}
%	} 
%	\quad \rightsquigarrow \quad 
%	\vlderivation{
%		\vlid{}{}{\cds{\Gamma}{\cons{c}} ,\cds{\diaP{\alpha \comp \beta}A}{\cons{c}} }{
%		\vlid{(\exists)}{}{	\cds{\Gamma}{\cons{c}}, \str{ \exists y ( \exists z (\st{\cons{c}, z}{\alpha} \vlan \st{z, y}{\alpha} ) \vlan \st{y}{A} ) }{\varnothing}}{
%		\vlid{(\exists)}{}{
%		\cds{\Gamma}{\cons{c}}, \str{  \exists z (\st{\cons{c}, z}{\alpha} \vlan \st{z, y}{\alpha} ),  \st{y}{A}  }{y} }{
%		\vlid{}{}{
%		\cds{\Gamma}{\cons{c}}, \str{ \st{\cons{c}, z}{\alpha} , \st{z, y}{\alpha} , \st{y}{A}  }{y, z} 
%		}{
%		\vlhy{ \cds{\Gamma}{\cons{c}, \cds{\diaP{\alpha}\diaP{\beta}A}{\cons{c}}} }
%					}
%				}
%			}
%		}
%	}
%	\]
%	
%	
	
%	\vspace{0.5cm}
	
	\[
	\vlderivation{
		\vlin{\boxP{\comp}}{}{ \Gamma, \boxP{\alpha \comp \beta}A}{
		%	\vlpr{\mathcal{D}}{}{	\Gamma, \boxP{\alpha}\boxP{\beta}A}
		\vlhy{		\Gamma, \boxP{\alpha}\boxP{\beta}A
%		\vltreeder{\der}{	\Gamma, \boxP{\alpha}\boxP{\beta}A}{\ }{}{\ }
		}
		}
	} 
	\quad \rightsquigarrow \quad 
	\vlderivation{
		\vlid{}{}{\strg, \ct{\boxP{\alpha \comp \beta}A}{\cons{c}} }{
%		\vlid{}{}{ \strg, \str{\forall y (\nf{\st{\cons{c},y}{\alpha \comp \beta}} \vlor  \st{y}{A})}{\varnothing}  }{
%		\vlin{\forall}{}{ \strg, \str{\forall y ( \forall z (\nf{\st{\cons{c},z}{\alpha}} \vlor \nf{\st{z,y}{\beta}} ) \vlor  \st{y}{A})}{\varnothing} }{
%		\vlin{\vlor}{}{ \strg, \str{ \forall z (\nf{\st{\cons{c},z}{\alpha}} \vlor \nf{\st{z,\cons{d}}{\beta}} ) \vlor  \st{\cons{d}}{A}}{\varnothing} }{
		\vlin{\forall}{}{ \strg, \str{ \forall z (\nf{\st{\cons{c},z}{\alpha}} \vlor \nf{\st{z,\cons{d}}{\beta}} ) }{\varnothing}, 
		\ct{A}{\cons{d}}
		%\str{\st{\cons{d}}{A}}{\varnothing} 
		}{
		\vlin{\vlor}{}{ \strg, \str{ \nf{\st{\cons{c},\cons{e}}{\alpha}} \vlor \nf{\st{\cons{e},\cons{d}}{\beta}}  }{\varnothing}, 
		\ct{A}{\cons{d}}
		%\str{\st{\cons{d}}{A}}{\varnothing} 
		}{
		\vlid{}{}{ \strg, \str{ \nf{\st{\cons{c},\cons{e}}{\alpha}} }{\varnothing},  \str{\nf{\st{\cons{e},\cons{d}}{\beta}}  }{\varnothing}, 
		\ct{A}{\cons{d}}
		%\str{\st{\cons{d}}{A}}{\varnothing}  
		}{
		%\vlpr{t(\mathcal{D})}{}{\strg, \ct{\boxP{\alpha}\boxP{\beta}A}{\cons{c}}}
		\vlhy{	 \strg, \ct{\boxP{\alpha}\boxP{\beta}A}{\cons{c}}
%		\vltreeder{\tderp}{ \strg, \ct{\boxP{\alpha}\boxP{\beta}A}{\cons{c}}  }{\quad }{}{\quad }
		}
		}
		}
		}
		}
%		}
%		}
%		}
	}
	\]
	
	\[
	\vlinf{\diaP{;}}{}{\Gamma, \diaP{\alpha ; \beta}A}{\Gamma, \diaP{\alpha}\diaP{\beta}A}
	\, 
	\leadsto
	\,
	\vlderivation{
		\vlid{=}{}{\strg,\ct{\diaP{\alpha \comp \beta}A}{\cons{c}} }{
%		\vlid{}{}{\strg, \str{ \exists y (\st{\cons{c}, y}{\alpha \comp \beta} \vlan \st{y}{A} ) }{\varnothing} }{
%		\vlin{\exists}{}{
%		\strg, \str{ \exists y ( \exists z (\st{\cons{c}, z}{\alpha} \vlan \st{z, y}{\alpha} ) \vlan \st{y}{A} ) }{\varnothing} }{
%		\vlin{\vlan}{}{
%		\strg, \str{ \exists z (\st{\cons{c}, z}{\alpha} \vlan \st{z, y}{\alpha} ) \vlan \st{y}{A}  }{y} }{
%		}{
%		\vlin{\vlan}{}{ 	\strg, \str{ \exists z (\st{\cons{c}, z}{\alpha} \vlan \st{z, y}{\alpha} ) \vlan \st{y}{A}  }{y} 
%		 }{
	 	\vlin{\exists}{}{ 	\strg, \str{ \exists z (\st{\cons{c}, z}{\alpha} \vlan \st{z, y}{\alpha} ) , \diat{A}{y}  }{\bvvar{A}, y} 
	 	 }{
 	 	\vlin{\vlan}{}{ 	\strg, \str{ \st{\cons{c}, z}{\alpha} \vlan \st{z, y}{\alpha}  , \diat{A}{y}  }{\bvvar{A}, y,z}  }{
 	 	\vlid{=}{}{	\strg, \str{ \st{\cons{c}, z}{\alpha} , \st{z, y}{\alpha}  , \diat{A}{y}  }{\bvvar{A}, y,z} }{
 	 %	\vlid{}{}{ \strg, \str{ \st{\cons{c}, z}{\alpha} , \exists y ( \st{z, y}{\alpha}  \vlan \st{y}{A})  }{z}  }{
 	 	% \vlpr{t(\mathcal{D})}{}{
 	 		\vlhy{\strg, \ct{\diaP{\alpha}\diaP{\beta}A}{\cons{c}}}   
 	 	% 	} 
  		}
 	 	%}
  		}
  		}
  		}
 		}
%		}
%		}
%		}
%		}
	\]

	\vspace{0.5cm}
	
%\begin{adjustbox}{max width= \textwidth}
    $
    	\vlderivation{
		\vlin{\diaP{+}_0}{}{\Gamma, \diaP{\plus{\alpha}}A  }{
		%	\vlpr{\mathcal{D}}{}{ \Gamma, \diaP{\alpha}A  }
		\vlhy{	\Gamma, \diaP{\alpha}A
%		\vltreeder{\der}{ \Gamma, \diaP{\alpha}A}{\ }{}{\ }
		}
		}
	} 
	\quad \rightsquigarrow \quad
	$
	\[
	\vlderivation{
		\vlid{}{}{\strg, \ct{\diaP{\plus{\alpha}}A }{\cons{c}} }{
		\vlin{\mathsf{tc}}{}{	\strg, \str{  \tc{\mathsf{ST}(\alpha)}{\cons{c}}{y} , \diat{A}{y} }{\bvvar{A}, y} }{
		\vlin{\wk}{}{\strg, \str{ \st{\cons{c}, y}{\alpha}, \diat{A}{y} }{\bvvar{A}, y}, \str{ \st{\cons{c},z}{\alpha},  \tc{\mathsf{ST}(\alpha)}{z}{y} , \diat{A}{y} }{\bvvar{A}, y,z}
		}{
		\vlid{}{}{
		\strg, \str{ \st{\cons{c}, y}{\alpha}, \diat{A}{y} }{\bvvar{A}, y}
		}{
		%\vlpr{}{t(\mathcal{D})}{ \strg, \ct{\diaP{\alpha}A}{\cons{c}} }}
		\vlhy{	\strg, \ct{\diaP{\alpha}A}{\cons{c}} 
%		\vltreeder{\tderp}{\strg, \ct{\diaP{\alpha}A}{\cons{c}}  }{\quad }{}{\quad }
		}
}	
	}
		}
		}
		}
%		}
%		}
%		}
	\]
%\end{adjustbox}

	\vspace{0.5cm}
	
%	\begin{adjustbox}{max width = \textwidth}
	$
	\vlderivation{
		\vlin{\diaP{+}_1}{}{\Gamma, \diaP{\plus{\alpha}}A  }{
		%	\vlpr{\mathcal{D}}{}{ \Gamma,  \diaP{\alpha}\diaP{\plus{\alpha}}A  }
			\vlhy{	\Gamma,  \diaP{\alpha}\diaP{\plus{\alpha}}A 
			%\vltreeder{\der}{ \Gamma,  \diaP{\alpha}\diaP{\plus{\alpha}}A  }{\ }{}{\ }
			}
		}
	} 
	\quad \rightsquigarrow \quad
	$
	\[
	\vlderivation{
		\vlid{}{}{\strg, \ct{\diaP{\plus{\alpha}}A }{\cons{c}} }{
		\vlin{\mathsf{tc}}{}{	\strg, \str{  \tc{\mathsf{ST}(\alpha)}{\cons{c}}{y} , \diat{A}{y} }{\bvvar{A}, y} }{
		\vlin{\wk}{}{	\strg, \str{ \st{\cons{c}, y}{\alpha}, \diat{A}{y} }{\bvvar{A}, y}, \str{ \st{\cons{c},z}{\alpha},  \tc{\mathsf{ST}(\alpha)}{z}{y} , \diat{A}{y} }{\bvvar{A}, y,z}}{ 
		\vlid{}{}{
		\strg,  \str{ \st{\cons{c},z}{\alpha},  \tc{\mathsf{ST}(\alpha)}{z}{y} , \diat{A}{y} }{\bvvar{A}, y,z}
		}{
		%\vlpr{t(\mathcal{D})}{}{ \strg,  \ct{\diaP{\alpha}\diaP{\plus{\alpha}}A }{\cons{c}} }
		\vlhy{	\strg,  \ct{\diaP{\alpha}\diaP{\plus{\alpha}}A }{\cons{c}}
%		\vltreeder{\tderp}{\strg,  \ct{\diaP{\alpha}\diaP{\plus{\alpha}}A }{\cons{c}} }{\quad }{}{\quad }
		}
		}
		}
		}
		}
		}
%		}
%		}
%		}
	\]
%	\end{adjustbox}
	
	\vspace{0.5cm}
	
	$
	\vliinf{\boxP{+}}{}{\Gamma, \boxP{\alpha^+ }A}{\Gamma, \boxP{\alpha}A}{\Gamma, \boxP{\alpha}\boxP{\alpha^+}A} 
	\, \leadsto \, 
	$
			\[
	%	\scriptsize
	\vlderivation{
		\vlid{=}{}{ \strg, {\ct{\boxP{\plus{\alpha}}A }{\cons{c}}} }{
		\vlin{\coTC}{}{ \strg, \blue{\str{ \cotc{\nf{\mathsf{ST}(\alpha)}}{\cons{c}}{\cons{d}} }{\varnothing}}, 
		\ct{A}{\cons{d}}
%		\str{\st{\cons{d}}{A}}{\varnothing}
		}{
		\vliin{\cup}{}{ \strg, \str{\nf{\st{\cons{c}, \cons{d}}{\alpha}}, \nf{\st{\cons{c},\cons{f}}{\alpha}} }{\varnothing}, \blue{\str{\nf{\st{\cons{c}, \cons{d}}{\alpha}}, \cotc{\nf{\mathsf{ST}(\alpha)}}{\cons{f}}{\cons{d}} }{\varnothing}}, 
		\ct{A}{\cons{d}}
		%\str{\st{\cons{d}}{A}}{\varnothing} 
		}{\vlhy{\mathcal E}	}{
		\vliin{\cup}{}{\strg,  \str{ \nf{\st{\cons{c}, \cons{f}}{\alpha}} }{\varnothing}, \blue{\str{\nf{\st{\cons{c}, \cons{d}}{\alpha}}, \cotc{\nf{\mathsf{ST}(\alpha)}}{\cons{f}}{\cons{d}} }{\varnothing}}, 
		\ct{A}{\cons{d}}
		%\str{\st{\cons{d}}{A}}{\varnothing}
		}{
		\vlhy{\mathcal E'}
		}{
		\vlid{=}{}{\strg,\str{ \nf{\st{\cons{c}, \cons{f}}{\alpha}} }{\varnothing}, \blue{\str{ \cotc{\nf{\mathsf{ST}(\alpha)}}{\cons{f}}{\cons{d}}}{\varnothing}}, 
		\ct{A}{\cons{d}}
		%\str{\st{\cons{d}}{A}}{\varnothing}
		}{
% 		\vlpr{t(\mathcal{D}_2)}{}{
		\vlhy{ 
		\strg, {\ct{\boxP{\alpha}\boxP{\plus{\alpha}}A}{\cons{c}}} }}
		}
		}
		}
		}
		}
	%	}
	\]
	where $\mathcal E,\mathcal E'$ derive %are just $\wk$-steps applied to 
	$\strg,  \ct{\boxP{\alpha}A}{\cons{c}} $ using $\wk$-steps.
	
\marianna{they're commented right below if we want to add them in}

% 	\(
% 	\scriptsize
% 	\vlderivation{
% 		\vlin{\mathsf{wk}}{}{\strg,\str{\nf{\st{\cons{c}, \cons{d}}{\alpha}} }{\varnothing}, {\str{\nf{\st{\cons{c}, \cons{d}}{\alpha}}, \cotc{\nf{\mathsf{ST}(\alpha)}}{\cons{f}}{\cons{d}} }{\varnothing}}, 
% 		\ct{A}{\cons{d}}
% 		%\str{\st{\cons{d}}{A}}{\varnothing}
% 		}{
% 		\vlid{}{}{\strg,\str{\nf{\st{\cons{c}, \cons{d}}{\alpha}} }{\varnothing}, 
% 		\ct{A}{\cons{d}}
% 		%\str{\st{\cons{d}}{A}}{\varnothing}
% 		}{
% % 		\vlpr{t(\mathcal{D}_1)}{}{
% 		\vlhy{	
% 		\strg,  \ct{\boxP{\alpha}A}{\cons{c}}  
% 		}
% 		}
% 		}
% 	}
% 	\) and \(
% 	\vlderivation{
% 		\vlin{\mathsf{wk}}{}{\strg,\str{ \nf{\st{\cons{c}, \cons{f}}{\alpha}} }{\varnothing}, \str{\nf{\st{\cons{c}, \cons{d}}{\alpha}}}{\varnothing}, 
% 		\ct{A}{\cons{d}}
% 		%\str{\st{\cons{d}}{A}}{\varnothing}
% 		}{
% 		\vlid{}{}{\strg, \str{\nf{\st{\cons{c}, \cons{d}}{\alpha}}}{\varnothing}, 
% 		%\str{\st{\cons{d}}{A}}{\varnothing}
% 		\ct{A}{\cons{d}}
% 		}{
% % 		\vlpr{t(\mathcal{D}_1)}{}{ 
% 		\vlhy{
% 		\strg, \ct{\boxP{\alpha}A}{\cons{c}} 
% 		}
% 		}
% 		}
% 	}
% 	\).

\vspace{0.5cm}

By applying the above translation to each rule of a $ \PDLtr $ cyclic proof $ \mathcal{D} $ of $ A $, we obtain a preproof $ \tder $ of $ \ct{A}{\cons{c}} $. The last step, from $\ct{A}{\cons{c}}$ to $\st{\cons{c}}{A}$, follows from \Cref{prop:ht_st_tcl}.

\end{definition}
%\anupam{everything else above seems fine.}
% \end{proof}

\begin{remark}
[Deeper inference]
Observe that
% as well as locally simulating $\lpdltr$,
$\htcl$ can also simulate `deeper' program rules than are available in $\lpdltr$.
% , that do not necessarily act on the leftmost modality.
E.g.\ a rule $\scriptsize \vlinf{}{}{\Gamma, \diaP\alpha \diaP{\beta_0 \cup \beta_1}A }{\Gamma, \diaP\alpha \diaP{\beta_i}A }$ may be simulated too (similarly for $\boxP{\, }$).
Thus $\diaP{a^+}\diaP b p \vljm \diaP{a^+}\diaP{b\cup c }p$ admits a \emph{finite} proof in $\htcl$ (under $\ST$), rather than a necessarily infinite (but cyclic) one in $\lpdltr$.  
% in such an extension of $\lpdltr$ (and thus in $\ltcl$ under the standard translation), but only an infinite (cyclic) one in $\lpdltr$.
\end{remark}

\subsection{Justifying regularity and progress}
\label{ssec:progress_reg}
\begin{proposition}
\label{prop:ht-trans-pres-reg}
    If $\der $ is regular, then so is $\htr c \der$.
\end{proposition}
\begin{proof}%[sketch]
\anupam{changed $\branch^\derd$ to $\branch'$ here. I found former notation a little confusing because we are \emph{reflecting} into $\derd$, not \emph{mapping} to $\derd$.}
\marianna{Ok, I wanted to make clear to which proof $\branch$ belongs. I found not immediately clear that trace $\tau$ belongs to branch $\branch'$ (would have expected $\tau'$ or something like that). It's a very minor thing can leave as it is now. I can change the appendix accordingly.}
Notice that each rule in $\derd $ is translated to a finite derivation in $\htr c \derd$.
% , up to choice of constant/variable symbols. 
Thus, if $ \mathcal{D} $ has only finitely many distinct subproofs, then also $\tder$ has only finitely many distinct subproofs.
%
%Regularity follows from regularity of $ \mathcal{D} $, and from the fact that the translation 
%
%
%the translation is stepwise applied to all the rules occurring in $ \mathcal{D} $. 
%\anupam{Careful here, this is where the imprecision in the definition of $\ct A$ becomes important: since it is only determined up to the choice of fresh constant symbols and bound variables. Our quantifier steps need to introduce fresh symbols bottom-up, so at a repetition we may have different function symbols/variables. This is why we need the substitution rule (for regularity). It is a little subtle.}
%
%\anupam{8/2: is my above comment addressed yet? I think a brief argument should suffice. We can say something like:
%
%each step is translated to a finite derivation, up to choice of constant/variable symbols. Since there are only fintiely many subproofs of the cyclic PDL+ proof, there are thus only finitely many distinct subproofs of the HTCL proof, up to choice of constant/variable symbols. Thus, using substitution, we have a cyclic proof.}
\end{proof}

\begin{proposition}
\label{prop:ht-trans-pres-prog}
    If $\der$ is progressing, then so is $\htr c \der$.
\end{proposition}
\begin{proof}[Proof sketch]
%\marianna{If we keep the proof as it is now, I would remove the blue color of formulas, because we're not relying on it}
%\anupam{I agree that we don't use it in the text, but it helps understanding of progress by way of the blue colouring in the $\boxP +$ case in the previous definition.}
We need to show that
every infinite branch of $ \tder$  has a progressing hypertrace. 
Since the $\HT$ translation is defined stepwise on the individual steps of $\derd$, we can associate to each infinite branch $\branch$ of $\tder$ a unique infinite branch $\branch'$ of $\der$.
Since $\der$ is progressing, let $ \trace= (F_i)_{i < \omega} $ be a progressing thread along $\branch'$. 
By inspecting the rules of $\lpdltr$ (and by definition of progressing thread), for some $k\in \Nat$, each $F_i$ for $i>k$ has the form:
$
\boxP{\alpha_{i,1}}\cdots \boxP{\alpha_{i,n_i}} \boxP{\alpha^+}A, 
$
for some $n_i \geq 0$.
So, for $i>k$, $\htr {d_i} {F_i}$ has the form:

\ 

\begin{adjustbox}{max width = \textwidth}
$
%\scriptsize
\str{ \nf{\st{\cons{c} ,\cons{d}_{i,1 }}{\alpha_{i,1} } } }{\varnothing} , \dots , \str{\nf{ \st{\cons{d}_{i,{n_i-1}},\cons{d}_{i,n_i}}{\alpha_{i,n_i} }} }{\varnothing},
\blue{\str{ \cotc{\nf{\mathsf{ST}(\alpha)}}{ \cons{d}_{i,{n_i}} }{d_i}  }{\varnothing}},
\ct{A}{\cons{d_i}}
%\]
$
\end{adjustbox}

\ 

\noindent By inspection of the $\HT$-translation (Dfn.~\ref{dfn:ht-translation}) whenever $ F_{i+1} $ is an immediate ancestor of $ F_i $ in $\branch'$, 
there is a path from the cedent $\blue{\str{ \cotc{\nf{\mathsf{ST}(\alpha)}}{ \cons{d}_{i+1,{n_{i+1}}} }{d_{i+1}}  }{\varnothing}}$ to the cedent $\blue{\str{ \cotc{\nf{\mathsf{ST}(\alpha)}}{ \cons{d}_{i,{n_i}} }{d_i}  }{\varnothing}}$ in the graph of immediate ancestry along $\branch$.
Thus, since $\tau=(F_i)_{i<\omega}$ is a trace along $\branch'$, we have a (infinite) hypertrace of the form $\htrace_\trace \defsym ( \str{\Delta_i, \cotc{\nf{\mathsf{ST}(\alpha)}}{ \cons{d}_{i,{n_i}} }{d_i}  }{\varnothing} )_{i>k'}  $
along $\branch$.
By construction $\Delta_i = \emptyset $ for infinitely many $i>k'$, and so $\htrace_\trace$ has just one infinite trace.
Moreover, by inspection of the $\boxP+$ step in Dfn.~\ref{dfn:ht-translation}, 
this trace progresses in $\branch$ every time $\tau$ does in $\branch'$, and so progresses infinitely often.
Thus, $\htrace$ is a progressing hypertrace. 
Since the choice of the branch $\branch$ of $\derd $ was arbitrary, we are done.
\end{proof}

\subsection{Putting it all together}
\label{ssec:translation_together}
We can now finally conclude our main simulation theorem:

\begin{proof}
    [Proof of Thm.~\ref{thm:completeness-htcl-pdltr}, sketch]
    Let $A$ be a $\PDLtr$ formula s.t.\ $\models A$.
    By the completeness result for $\lpdltr$, Thm.~\ref{pdl-soundness-completeness}, we have that $\lpdltr \dercyc A$, say by a cyclic proof $\derd$.
    From here we construct the $\htcl$ preproof $\htr c \derd$ which, by Props.~\ref{prop:ht-trans-pres-reg} and \ref{prop:ht-trans-pres-prog}, is in fact a cyclic proof of $\htr c A$.
    Finally, we apply some basic $\vlor, \vlan , \exists, \forall$ steps to obtain a cyclic $\htcl$ proof of $\st c A$.
%     Diagrammatically:
%     	$$
% 	\htcl \models \st{\cons{c}}{A}
% 	\,
% 	\overset{\text{Cor.}~\ref{thm:standard_t_pdltr}}{\Rightarrow} 
% 	\,  
% 	\PDLtr \models A
% 	\, 
% 	\overset{\text{Thm.}~\ref{pdltr-soundness-completeness}}{\Rightarrow} 
% 	\,  
% 	\lpdltr \dercyc A 
% 	\, 
% 	\overset{\text{Cor.}~\ref{thm:pdlplus_to_htcl}}{\Rightarrow}
% 	\,   
% 	\htcl \dercyc \st{\cons{c}}{A}.
% 	%\, 
% 	%\Rightarrow
% 	%\,   
% 	%\htcl \dercyc \str{\st{\cons{c}}{A}}{\varnothing}
% 	$$
\end{proof}
% \begin{theorem}
% 	\label{thm:pdlplus_to_htcl}
% If $ \lpdltr \dercyc A $ then $ \htcl \dercyc \st{\cons{c}}{A} $. 
% \end{theorem}

% \begin{corollary}[Completeness for $\PDLtr$]
% 	For $ A $ $ \pdl $ formula and $ \cons{c} $ constant, if $ \st{\cons{c}}{A} $ is valid in $ \TCL $ then $ 	\htcl \dercyc \str{\st{\cons{c}}{A}}{\varnothing} $. 
% \end{corollary}

% \begin{proof} The proof proceeds as follows:
% 	$$
% 	\htcl \models \st{\cons{c}}{A}
% 	\,
% 	\overset{\text{Cor.}~\ref{thm:standard_t_pdltr}}{\Rightarrow} 
% 	\,  
% 	\PDLtr \models A
% 	\, 
% 	\overset{\text{Thm.}~\ref{pdltr-soundness-completeness}}{\Rightarrow} 
% 	\,  
% 	\lpdltr \dercyc A 
% 	\, 
% 	\overset{\text{Cor.}~\ref{thm:pdlplus_to_htcl}}{\Rightarrow}
% 	\,   
% 	\htcl \dercyc \st{\cons{c}}{A}.
% 	%\, 
% 	%\Rightarrow
% 	%\,   
% 	%\htcl \dercyc \str{\st{\cons{c}}{A}}{\varnothing}
% 	$$
% %	Finally, from $ \htcl \dercyc \ct{A}{\cons{c}}  $ it is immediate to conclude that	$\htcl \dercyc \str{\st{\cons{c}}{A}}{\varnothing}$, by applying the $ \htcl $ rules $ \forall$, $\exists $, $ \vlan $ and $ \vlor $ to $ \ct{A}{\cons{c}} $. 
% \end{proof}

%%%%%%%%%%%%%%%%%%%%%%%%%%%%%%%%%%%%%%%%%%%%%%%%%%%%%%%%%%%%%%%%%%%%%%%%%%%
%%%%%%%%%%%%%%%%%%%%%%%%%%%%%%%%%%%%%%%%%%%%%%%%%%%%%%%%%%%%%%%%%%%%%%%%%%%

%%%%%%%%%%%%%%%%%%%%%%%%%%%%%%%%%%%%%%%%%%%%%%%%%%%%%%%%%%%%%%%%%%%%%%%%%%%
%%%%%%%%%%%%%%%%%%%%%%%%%%%%%%%%%%%%%%%%%%%%%%%%%%%%%%%%%%%%%%%%%%%%%%%%%%%
 \section{Extension by equality and simulating full $\PDL$}
 \label{sec:pdl}
We now briefly explain how our main results are extended to the `reflexive' version of $\TCL$. 

\subsection{Hypersequential system with equality}
The language of $\hrtcl$ allows further atomic formulas of the form $s=t$ and $s\neq t$.
The calculus $\hrtcl$ extends $\htcl$ by the rules:
\begin{equation}
\label{eq:eq_rules_hrtcl}
%\footnotesize
\vlinf{\eqrule}{
				% \footnotesize{
				% 	\begin{matrix}
				% 	\freev{t}\cap \bv = \varnothing
				% 	\end{matrix}
				% }
				}{\violet \sq, \str{t=t, \magenta \Gamma}{\bv}}{\violet \sq, \str{\magenta \Gamma}{\bv}}
				\qquad 
				\vlinf{\coeqrule}{}{\violet \sq, \str{\magenta {\Gamma(s)}, \blue {s \neq t}}{\bv}, \str{\Delta(t)}{\bv} }{\violet \sq, \str{\magenta{\Gamma(s)}, \blue{\Delta(s)}}{\bv}}
\end{equation}

The notion of immediate ancestry for formulas and cedents is colour-coded in \eqref{eq:eq_rules_hrtcl} just as we did for $\htcl$ in \Cref{sec:htcl_proofs}.
The resulting notions of (pre)proof, (hyper)trace and progress are as in Dfn.~\ref{def:progress_tcl}. 
Specifically, we have that in rule $ \eqrule $ as typeset in \eqref{eq:eq_rules_hrtcl}, no infinite trace can include formula $t = t$. Moreover, in rule $ \coeqrule $ no infinite trace can include formulas in $ \Delta(t) $, while \emph{all} formulas occurring in $ \Delta(s) $ belong to a trace where $ s \neq t $ belongs.

The simulation of Cohen and Rowe's system $\ltcl$ extends to their reflexive system, $\lrtcl$, by the definition of their operator $\rtc {\lambda x,y.A} st$ $ \defsym$ $ \tc{\lambda x,y.(x=y \vlor A)}st$.
Semantically, it is correct to set $\rtc Ast$ to be $s=t \vlor \tc Ast$, but this encoding does not lift to the Cohen-Rowe rules for $\RTC$.

\subsection{Extending the soundness argument}

Understanding that structures interpret $=$ as true equality, a modular adaptation of the soundness argument for $\htcl$, cf.~Sec.~\ref{sec:soundness_htcl}, yields:

\begin{theorem}[Soundness of $\hrtcl$]
	\label{thm:soundness_rtcl}
	If $\hrtcl \dernwf \sq $ then $  \models  \sq $.
\end{theorem}

\begin{proof}[Proof sketch]
In the soundness argument for $\htcl$, in Lem.~\ref{lem:countermodel-branch}, we must further consider cases for equality as follows:
    
    \ 
    
\noindent \textbf{$\triangleright$  Case $ (\eqrule ) $}
$$
\vlinf{\eqrule}{}{\sq_{i} \, = \, \qq, \str{t=t, \Gamma}{\bv}}{\sq_{i+1} \, = \, \qq, \str{\Gamma}{\bv}}
$$
By assumption, $ \rho_{i} \not\models \qq $ and $ \rho_{i} \not \models \str{\Gamma, t = t}{\bv} $, i.e., $ \rho_{i} \models \forall \bv (\bigvee \nf{\Gamma} \vlor t \neq t)$. Since $ \freev{t} \cap \bv = \varnothing $, this is equivalent to  $ \rho_{i} \models \forall \bv (\bigvee \nf{\Gamma})$. We set $ \rho_{i+1} = \rho_{i} $ and conclude that $ \rho_{i+1} \not \models \str{\Gamma}{\bv} $. 

\ 

\noindent \textbf{$\triangleright$  Case $ (\neq)$}
$$
\vlinf{\neq}{}{\sq_{i} \,  = \, \qq, \str{\Gamma(s), s \neq t}{\bv}, \str{\Delta(t)}{\bv} }{\sq \, = \, \qq, \str{\Gamma(s), \Delta(s)}{\bv}}
$$
By assumption, $ \rho_{i} \not\models \qq $ and $ \rho_{i} \not\models \str{\Gamma(s), s \neq t}{\bv} $ and $ \rho_{i} \not \models \str{\Delta(t)}{\bv}$. Thus, 
$$ \rho_{i} \models \forall \bv (\bigvee \nf{\Gamma}(s) \vlor s = t) 
\quad 
\text{ and }
\quad 
\rho_{i} \models \forall \bv (\bigvee\nf{\Delta}(t)) $$ 
Set $ \rho_{i+1} = \rho_{i} $. If $ \rho_{i} \models \forall \bv (\bigvee \nf{\Gamma}(s)) $, then $ \rho_{i+1} \models  \forall \bv (\bigvee \nf{\Gamma}(s) \vlor \bigvee \nf{\Delta}(s))   $, and thus $ \rho_{i+1} \not \models \str{\Gamma(s) \vlor \Delta(s)}{\bv} $. 
Otherwise, if $ \rho_{i} \models \forall \bv (s = t) $, we can safely substitute term $ t $ with term $ s $ in the second conjunct, obtaining $ \rho_{i+1} \models \forall \bv (\bigvee \nf{\Delta} (s)) $. Thus, $ \rho_{i+1} \models  \forall \bv (\bigvee \nf{\Gamma}(s) \vlor \bigvee \nf{\Delta}(s))   $, and $ \rho_{i+1} \not \models \str{\Gamma(s) \vlor \Delta(s)}{\bv} $.

For the construction of the `false trace' in Lem.~\ref{lemma:failed_trace}
we add the following cases for equality:

\ 

	\noindent \textbf{$\triangleright$  Case $ (\eqrule) $}.
	Suppose $ \str{\Gamma_i}{\bv_i} = \str{t = t, \Gamma}{\bv}$. By assumption, $ \fint \models \forall \bv (t \neq t \vlor \bigvee \nf{\Gamma}) $ and thus  $ \fint, \delta_\hyp \models t \neq t$ or $\fint, \delta_\hyp \models  \bigvee \nf{\Gamma} $. 
	Since the first disjunct cannot hold, the trace cannot follow the formula $ t =t $. Thus, if $ F = C $ for some $ C $ occurring in $ \Gamma $ and  $ \fint, \delta_\hyp \models \nf{C} $, set $ F' = F $ and conclude that  $ \fint, \delta_\hyp \models \nf{C} $. 
	
	\ 
	
	\noindent \textbf{$\triangleright$  Case $ (\coeqrule) $}.
	The hypertrace $ \hyp $ cannot go through the structure $ \str{\Delta(t)}{\bv} $, because by hypothesis $ \hyp $ is infinite. 
	%This rule sees two hypertrace which interact. Let us call $ \hypL $ the hypertrace going through the structures $ \str{\Gamma(s), s \neq t}{\bv} $ and $ \str{\Gamma(s), \Delta(s)}{\bv} $, and $ \hypR $ the hypertrace going through the structures $ \str{\Delta(t)}{\bv} $ and $ \str{\Gamma(s), \Delta(s)}{\bv} $. Thus, $ \delta_{\hypL} $ denotes the assignment $ \fv{\hypL} \to D $, and  $ \delta_{\hypR} $ denotes the assignment $ \fv{\hypR} \to D $. Note that  $ \delta_{\hypL} $ and  $ \delta_{\hypR} $ assign the same value to the variables in $ \bv $, as their value is set by occurrences of $ \mathsf{inst} $ which occur in the structure followed by both hypertraces. We distinguish cases according to the chosen hypertrace. 
	Thus, $ \str{\Gamma_{i}}{\bv_{i}} = \str{\Gamma(s), s \neq t}{\bv} $. 
	By assumption, we have that either:
		$$
		\fint, \delta_{\hyp} \models \bigvee \nf{\Gamma(s)} 
		\quad 
		\text{ or }
		\quad 
		\fint, \delta_{\hyp} \models s = t
		$$
%		By Lemma~\ref{lemma:failed_trace_rtcl}, $ \fint_{i+1} = \fint_{i} $. Fix an arbitrary $ \delta_{\hyp} $.
	If $ F = C(s) $ for some $ C(s) $ in $ \Gamma(s) $ and  $ \fint, \delta_{\hyp} \models \nf{F} $, set $ F' = F $. Otherwise, if $\fint, \delta_{\hyp} \not \models \bigvee \nf{\Gamma(s)} $, then  $\fint, \delta_{\hyp} \models s = t$ and $ F  $ is $ s\neq t $. 
	%and $\fint, \delta_{\hyp} \models s = t$, we need to find a $ G \in \Gamma(s) \cup \Delta(s) $ which belongs to the same trace to which $ F $ belongs and such that $ 	\fint, \delta_{\hyp} \models \nf G $. 
	At every occurrence of rule $ \coeqrule $ all the formulas occurring in $\Delta(s) $ are immediate ancestors of  $s\neq t $. %If there is a formula $ C $ in $ \Gamma(s) $ such that  $ \fint, \delta_{\hyp} \models \nf{C} $, set $ F' = C $. 
	%Otherwise, $\fint, \delta_{\hyp} \models s = t$. Let us consider $ \Delta(s) $. 	
	Moreover, by assumption $ \fint \models \forall \bv ( \bigvee \nf{\Delta}(t)) $. By the truth condition associated to $ \forall $ and since $ \bv $ is contained in the domain of $ \delta_{\hyp} $, we have that  $ \fint, \delta_\hyp \models \bigvee \Delta(t) $. Since  $\fint, \delta_{\hyp} \models s = t$, we conclude that $ \fint, \delta_\hyp \models \bigvee \Delta(s) $. Thus, there exists a $ D(s) \in \Delta(s) $ such that $ \fint, \delta_\hyp \models \bigvee D(s)$. Set $ F' = D(s) $. 
\end{proof}

% \begin{proof}[sketch]
%     The proof proceeds similarly as in the case of $\htcl$, by constructing a countermodel branch $\fbr$ in a derivation $\der$ of $\sq$ and by selecting a false trace $\ftr$ for a given infinite hypertrace through $\fbr$. 
%     \marianna{can write more if there's more space}
% \end{proof}

\subsection{Completeness for $\PDL$ (with tests)}

Turning to the modal setting, $\PDL$ may be defined as the extension of $\PDLtr$ by including a program $A?$ for each formula $A$.
Semantically, we have $\interp{\mathcal M}{A?} = \{ (v,v) : \mathcal M,v\models A \}$.
From here we may define $\epsilon \defsym \top?$ and $\alpha^* \defsym (\epsilon \cup \alpha)^+$.
Again, while it is semantically correct to set $\alpha^* = \epsilon \cup \alpha^+$, this encoding does not lift to the standard sequent rules for $*$.

The system $\lpdl$ is obtained from $\lpdltr$ by including the rules:
%\marianna{two premiss rule should be context-joining?}
%\anupam{didn't we agree that for IJCAR we would go for context-sharing? That is how I've done it everywhere else, I believe.}
%\marianna{Ok, I didn't remember}
\[
%\footnotesize
% \scriptsize
\vliinf{\diaP ?}{}{\magenta \Gamma, \blue{\diaP{A?}B}}{\magenta\Gamma, \blue A}{\magenta \Gamma, \blue B}
\qquad
% \) and \(
\vlinf{\boxP ?}{}{\magenta \Gamma , \blue{\boxP{A?}B}}{\magenta \Gamma, \blue{\bar A}, \blue B}
\]
The notion of ancestry for formulas is defined as for $\lpdltr$ (Dfn.~\ref{def:ancestry_fml_pdltr}) and colour-coded in the rules. 
%Again, the notion of immediate ancestry is colour-coded as for $\lpdltr$; 
The resulting notions of (pre)proof, thread and progress are as in Dfn.~\ref{dfn:pdl-threads-proofs}. 
We write $\lpdl \dercyc A$ meaning that there is a cyclic proof of $A$ in $\lpdl$. 
Just like for $\lpdltr$, a standard encoding of $\lpdl$ into the $\mu$-calculus yields its soundness and completeness, thanks to known sequent systems for the latter, cf.~\cite{studer2008proof,NIWINSKI1996,lange2003games}. %\anupam{mention appendix here.}
\marianna{added ref to Lange}

\begin{theorem}[Soundness and completeness, \cite{lange2003games}]
    Let $A$ be a $\PDL$ formula. $\models A$ iff $\lpdl \dercyc A $.
\end{theorem}

Again, a modular adaptation of the simulation of $\lpdltr$ by $\htcl$, cf.~Sec.~\ref{sec:htcl_simulates_pdltr}, yields:
\begin{theorem}
    [Completeness for $\PDL$]
    \label{thm:completeness-hrtcl-pdl}
    Let $A$ be a $\PDL$ formula. If $\models A$ then $\hrtcl \dercyc \st c A$.
\end{theorem}

% The programs of $\pdl$ are defined as follows:
% \begin{eqnarray*}
% 	\alpha, \beta &\defl & a \mid \alpha \comp \beta \mid \alpha \union \beta 
% 	\mid ? A \mid \alpha^* 
% \end{eqnarray*}
% The standard translation (Dfn.~\ref{dfn:stand-trans-pdl-rtcl}) is extended to tests ($?A$) and $*$ programs:  $ \st{x,y}{A?} \defsym x=y \vlan \st x A$ and $ \st{x,y}{\alpha^*} \defsym \rtc{\ST(\alpha)} x y $.

\begin{proof}[Proof sketch]
For the Simulation Theorem, Thm.~\ref{thm:completeness-hrtcl-pdl}, we must add the following cases for the test rules:

	\ 
	
	\begin{adjustbox}{max width = \textwidth}
	$
	\vlderivation{
		\vliin{\diaP{?}}{}{\Gamma, \Delta, \diaP{\test{A}}B  }{
		%\vlpr{\mathcal{D}_1}{}{ 
		\vlhy{\Gamma, A  }
		%}
		}{
		%\vlpr{\mathcal{D}_2}{}{ 
		\vlhy{\Delta, B}  
		%}
		}
	} 
	\quad \rightsquigarrow \quad 
	\vlderivation{
	\vlid{}{}{\strg, \strd, \ct{\diaP{\test{A}}B }{\cons{c}} }{
	%\vlid{}{}{\strg, \strd, \str{ \st{\cons{c},y}{\test{A}}, \st{y}{B} }{y} }{
	\vlin{\vlan}{}{\strg,\strd, \str{ \cons{c}= y \vlan \st{\cons{c}}{A} ,\diat{B}{y} }{\bvvar{B}, y} }{
	\vliin{\cup}{}{ \strg, \strd, \str{ \cons{c}= y , \st{\cons{c}}{A} ,\diat{B}{y} }{\bvvar{B}, y} }{
	\vliq{\doubleline}{}{ \strg, \str{  \st{\cons{c}}{A}  }{\varnothing} }{
	%\vlpr{t(\mathcal{D}_1)}{}{
	\vlhy{\strg, \ct{A}{\cons{c}}} 
	%}
	}
	}{
	\vlin{\instrule}{}{ \strd, \str{ \cons{c}= y ,\diat{y}{B} }{\bvvar{B}, y} }{
	\vlin{=}{}{\strd, \str{ \cons{c}= \cons{c} ,\diat{B}{x}}{\bvvar{B}} }{
	\vliq{\doubleline}{}{\strd, \str{ \diat{B}{\cons{c}}}{\bvvar{B}} }{
	%\vlpr{t(\mathcal{D}_2)}{}{
	\vlhy{\strd,\ct{B}{\cons{c}}}
	%}
	}
	}
	}
	}
	}
	}
	}
%	}
	$
	\end{adjustbox}

\[
\vlderivation{
	\vlin{\boxP{?}_0}{}{\Gamma, \boxP{\test{A}}B  }{
		%\vlpr{\mathcal{D}}{}{
		\vlhy{\Gamma, \nf{A} }
		%}
	}
} 
\quad \rightsquigarrow \quad 
\vlderivation{
	\vlid{}{}{\strg, \ct{\boxP{\test{A}}B}{\cons{c}} }{
	\vlin{\vlor}{}{\strg, \str{ \cons{c} \neq \cons{d} \vlor \nf{\st{\cons{c}}{A}} }{\varnothing},
	\ct{B}{\cons{d}}
	%\str{\st{\cons{d}}{B}}{\varnothing}
	}{
	\vlin{\wk}{}{\strg, \str{ \nf{\st{\cons{c}}{A}} }{\varnothing},
	\ct{B}{\cons{d}}
	%\str{\st{\cons{d}}{B}}{\varnothing}  
	}{
				%\vlin{\wk}{}{\strg, \str{ \nf{\st{\cons{c}}{A}} }{\varnothing},\str{\st{\cons{c}}{B}}{\varnothing}}{
					\vliq{\doubleline}{}{\strg, \str{ \nf{\st{\cons{c}}{A}} }{\varnothing}}{
						%\vlpr{t(\mathcal{D})}{}{
						\vlhy{\strg, \ct{\nf A}{\cons{c}} }
						%}
					}
				}
			}
		}
	}
%}
%}
\]

\vspace{0.5cm}

\[
\vlderivation{
	\vlin{\boxP{?}_1}{}{\Gamma, \boxP{\test{A}}B  }{
		%\vlpr{\mathcal{D}}{}{ 
		\vlhy{\Gamma,  B }
		%}
	}
} 
\, \rightsquigarrow \, 
\vlderivation{
	\vlid{}{}{\strg, \ct{\boxP{\test{A}}B}{\cons{c}} }{
	\vlin{\vlor}{}{\strg, \str{ \cons{c} \neq \cons{d} \vlor \nf{\st{\cons{c}}{A}} }{\varnothing}, \ct{B}{\cons{d}} }{
	\vlin{\neq}{}{\strg, \str{ \cons{c} \neq \cons{d}}{\varnothing}, \ct{B}{\cons{d}}  }{
%	\vlin{\wk}{}{\strg, \ct{B}{\cons{c}}}{
%	\vliq{\doubleline}{}{\strg,\str{\st{\cons{d}}{B}}{\varnothing}}{
	%\vlpr{t(\mathcal{D})}{}{
	\vlhy{\strg, \ct{B}{\cons{c}} }
	%}
	}
	}
	}
	}
%	}
%	}
\]

\end{proof}

% %%%%%%%%%%%%%%%%%%%%%%%%%%%%%%%%%%%%%%%%%%%%%%%%%%%%%%%%%%%%%%%%%%%%%%%%%%%
% %%%%%%%%%%%%%%%%%%%%%%%%%%%%%%%%%%%%%%%%%%%%%%%%%%%%%%%%%%%%%%%%%%%%%%%%%%%
% \vspace{-1em}
\section{Conclusions}
\label{sec:conclusions}
%%%%%%%%%%%%%%%%%%%%%%%%%%%%%%%%%%%%%%%%%%%%%%%%%%%%%%%%%%%%%%%%%%%%%%%%%%%
%%%%%%%%%%%%%%%%%%%%%%%%%%%%%%%%%%%%%%%%%%%%%%%%%%%%%%%%%%%%%%%%%%%%%%%%%%%
\anupam{vspace hack here}
In this work we proposed a novel cyclic system $\htcl$ for Transitive Closure Logic ($\TCL$) based on a form of hypersequents.
We showed a soundness theorem for standard semantics, requiring an argument bespoke to our hypersequents.
Our system is cut-free, rendering it suitable for automated reasoning via proof search.
We showcased its expressivity by demonstrating completeness for $\PDL$, over the standard translation.
In particular, we demonstrated formally that such expressivity is not available in the previously proposed system $\ltcl$ of Cohen and Rowe (Thm.~\ref{thm:incompleteness}).
Our system $\htcl$ locally simulates $\ltcl$ too (Thm.~\ref{simulation-of-cohen-rowe}).
% We have extended our systems to include equality too, resulting in \emph{reflexive} $\TCL$.
% The notion of ancestry is a little more delicate, hence why we stuck to the equality-free presentation here, but similar results go through nonetheless.
% In particular, the resulting cut-free cyclic system is complete for $\PDL$ over the standard translation.

As far as we know, $\htcl$ is the first cyclic system employing a form of \emph{deep inference} resembling \emph{alternation} in automata theory, e.g.\ wrt.\ proof checking, cf.~Prop.~\ref{prop:checking-progress-htcl-automata}.
It would be interesting to investigate the structural proof theory that emerges from our notion of hypersequent.
As hinted at in Exs.~\ref{fixed-point-identity} and \ref{ex:transitivity}, $\htcl$ exhibits more liberal rule permutations than usual sequents, so we expect their \emph{focussing} and \emph{cut-elimination} behaviours to similarly be richer, cf.~\cite{miller2015focused,MarinMillerVolpe16:emulating-modal-proof-systems}.

Finally, our work bridges the cyclic proof theories of (identity-free) $\PDL$ and (reflexive) $\TCL$.
With increasing interest in both modal and predicate cyclic proof theory, it would be interesting to further develop such correspondences.

%\subsubsection{Acknowledgements} Please place your acknowledgements at the end of the paper, preceded by an unnumbered run-in heading (i.e. 3rd-level heading).


\bibliographystyle{plain}
\bibliography{das_girlando_cyclic2022.bib}
 
%\anupam{I inlined specific cases for the $\hrtcl$ proof in this section.}
 
\end{document}